\documentclass[a4paper]{article}
\usepackage{a4wide}
\usepackage{enumerate}
\usepackage{url}
\usepackage{hyperref}
\usepackage{graphicx}
\usepackage{bigfoot}
\usepackage{tikz}

\usepackage{amsmath,amsthm,amssymb,stmaryrd}
\usepackage{mathrsfs}
\usepackage{xfrac}

\newtheorem{theorem}{Theorem}
\newtheorem{remark}[theorem]{Remark}

\newtheorem{corollary}[theorem]{Corollary}
\newtheorem{example}[theorem]{Example}

\newtheorem{definition}[theorem]{Definition}

\newtheorem{lemma}[theorem]{Lemma}

\usepackage{microtype}
\usepackage{complexity}

\newif \ifconf \conftrue
\newif\iffinal\finalfalse

\usepackage{booktabs}   
                        
\usepackage{subcaption} 

\usepackage{todonotes}

\usepackage{amsmath}
\usepackage{dsfont}

\usepackage{ textcomp } 

\usepackage{stmaryrd}
\usepackage{MnSymbol}
\usepackage{graphicx}
\usepackage{wrapfig}

\makeatletter

\renewcommand{\poly}{\mathrm{poly}}

\newcommand*{\da@rightarrow}{\mathchar"0\hexnumber@\symAMSa 4B }
\newcommand*{\da@leftarrow}{\mathchar"0\hexnumber@\symAMSa 4C }
\newcommand*{\xdashrightarrow}[2][]{%
  \mathrel{%
    \mathpalette{\da@xarrow{#1}{#2}{}\da@rightarrow{\,}{}}{}%
  }%
}
\newcommand{\xdashleftarrow}[2][]{%
  \mathrel{%
    \mathpalette{\da@xarrow{#1}{#2}\da@leftarrow{}{}{\,}}{}%
  }%
}
\newcommand*{\da@xarrow}[7]{%
  \sbox0{$\ifx#7\scriptstyle\scriptscriptstyle\else\scriptstyle\fi#5#1#6\m@th$}%
  \sbox2{$\ifx#7\scriptstyle\scriptscriptstyle\else\scriptstyle\fi#5#2#6\m@th$}%
  \sbox4{$#7\dabar@\m@th$}%
  \dimen@=\wd0 %
  \ifdim\wd2 >\dimen@
    \dimen@=\wd2 %
  \fi
  \count@=2 %
  \def\da@bars{\dabar@\dabar@}%
  \@whiledim\count@\wd4<\dimen@\do{%
    \advance\count@\@ne
    \expandafter\def\expandafter\da@bars\expandafter{%
      \da@bars
      \dabar@ 
    }%
  }%
  \mathrel{#3}%
  \mathrel{%
    \mathop{\da@bars}\limits
    \ifx\\#1\\%
    \else
      _{\copy0}%
    \fi
    \ifx\\#2\\%
    \else
      ^{\copy2}%
    \fi
  }%
  \mathrel{#4}%
}

\newcommand{\semi}[1]{\xdashrightarrow{#1}}

\newcommand{\problemx}[3]{
	\vspace{0.2cm}
\par\noindent\underline{\sc#1}\par\nobreak\vskip.2\baselineskip
\begingroup\clubpenalty10000\widowpenalty10000
\setbox0\hbox{\bf INPUT:\ }\setbox1\hbox{\bf QUESTION:\ }
\dimen0=\wd0\ifnum\wd1>\dimen0\dimen0=\wd1\fi
\vskip-\parskip\noindent
\hbox to\dimen0{\box0\hfil}\hangindent\dimen0\hangafter1\ignorespaces#2\par
\vskip-\parskip\noindent
\hbox to\dimen0{\box1\hfil}\hangindent\dimen0\hangafter1\ignorespaces#3\par
\endgroup
	\vspace{-0.2cm}
}

\newcounter{claimcounter}
\newtheorem{subclaim}{Subclaim}{}
\newtheorem{claim}{Claim}{}

\newcommand\sg[1]{\todo[inline,size=\scriptsize]{#1 - \textbf{Stefan}}}
\newcommand\mh[1]{\todo[inline,size=\scriptsize]{#1 - \textbf{Mathieu}}}

\renewcommand{\A}{\mathcal{A}}
\newcommand{\B}{\mathcal{B}}
\renewcommand{\C}{\mathcal{C}}

\renewcommand{\G}{\mathcal{G}}

\renewcommand{\R}{\mathbb{R}}
\newcommand{\Z}{\mathbb{Z}}
\newcommand{\N}{\mathbb{N}}

\newcommand{\Const}{\mathsf{Consts}}

\newcommand{\macro}{Z}
\newcommand{\Op}{\mathsf{Op}}

\iffinal 
\renewcommand\sg[1]{}
\renewcommand\mh[1]{}

\fi

\newif\ifdraft\drafttrue

\title{Reachability in two-parametric timed automata with one parameter is EXPSPACE-complete}

\author{Stefan G\"oller}

\author{Stefan G\"oller\thanks{University of Kassel\newline
School of Electrical Engineering and Computer Science\newline
34121 Kassel, Germany\newline
\url{stefan.goeller@uni-kassel.de}\newline
Stefan G\"oller was supported by the Agence nationale de la recherche
grant no. ANR-17-CE40-0010.} 
\and Mathieu Hilaire\thanks{
	Laboratoire Spécification et Vérification (LSV)\newline
ENS Paris-Saclay\newline
CNRS\newline
Université Paris-Saclay\newline
91190 Gif-sur-Yvette, France\newline
\url{mathieu.hilaire@lsv.fr}}
}

\begin{document}

\maketitle

\begin{abstract}
	Parametric timed automata (PTA) have been introduced by Alur,
	Henzinger, and Vardi as an extension of timed automata in
	which clocks can be compared against parameters. 
	The reachability problem asks for the existence of an assignment of the parameters to
	the non-negative integers such that reachability holds
	in the underlying timed automaton.
	The reachability problem for PTA is long known to be undecidable,
	already over three parametric clocks.

	A few years ago, Bundala and Ouaknine proved that for PTA over two parametric clocks 
	and one parameter the 
	reachability problem is decidable and also showed a lower bound for the complexity class
	$\mathsf{PSPACE}^{\mathsf{NEXP}}$. 
	Our main result is that the reachability problem for two-parametric timed automata with one
	parameter is $\mathsf{EXPSPACE}$-complete.

	Our contribution is two-fold.

	For the $\EXPSPACE$ lower bound, inspired by \cite{GHOW10,GollerL13},
	we make use of deep results from complexity theory,
	namely a serializability characterization of $\EXPSPACE$ (in turn
	based on Barrington's Theorem) 
	and a logspace translation of numbers in chinese
	remainder representation to binary representation due to Chiu, Davida, and Litow. 
	It is shown that with small PTA over two parametric clocks
	and one parameter one can simulate serializability computations.

	For the $\EXPSPACE$ upper bound, we first give a careful
	exponential time reduction from PTA over two parametric clocks 
	and one parameter to a (slight subclass of) parametric one-counter automata over 
	one parameter
	based on a minor adjustment of a construction due to Bundala and Ouaknine.
	For solving the reachability problem for parametric one-counter automata with one parameter,
	we provide a series of techniques to partition a fictitious run into several 
	carefully chosen subruns 
	that allow us to prove that it is sufficient to consider a parameter value of 
	exponential magnitude only. This allows us to show a doubly-exponential upper
	bound on the value of the only parameter of a PTA over two parametric clocks
	and one parameter.
	We hope that extensions of our techniques lead to finally establishing decidability 
of the long-standing open problem of reachability in 
parametric timed automata with two parametric clocks
(and arbitrarily many parameters) and, if decidability holds, 
	determinining its precise computational complexity.
\end{abstract}

\medskip

\noindent
{\em Acknowledgments.}
We thank Benedikt Bollig and Karin Quaas for discussions and feedback.

\newpage
\tableofcontents
\newpage

\section{Introduction}

\noindent
{\bf Background.} In the 1990's {\em timed automata} have been introduced by Alur and Dill \cite{alur1994theory}.
They extend finite automata by clocks that can be compared against integer constants
and provide a popular formalism to reason about the behavior of real-time systems
with desirable algorithmic properties; for intance the reachability/emptiness
problem is decidable and in fact $\PSPACE$-complete~\cite{AlurCD90}.

For a more general means to specify the behavior of under-specified systems, such as 
embedded systems, Alur, Henzinger and Vardi~\cite{AHV93-stoc} have introduced 
{\em parametric timed automata (PTA)} only a few years later. 
Here, the clocks can additionally be compared
against parameters that can take unspecified non-negative integer values.
Towards the verification of safety properties, or loosely speaking ruling out the existence of 
an execution to a bad state, the {\em reachability problem for PTA} in turn asks for 
the existence of an assignment of the parameters to the
non-negative integers such that reachability holds in the resulting timed automaton.

A clock of a PTA that is being compared to at least one parameter is called {\em parametric}.
On the negative side, it has been shown in~\cite{AHV93-stoc} 
that already for PTA that contain {\em three parametric clocks} reachability is
undecidable --- even in the presence of one parameter~\cite{BenesBLS15}.
On the positive side however,  Alur, Henzinger and Vardi have shown
 in~\cite{AHV93-stoc} that reachability
is decidable for PTA that contain only {\em one parametric clock}, yet by an algorithm whose running
time is non-elementary.

Reachability in PTA with two or less parametric clocks has 
not attracted much attention for many years, up until recently.

For PTA over {\em one parametric clock},
Bundala and Ouaknine have shown a first elementary complexity upper bound
for the reachability problem; it is shown to be
$\mathsf{NEXP}$-hard and in $2\mathsf{NEXP}$~\cite{BundalaO17}. 
The matching $\mathsf{NEXP}$ upper bound has been proven by Bene\v{s} et al. in~\cite{BenesBLS15}
(also in the continuous time setting),
we refer to~\cite{BQS-lmcs19} for an alternative proof by
Bollig, Quaas and Sangnier using alternating 2-way automata.

Bundala and Ouaknine~\cite{BundalaO17} have recently advanced the decidability and complexity status
of the reachability problem for PTA over {\em two parametric clocks}~\cite{BundalaO17}:
it is shown that in presence
of {\em one parameter} the reachability problem is decidable and hard for
the complexity class $\mathsf{PSPACE}^{\mathsf{NEXP}}$.
To the best of our knowledge, this  is in fact the
largest subclass of PTA for which reachability
is known to be decidable.
For showing the above-mentioned decidability result~\cite{BundalaO17} provides a reduction 
from PTA over two parametric clocks
to a suitable formalism of {\em parametric one-counter automata}.
Such an approach via parametric one-counter automata has already successfully been applied to model
checking freeze-LTL as shown by Demri and Sangnier \cite{DemriS10}
and Lechner et al. \cite{LechnerMOPW18},
yet notably over a weaker model of parametric one-counter automata than the one
 introduced in \cite{BundalaO17}.
On this note, it is worth mentioning that inter-reductions between the reachability problem
of (non-parametric) timed automata involving two clocks and one-counter automata
have already been established by Haase et al.~\cite{haase2012complexity,HaaseOW16}.

Decidability of reachability in PTA over two parametric clocks
(without parameter restrictions)
is still considered to be a challenging open
problem to the best of our knowledge.
For instance, as already remarked 
in~\cite{AHV93-stoc}, there is an easy reduction
from the existential fragment of 
Presburger Arithmetic with divisibility to 
reachability in PTA over two parametric clocks.

\medskip

\noindent
{\bf Our contribution.} Our main result (Theorem~\ref{main result}) states that 
reachability in parametric timed automata over two parametric clocks
and one parameter is $\EXPSPACE$-complete. Our contribution is two-fold.

	Inspired by~\cite{GHOW10,GollerL13}, for the $\EXPSPACE$ lower bound we 
make use of deep results from complexity theory,
	namely a serializability characterization of $\EXPSPACE$ (in turn
originally based on Barrington's Theorem~\cite{Bar89}) 
	and a logspace translation of numbers in chinese
	remainder representation to binary representation due to Chiu, Davida, and 
	Litow~\cite{ChDaLi01}. 
	It is shown that with small PTA over two parametric clocks
	and one parameter one can simulate serializability computations.

	For the $\EXPSPACE$ upper bound, we first give a careful
	exponential time reduction from PTA over two parametric clocks and one parameter to 
	a (slight subclass of) parametric one-counter automata over one parameter
	based on a minor adjustment of a construction due to Bundala and Ouaknine~\cite{BundalaO17}.
	In solving the reachability problem for parametric one-counter automata with one parameter,
	we provide a series of techniques to partition a fictitious run into several 
	carefully chosen subruns 
	that allow us to prove that it is sufficient to consider a parameter value of 
	exponential magnitude. This allows us to show a doubly-exponential upper
	bound on the value of the only parameter of PTA with two parametric clocks
	and one parameter.
We hope that extensions of our techniques lead to finally establishing decidability 
of the long-standing open problem of reachability in 
parametric timed automata with two parametric clocks
(and arbitrarily many parameters) and, if decidability holds, 
determinining its precise computational complexity.

As the results in \cite{AHV93-stoc}, our results hold for PTA over discrete time.
Indeed,  for PTA with closed (i.e., non-strict) clock constraints and parameters 
ranging over integers, techniques~\cite{HenzingerMP92,OuaknineW03} exist that
allow to reduce the
 reachability problem over continuous time to discrete time.
There is a plethora of variants of PTA that have recently been studied, 
we refer to \cite{Andre19} for an extensive overview by André.

\medskip

\noindent
{\bf Overview of this paper.} 
In Section~\ref{defs section} we introduce general notations and state our main result.
Our $\EXPSPACE$ lower bound can be found in Section~\ref{lower section}.
Section~\ref{ptatopoca section} introduces parametric one-counter automata
and states an exponential time reduction from 
reachability in parametric timed automata with two parametric clocks and one parameter
to reachability in parametric one-counter automata.
Section~\ref{upper section} states the $\EXPSPACE$ upper bound and the central
Small Parameter Theorem (Theorem~\ref{theorem upper}). 
It also gives an overview of the proof of the Small Parameter Theorem, which itself
stretches over 
Sections~\ref{semirun section},\ref{hill section},\ref{5-6 section}, and \ref{application section}.

\section{Preliminaries}\label{defs section}

\newcommand{\range}{\mathsf{range}}
\newcommand{\LCM}{\mathsf{LCM}}
\newcommand{\LOGSPACE}{\mathsf{LOGSPACE}}

We assume the reader is familiar with 
Turing machines and standard complexity
classes such as $\LOGSPACE$, $\PSPACE$ and
$\EXPSPACE$. We refer to \cite{Papa94,AroBar09}
for further details on complexity.
We also assume the reader is familiar with (deterministic) finite automata and regular languages,
we refer to \cite{Har78} for more details on this.

By $\Z $ we denote the {\em integers} and by $\N=\{0,1,\ldots\}$ we denote the {\em non-negative integers}.
For every $a,b\in\Z$ with $a\leq b$ we define $[a,b]=\{k\in\Z\mid a\leq k\leq b\}$.
For every $n\geq 1$ we define $n\Z=\{n\cdot z\mid z\in\Z\}$.
For every number $n\in\N$ we define $\log(n)=\min\{i+1\mid i\in\N, n\leq 2^i\}$,
which is the smallest number of bits necessary to write down $n$ in binary.
For every finite alphabet $A$ we denote by $A^*$ the set of finite
words over $A$ and denote the empty word by $\varepsilon$.
For all $a\in A$ and all $w\in A^*$ let $|w|_a$ denote the number of occurrences
of the letter $a$ in $w$.
For every finite set $M\subset\N\setminus\{0\}$ let 
$\LCM(M)=\min\{n\geq 1\mid \forall m\in M\setminus\{0\}: m|n\}$
denote the least common multiple of the elements in $M$. 
For any $j \in \N$ let 
$\LCM(j)=\LCM([1,j])$ denote the least 
common multiple of the numbers $\{1,\ldots,j\}$.
 For any two sets $X$ and $S$, let $X^S$ denote the set of all functions from $S$ to $X$.
For any set $S$ let 
$\mathscr{P}(S) = \{ X \mid X \subseteq S \}$
denote the power set of $S$.

A {\em guard} over a finite set of clocks $\Omega$ and a finite set of parameters $P$ 
is a comparison of the form
$g=\omega \bowtie e$, where $ \omega \in \Omega$, $e\in P\cup\N$,
and $\bowtie\in\{<,\leq,=,\geq,>\}$;
in case $e\in P$ we call $g$ {\em parametric}, and 
{\em non-parametric} otherwise.
We denote by $\G(\Omega,P)$ the {\em set of guards} over the set of 
clocks $\Omega$ and the set of parameters $P$.
The {\em size} $|g|$ of a guard 
$g=\omega \bowtie e$ is defined as
$$
|g|=\begin{cases}\log(e)&\text{if $e\in\N$}\\
	1 &\text{otherwise}.
\end{cases}
$$
A {\em clock valuation} is a function from $\Omega$ to $\N$;
we write $\vec{0}$ to denote the clock valuation $\omega \mapsto 0$.
For each clock valuation $v$ and each $t\in\N$ we denote
by $v+t$ the clock valuation $\omega \mapsto v(\omega)+t$.
A {\em parameter valuation} is a function $\mu$ from $P$ to $\N$.
For every guard $g= \omega \bowtie p$ with $p\in P$ 
(resp. $g=\omega \bowtie k$ with $k\in\N$)
we write $v\models_{\mu} g$ if $v(\omega)\bowtie\mu(p)$
(resp. $v(\omega)\bowtie k$; in this case we may also simply write $v\models g$).
We define an {\em empty guard} $g_\epsilon$ over a non-empty finite set of clocks
$\Omega$ and a finite set of parameters $P$ to be of the form $\omega \geq 0$ for some 
$\omega \in \Omega$. In particular, we
defined $g_\epsilon$ such that for all $v \in \N^\Omega$ and all $\mu \in \N^P$
we have
$v \models_{\mu} g_\epsilon$, hence $g_\epsilon$ can be used as a guard that is always true.

\begin{center}
	\begin{figure}
\includegraphics[width=0.6\textwidth]{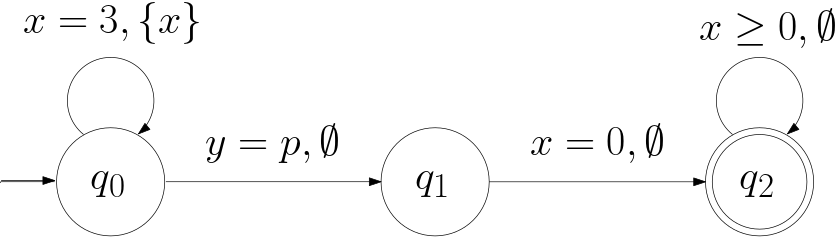}
	\caption{An example of a PTA. The automaton consists of three states, 
		the set of clocks is $\{x, y\}$, the set of parameters is $\{p\}$.
		The edges are represented by arrows labeled with the corresponding guard
		and the set of clocks $U$ to be reset.
		A parameter valuation $\mu$ witnesses that reachability holds
		for this PTA if, and only, if and only if, $\mu(p)\in3\mathbb{Z}$.}
		\label{example pta}
	\end{figure}
\end{center}
A parametric timed automaton as introduced in \cite{AHV93-stoc} is a finite automaton extended with a finite set of parameters $P$ and a finite set of clocks $\Omega$ that all progress at the same rate and that can be individually reset to zero. Moreover, every transition is labeled by a guard over 
$\Omega$ and $P$ and by a set of clocks to be reset.
Formally, a {\em parametric timed automaton } ({\em PTA} for short) is a tuple
$\A=(Q,\Omega,P,R,q_{init},F)$, where
\begin{itemize}
	\item $Q$ is a non-empty finite {\em set of control states}, 
	\item $\Omega$ is a non-empty finite {\em set of clocks},
	\item $P$ is a finite {\em set of parameters}, 
	\item $R\subseteq Q\times\G(\Omega,P)\times \mathscr{P}(\Omega) \times Q$
	is a finite {\em set of  rules},
	\item $q_{init}\in Q$ is an {\em initial control state}, and 
	\item $F\subseteq Q$ is a {\em set of final control states}.
\end{itemize}

A clock $\omega\in \Omega$ is called {\em parametric} if there exists some $(q,g,U,q')\in R$ such that
the guard $g$ is parametric.
	We also refer to $\A$ as a $(m,n)$-PTA
	if $m=|\{\omega\in\Omega\mid \omega\text{ is parametric}\}|$ is the number of parametric
clocks and $n=|P|$ is the number of parameters of $\A$ ---
sometimes we also just write $(m,*)$-PTA (resp. $(*,n)$-PTA) when $n$ (resp. $m$)
is a priori not fixed.

The {\em size} of $\A$ is defined as
$$
|\A| \ = \ |Q|+|\Omega|+|P|+|R|+\sum_{(q,g,U,q')\in R}|g|.
$$
Let 
$\Const(\A) = \{ c\in\N \mid \exists(q,g,U,q')\in R, \ \exists \omega \in \Omega : g = \omega \bowtie c \}$ denote the 
set of constants that appear in the guards of the rules of $\A$.
\newcommand{\Conf}{{\sf Conf}}
By $\Conf(\A)=Q\times\N^\Omega$ we denote the set of
{\em configurations} of $\A$. 
We prefer however to denote a configuration by $q(v)$ instead of $(q,v)$.

\begin{definition}
For each parameter valuation $\mu:P\rightarrow\N$ 
and each $(\delta,t)\in R\times\N$ with  	$\delta = (q,g,U,q')\in R$,
let $\xrightarrow{\delta,t,\mu}$ denote the binary relation
$\Conf(\A)$, where
$q(v)\xrightarrow{\delta,t, \mu} q'(v')$ if
	$v+t\models_{\mu} g$, 
	 $v'(u)=0$ for all $u \in U$ and $v'(\omega)=v(\omega)+t$ for all 
	$\omega \in \Omega \setminus U$.
\end{definition}

A {\em $\mu$-run} from $q_0(v_0)$ to $q_n(v_n)$ is a sequence 
$q_0(v_0)\xrightarrow{\delta_1,t_1,\mu}q_1(v_1)\cdots\xrightarrow{\delta_n,t_n,\mu}q_n(v_n)$;
it is called {\em reset-free} if the set appearing in the third component is empty
for all $\delta_i$.
In case $P=\{p\}$ is a singleton and $\mu(p)=N$ we prefer to say {\em $N$-run}
instead of $\mu$-run.
We say {\em reachability holds} for $\A$ if
there is a $\mu$-run from $q_{init}(\vec{0})$  
to some configuration $q(v)$ for some $q\in F$, some $v\in\N^\Omega$, and some
$\mu\in\N^P$.
We refer to Figure~\ref{example pta} for an instance
of a PTA for which reachability holds. \newline

In the particular case where $P=\{p\}$ is a singleton for some parameter $p$ and $\mu(p)=N$
we prefer to write $q(v)\xrightarrow{N}q'(v')$
to denote $q(v)\xrightarrow{\mu}q'(v')$ and prefer to write $\models_N$ to denote $\models_\mu$.

It is worth mentioning that there are further modes of time valuations and guards which exist in the literature, we refer
to \cite{Andre19} for a recent overview.

We are interested in the following decision problem.

\problemx{$(m,n)$-PTA-Reachability}
{A $(m,n)$-PTA $\A$.}
{Does reachability hold for $\A$?\newline}

Alur et al. have already shown in their seminal paper that {\sc PTA-Reachability} is in general undecidable,
already in presence of only three parametric clocks~\cite{AHV93-stoc},
Bene\v{s} et al. strengthened this when only one parameter is present~\cite{BenesBLS15}.

\begin{theorem}[\cite{BenesBLS15}]
	{\sc $(3,1)$-PTA-Reachability} is undecidable.
\end{theorem}

To the contrary, {\sc $(1,*)$-PTA-Reachability} has
recently been shown to be complete 
for $\mathsf{NEXP}$, where
a non-elementary upper bound was initially given by
Alur et al.~\cite{AHV93-stoc}.

\begin{theorem}[\cite{BundalaO17,BenesBLS15,BQS-lmcs19}]
	{\sc $(1,*)$-PTA-Reachability} is $\mathsf{NEXP}$-complete.
\end{theorem}

On the other end, decidability of {\sc $(2,*)$-PTA-Reachability} is
still open to the best of our knowledge. In presence of one parameter the following is known.

\begin{theorem}[\cite{BundalaO17}]
	{\sc $(2,1)$-PTA-Reachability} is decidable and $\mathsf{PSPACE}^{\mathsf{NEXP}}$-hard.
\end{theorem}
The following theorem states our main result.

\begin{theorem}\label{main result}
	{\sc $(2,1)$-PTA-Reachability} is $\EXPSPACE$-complete.
\end{theorem}

\section{Lower bounds}\label{lower section}
\newcommand{\bit}{\textsc{Bit}}
\newcommand{\bin}{\textsc{Bin}}
\newcommand{\val}{\textsc{Val}}

In this section we show an $\EXPSPACE$ lower bound for 
{\sc $(2,1)$-PTA-Reachability}.
Section~\ref{section computations} introduces some auxiliary gadgets that 
show that on small $(2,1)$-PTA with two parametric clocks $x$ and $y$ and one parameter
$p$ one can perform both
(i) $\PSPACE$ computations and 
(ii) compute $x-y\text{ mod }p$ modulo numbers that are
dynamically given in binary.
Section~\ref{section EXPSPACE} builds upon these auxiliary gadgets
and shows how to implement serializability
computations
in a leaf language characterization of 
$\EXPSPACE$~\cite{GHOW10}, which
is a simple padded variant of the leaf language characterization
of $\PSPACE$ due to Hertrampf et al.~\cite{HLSVW93}.

\subsection{PSPACE and modulo computations}\label{section computations}
For each $i,n\in\N$ let $\bit_i(n)$ denote the $i$-th least significant bit of the
binary presentation of $n$, where the least significant bit is on the left,
i.e. $n=\sum_{i\in\N}2^i\cdot\bit_i(n)$.
For each $m\geq1$, by $\bin_m(n)=\bit_0(n)\cdots\bit_{m-1}(n)$ 
we denote the sequence of the first $m$ least significant bits
 of the binary representation of $n$.
Conversely, given a binary string 
$w=w_0\cdots w_{m-1}\in\{0,1\}^{m}$ 
of length $m$
we denote by $\val(w)=\sum_{i=0}^{m-1}
2^i\cdot w_i\in[0,2^{m}-1]$ the value of
$w$ interpreted as a non-negative integer.

Let $\A$ be a parametric timed automaton over a set of clocks $\Omega$
with two parametric clocks $x$ and $y$.
We say a valuation $v:\Omega \rightarrow\N$ is {\em bit-compatible}
if $v(\zeta)\in\{0,1\}$ for all non-parametric clocks $\zeta\in \Omega$ of $\A$.
Assume moreover that $\Omega$ contains {\em non-parametric clocks} 
$\Theta_+\cup \Theta_-$,
where $\Theta$ is some set and $\Theta_+=\{\vartheta^+\mid \vartheta\in \Theta\}$ and
$\Theta_-=\{\vartheta^-\mid \vartheta\in \Theta\}$ are two disjoint corresponding copies of $\Theta$;
in this case, for any valuation $v:\Omega \rightarrow\N$ we define the mapping $\widehat{v}:\Theta \rightarrow\{0,1\}$
as 
$$\widehat{v}(\vartheta)=
\begin{cases}0 & \text{if $v(\vartheta^+)=v(\vartheta^-)$}\\
	1 & \text{otherwise}.
\end{cases}
$$
In the following we call such non-parametric clocks 
$\{\vartheta^+,\vartheta^-\mid \vartheta \in \Theta\}$, 
appearing as implicit pairs, {\em bit clocks} since
they are used to encode bits.
The following definition expresses
when a parametric timed automaton 
over two parametric clocks
and one parameter computes a function from $\N\times\{0,1\}^n$
to $\{0,1\}^m$. 
Notably, both before execution it assumes,
and after execution it guarantees, a 
bit-compatible valuation that assigns
its two parametric clocks values in the interval $[0,N-1]$,
where $N$ denotes the assigned value of its only parameter.

\begin{definition} \label{def computes}
A $(2,1)$-PTA 
	$\A=(Q,\Omega,\{p\},R,q_{init},\{q_{fin}\})$ 
	whose parametric clocks are $x$ and $y$ and whose one parameter is $p$ {\em computes a function 
	$f:\N\times\{0,1\}^n\rightarrow\{0,1\}^m$}
	if its set of clocks $\Omega$ contains two 
	disjoint sets of
	\begin{itemize}
		\item non-parametric ``input'' bit clocks
			$\{{in_0}^+,{{in}_0}^-,\ldots, {in}_{n-1}^+,{in}_{n-1}^-\}$,
		\item non-parametric ``output'' bit 
			clocks 
			$\{{{out}_0}^+,{{out}_0}^-,\ldots, {{out}_{m-1}}^+,{{out}_{m-1}}^-\}$,
	\end{itemize}
	such that for all $N\in\N$ and all bit-compatible $v_0:\Omega \rightarrow[0,N-1]$ we have
	\begin{enumerate}
		\item $q_{init}(v_0)\xrightarrow{N}^*q_{fin}(v')$ for some bit-compatible
			$v':\Omega \rightarrow[0,N-1]$ and
		\item for all $v':\Omega \rightarrow\N$ for which 
		$q_{init}(v_0)\xrightarrow{N}^*q_{fin}(v')$ we have
		\begin{itemize}
			\item $v'\in[0,N-1]^\Omega$ is bit-compatible,
			\item  $\widehat{v'}(in_i)=\widehat{v_0}(in_i)$ for all 
				$i\in[0,n-1]$,
			\item $v'(x)-v'(y)\equiv v_0(x)-v_0(y)\text{ mod } N$, and
			\item $\prod_{j=0}^{m-1}\widehat{v'}({out}_j)= 
				f(v_0(x)-v_0(y)
				\text{ mod }N,\prod_{i=0}^{n-1} \widehat{v_0}({in}_i))$,
				where $\prod$ denotes concatenation.
		\end{itemize}
	\end{enumerate}
\end{definition}
Importantly, the execution of any
$N$-run $q_{init}(v_0)\xrightarrow{N}q_{fin}(v')$
does not have any side effects on the binary interpretation
of the `input'' bit clocks, i.e. the
string 
$\prod_{i=0}^{n-1}\widehat{v_0}(in_i)$
equals 
$\prod_{i=0}^{n-1}\widehat{v'}(in_i)$.

The following lemma essentially has its
roots in the $\PSPACE$-hardness proof for
the emptiness problem for timed automata (without parameters) introduced
by Alur and Dill~\cite{alur1994theory},
however constructed
to satisfy the carefully chosen interface given by
Definition~\ref{def computes}.

\begin{lemma}\label{lemma PSPACE}
For every $\PSPACE$-computable function $g:\{0,1\}^n\rightarrow\{0,1\}^m$ 
one can compute in polynomial time in $n+m$ a 
	$(2,1)$-PTA computing the function 
	$f:\N\times\{0,1\}^n\rightarrow\{0,1\}^m$,
	where $f(k,w)=g(w)$ for all $(k,w)\in\N\times\{0,1\}^n$.
\end{lemma}

\begin{proof}
Let us fix some $\PSPACE$-computable function $g:\{0,1\}^n\rightarrow\{0,1\}^m$.
Let us moreover fix some $t(n)$-space bounded deterministic Turing
	machine $\mathcal{M}$ computing $g$, where $t$ is some fixed polynomial.

We explicitly store the value of our input 
	by making use of our 
non-parametric ``input'' bit clocks 
			$\{{in_0}^+,{{in}_0}^-,\ldots, 
			{in}_{n-1}^+,{in}_{n-1}^-\}$.
Similarly, we explicitly store the value of our output with
the 
non-parametric ``output'' bit 
			clocks $\{{{out}_0}^+,{{out}_0}^-,\ldots, {{out}_{m-1}}^+,{{out}_{m-1}}^-\}$.
Since $f(k,w)=g(w)$ we need to provide 
a computation that presents $g(w)\in\{0,1\}^m$
using the ``output'' bit clocks.
	Let $\Omega$ denote the set of clocks
	of the $(2,1)$-PTA $\A$ whose construction we discuss next.

For every non-parametric clock in $\A$ we reset it once it has value $2$;
this is achieved by suitable self-loops in every state of the construction
except for the final control state $q_{fin}$.
	Similarly, we establish that both 
	of the parametric clocks $x$ and $y$ are being reset 
once they have reached value $N$. This way the difference between the values of $x$ and $y$ will stay
unchanged modulo the valuation $N$ of the only parameter $p$. 
	Importantly, other than that neither $x$ nor $y$ will be modified during the following construction.

	We will enforce that finally the values of all non-parametric clocks remain in $\{0,1\}$ and that 
	the two parametric clocks have a value
	in $[0,N-1]$ as follows.
	A final control state $q_{fin}$ 
	is preceded by a final
gadget in which no time elapses 
	that verifies via a sequence of suitable
guards that the parametric and non-parametric clocks 
	are as required.

Let us consider now any pair of bit clocks $\vartheta^+$ and $\vartheta^-$
and any current bit-compatible valuation $v:\Omega \rightarrow\N$.
We have $\widehat{v}(\vartheta)=1$ if, and only if, 
either $v(\vartheta^+) =0$ and $v(\vartheta^-) =1$ or conversely $v(\vartheta^+)=1$ and $v(\vartheta^-) =0$.
Similarly, when we want to set the value $\widehat{v}(\vartheta)$
to $0$, we reset both clocks $\vartheta^+$ and $\vartheta^-$ at the same time,
and when we want to set the value $\widehat{v}(\vartheta)$
to $1$, we reset $\vartheta^-$ when $v(\vartheta^+)= 1$ without resetting $\vartheta^+$.

	For simulating $\mathcal{M}$ our $(2,1)$-PTA
	$\A$ will also use 
 suitable $O(t(n))$ bit clocks, 
to store in binary the working tape of $\mathcal{M}$.

Given the current bit-compatible 
	valuation $v:\Omega \rightarrow\N$,
it is thus possible to inspect the input bit string $\prod_{i=0}^{n-1} \widehat{v}({in}_i)$,
read and write the polynomially sized working tape, 
and to write the output $\prod_{j=0}^{m-1}\widehat{v}({out}_j)$.
Let us discuss this in more detail.

For simulating $\mathcal{M}$, we choose the control states of our $(2,1)$-PTA $\A$
as $$S \times \{0, \ldots,n-1 \}
\times \{0, \ldots,m-1 \}\times
\{0, \ldots,t(n)-1\}
 \times \{0,1 \} 
\times \{0,1 \},$$ 
where $S$ is the set of states of $\mathcal{M}$.
We then simulate any step of $\mathcal{M}$ from a state $q$, current position $i$ on the input
tape, current position $j$ on the output tape, current position
$h$ on the working tape, reading letter $a$ on the input tape, reading letter $b$ on
the working tape, changing to a new state $q'$, new input head position $i'$, 
new output head position $j'$, and
new working head position $h'$.
To do that, we add to $\A$ sequences of suitable rules from
control state $(q,i,j,h,a,b)$ to control state $(q',i',j',h',a',b')$
 for all $a',b' \in \{0,1 \}$,
by using suitable guards and reset operations that serve two purposes:
 first, checking whether $a'$ and $b'$ are indeed the values of 
the $i'$-th  (resp. $h'$-th) cell of  the  input (resp. working)  tape
and second, writing on the $j$-th (resp. $h$-th) cell of  the  output (resp. working) tape.

Letting 
$q_{init}$ denote some suitable initial state 
one can thus achieve that for all bit-compatible $v_0: \Omega \rightarrow[0,N-1]$ and all $v':\Omega \rightarrow\N$, 
if $q_0(v_0)\xrightarrow{N}^*q_{fin}(v')$ then $v'$ is again a bit-compatible
valuation from $\Omega$ to $[0,N-1]$.
\end{proof}

\begin{remark}\label{remark wlog}
The proof of Lemma~\ref{lemma PSPACE} shows that 
if $g:\N\times\{0,1\}^n\rightarrow\{0,1\}^m$ is computable
	by a $(2,1)$-PTA, then so is
	the function $f:\N\times\{0,1\}^{n+\ell}\rightarrow\{0,1\}^m$, where
	$f(k,w)=g(k,w_1\cdots w_n)$ 
	for all $k\in\N$ and all $w=w_{1}\cdots 
	w_{n+\ell}\in\{0,1\}^{n+\ell}$:
	indeed, one can manipulate the $\ell$
	additional input bit clocks
by repeatedly resetting them once 
	they have value $2$, enforcing that the associated $\widehat{v}$-values stay throughout unchanged
and that their value is finally strictly smaller than $2$.
\end{remark}

The following lemma shows that $(2,1)$-PTA
can compute modulo dynamically given numbers in binary.

\begin{lemma}\label{lemma modulo}
	One can compute in polynomial time in $n+m$ a 
	$(2,1)$-PTA that
	computes the function 
	$f:\N\times\{0,1\}^n\rightarrow\{0,1\}^m$,
	where $f(k,w)=\bin_m(k\text{ mod }\val(w))$.
\end{lemma}

\begin{proof}

We need to show that in time polynomial in $n+m$ one
	can construct a $(2,1)$-PTA 
	$\A$ whose set of clocks $\Omega$
contains the ``input'' bit clocks 
	$\{{in_0}^+,{{in}_0}^-,\ldots, {in}_{n-1}^+,{in}_{n-1}^-\}$
and ``output'' bit clocks
$\{{{out}_0}^+,{{out}_0}^-,\ldots, {{out}_{m-1}}^+,{{out}_{m-1}}^-\}$
that computes $f$.
Let us assume some parameter value $N\in\N$ and some
 bit-compatible valuation $v_0:\Omega \rightarrow[0,N-1]$
satisfying $w=\prod_{i=0}^{n-1}\widehat{v_0}({in}_i)$.

Again, we establish here also that the parametric clocks $x$ and $y$ are being reset 
	once they have reached value $N$ ---
	however we sometimes 
	explicitly disallow $x$ to reach value $N$ in
	certain gadgets mentioned below. 
	This will be the only modification of $x$ and $y$.
In the following, reading and writing the $\widehat{v}(\vartheta)$-value
for every pair of bit clocks $\vartheta^+,\vartheta^-$,
guaranteeing that $v(\vartheta^+),v(\vartheta^-)\in\{0,1\}$,
	and guaranteeing that the parametric clocks
	finally have values in $[0,N-1]$
can be done as in the proof of Lemma~\ref{lemma PSPACE}.

	We need the eventual output bit string
$\prod_{j=0}^{m-1}\widehat{v'}({out}_j)$
to be equal to
$$
f(v_0(x)-v_0(y)\text{ mod }N,w)
= \bin_m( (v_0(x)-v_0(y)\text{ mod }N) \text{ mod }\val(w)).$$

Our automaton starts in some initial control state $q_{init}$.
From $q_{init}$ we introduce a gadget that nondeterministically writes
some value $u\in\{0,1\}^m$ in our ``output'' bit clocks
	that satisfies $\val(u)<\val(w)$.
	From the end of the latter gadget we have a rule that checks if our parametric clock $x$ has value $0$ (just after being reset with value $N$), leading to a control state $q_{wait}$.
Assume our current valuation then is $v:\Omega \rightarrow\N$.
From $q_{wait}$ we have a rule to a state $q_{sub}$ letting no time elapse
from which we claim there is a gadget that allows us to loop in $q_{sub}$ for precisely
 $\val(w)=\val(\prod_{i=0}^{n-1} \widehat{v}({in}_i)))$
time units. One constructs the latter gadget as follows. 
	Subsequently for every $i\in[0,n-1]$ one reads $\widehat{v}({in}_{i})$
	and in case $\widehat{v}(in_{i})=1$ lets 
	precisely $2^i$ time units elapse via a suitable auxiliary
	clock and in case $\widehat{v}(in_{i})=0$ lets $0$ time 
	units elapse.
	The gadget ends with a sequence of rules leading back to $q_{sub}$ by
	letting $0$ time units 
	elapse
that verify that the parametric clock $x$ has a value strictly smaller than $N$.
	Importantly, the parametric clock $x$ is exceptionally {\em not reset} inside this gadget.

Finally, we add a rule from $q_{sub}$ to a suitable
gadget that lets precisely $\val(\prod_{j=0}^{m-1}\widehat{v}(out_j))$ time units 
	elapse (analogously as
	done above), followed by a test that verifies that the value of
	$y$ equals $0$ (after just being reset at value $N$).
	In addition, we append this latter gadget with a final sequence 
	of rules (again letting no time elapse) 
	to our final control state $q_{fin}$ 
	that test
if both $x$ and $y$ have a value strictly smaller than $N$
	and test if all 
	non-parametric clocks have a value strictly smaller than $2$.
Thus, every valuation $v':\Omega \rightarrow\N$ for which $q_{init}(v_0)\xrightarrow{N}^*q_{fin}(v')$
holds is a bit-compatible valuation from $\Omega$ to $[0,N-1]$.

It is worth noting that by construction precisely $v_0(x)-v_0(y)\text{ mod }N$ time units have passed 
in any computation $q_{wait}$ to $q_{fin}$.
Since we have repeatedly waited $\val(w)$ time units and finally verified that
the remaining time is the guessed value initially 
nondeterministically written to our ``output'' bit clocks, 
we have 

	\begin{eqnarray*}
		\prod_{j=0}^{m-1}\widehat{v'}(out_j)&=&
\bin_{m}((v_0(x)-v_0(y)\text{ mod }N)\text{ mod }\prod_{i=0}^{n-1}
\widehat{v_0}({in}_i)))\\
		&=&
f(v_0(x)-v_0(y)\text{ mod }N,\prod_{i=0}^{n-1}
\widehat{v_0}({in}_i))
	\end{eqnarray*}
		
for any valuation $v':\Omega \rightarrow\N$ with $q_{init}(v_0)\xrightarrow{N}^*q_{fin}(v')$, as required.
\end{proof}

\subsection{An EXPSPACE lower bound via serializability}\label{section EXPSPACE}

This section is devoted to showing the following lower bound.

\begin{theorem}
	{\sc$(2,1)$-PTA-Reachability} is $\EXPSPACE$-hard.
\end{theorem}

For each language $L\subseteq A^*$ let $\chi_L:A^*\rightarrow\{0,1\}$ denote
its characteristic function.
By $\preceq_n$ we denote the lexicographic order on $n$-bit strings,
thus $w\preceq_n v$ if $\val(w)\leq\val(v)$, e.g.
$0101\preceq_4 0011$.

Our $\EXPSPACE$ lower bound proof makes use of the following
leaf language view of $\EXPSPACE$ from~\cite{GHOW10}, which is a 
padded adjustment of the
leaf-language characterization of $\PSPACE$
from \cite{HLSVW93}, which in turn has its roots in Barrington's Theorem~\cite{Bar89}.

\begin{theorem}[Theorem 2 in \cite{GHOW10}]\label{thm serializability}
	For every language $L\subseteq\{0,1\}^*$ in $\EXPSPACE$ there exists a 
	polynomial $s:\N\rightarrow\N$, a regular language
	$R\subseteq\{0,1\}^*$, and a language $U\in\LOGSPACE$ such that for all $w\in\{0,1\}^n$
	we have
\begin{eqnarray}
	w\in L\quad\Longleftrightarrow\quad
	\prod_{m=0}^{2^{2^{s(n)}-1}}\chi_U(w\cdot\bin_{2^{s(n)}}(m))\in R,
\label{def serializability}
\end{eqnarray}
where $\cdot$ and $\prod$ denote string concatenation.
\end{theorem}

Let us fix any language $L$ in $\EXPSPACE$ and assume $L\subseteq\{0,1\}^*$ without
loss of generality.
Applying Theorem~\ref{thm serializability}, let us fix 
the regular language $R\subseteq\{0,1\}^*$ along with some fixed
deterministic finite automaton
$D=(Q_D,\{0,1\},q_0,\delta_D,F_D)$
with $L(D)=R$,
the fixed polynomial $s$ and the fixed language $U\in\LOGSPACE$.
Let us moreover fix an input $w\in\{0,1\}^n$ of length $n$ for $L$.
Figure~\ref{fig serializability} rephrases characterization (\ref{def serializability}) 
in Theorem~\ref{thm serializability}
in terms of an execution of a program that returns $1$ if, and only if, $w\in L$.

\begin{center}
\begin{figure}
\fbox{\parbox[t][5cm][c]{7cm}{
$\phantom{a}$\\
(1)\phantom{aaaaaa}\textbf{var} $q \in Q_D$\\
(2)\phantom{aaaaaa}\textbf{var} $b \in \{0,1\}$\\
(3)\phantom{aaaaaa}\textbf{var} $B\in\mathbb{N}$\\
	(4)\phantom{aaaaaa}$q := q_0$ \newline
	(5)\phantom{aaaaaa}$B:=0$ \newline
(6)\phantom{aaaaaa}\textbf{while} $B < 2^{2^n} $ \textbf{loop} \newline
	(7)\phantom{aaaaaa}$\phantom{aaaa}b:= \chi_U(w\cdot \bin_{2^{s(n)}}(B))$ \newline
(8)\phantom{aaaaaa}$\phantom{aaaa}q:= \delta_D(q,b)$ \newline
(9)\phantom{aaaaaa}$\phantom{aaaa}B:= B+1$ \newline
(10)\phantom{aaaaa}\textbf{end loop} \newline
	(11)\phantom{aaaaa}\textbf{return} $q \in F_D$ \newline
	}
}
	\caption{A program returning $1$ if, and only if, $w\in L$
	(using the characterization in Theorem~\ref{thm serializability}),
where $D=(Q_D,\{0,1\},q_0,\delta_D,F_D)$ is some deterministic finite automaton
	such that $L(D)=R$.}\label{fig serializability}
\end{figure}
\end{center}

Making use of $D$, $s$ and $U$ we will translate our input $w\in\{0,1\}^n$ in polynomial time
(in $|w|=n$) to some $(2,1)$-PTA
$\A=(Q,\Omega,P,R,q_{init},F)$ such that
\begin{itemize}
	\item $\Omega$ will contain precisely two parametric clocks $x$ and $y$ and
		further clocks that are non-parametric,
	\item $P=\{p\}$ is a singleton, and
	\item $w\in L$ if, and only, if reachability holds for $\A$.
\end{itemize}
The following lemma gives us a gadget $(2,1)$-PTA
that allows us to enforce that the parameter $p$ 
can only be evaluated to numbers that are larger than $2^{2^{s(n)}}$.

\begin{lemma}\label{lemma big}
	One can compute in polynomial time in $n$ some
	parametric timed automaton 
	$\A_{big}=(Q_{big},\Omega_{big},\{p\},R_{big},q_{big,init},\{q_{big,fin}\})$ 
	with two parametric clocks $x,y\in \Omega_{big}$
	and one parameter $p$ such that
	\begin{enumerate}
		\item $q_{big,init}(\vec{0})\xrightarrow{N}^*
			q_{big,fin}(v')$ for some $v':\Omega_{big}\rightarrow\N$ for some $N\in\N$, and
		\item for all $N\in\N$ and all $v':\Omega_{big}\rightarrow\N$ we have
			$q_{big,init}(\vec{0})\xrightarrow{N}^*q_{big,fin}(v')$ implies
			$N>2^{2^{s(n)}}$.
	\end{enumerate}
	
\end{lemma}

\begin{proof}
	Without loss of generality we may assume $2^{s(n)+1}\geq 10$.
	Letting $N$ denote the parameter value of its only parameter $p$,
	our 
	$(2,1)$-PTA $\A_{big}$ will test 
	whether $N-1$ is divisible by all numbers in the interval $[1,2^{s(n)+1}-1]$.
	This will be sufficient since $\LCM([1,k])\geq 2^k$ for all $k\geq 9$ by \cite{Nai82},
	thus implying $N>N-1\geq\LCM([1,2^{s(n)+1}-1])\geq2^{2^{s(n)+1}-1}>2^{2^{s(n)}}$.
Consider the following program which returns $1$ if, and only if, all numbers in
	$[1,2^{s(n)+1}-1]$ divide $N-1$.
	
\fbox{\parbox[t][4cm][c]{9cm}{
$\phantom{a}$\\
	(1)\phantom{aaaaaa}\textbf{var} $I \in \{0,1\}^{s(n)+1}$\\
	(2)\phantom{aaaaaa}\textbf{var} $J \in \{0,1\}^{s(n)+1}$\\
	(3)\phantom{aaaaaa}$I:=0^{s(n)+1}$\newline
	(4)\phantom{aaaaaa}\textbf{while} $I \not= 1^{s(n)+1} $ \textbf{loop} \newline
	(5)\phantom{aaaaaa}$\phantom{aaaa}I:=\bin_{s(n)+1}(\val(I)+1)$\newline
	(6)\phantom{aaaaaa}$\phantom{aaaa}J:= \bin_{s(n)+1}(N-1\text{ mod }\val(I))$\newline
	(7)\phantom{aaaaaa}$\phantom{aaaa}\textbf{if} J\not=0^{s(n)+1}\textbf{ then return } 0$\newline
(8)\phantom{aaaaa}\textbf{end loop} \newline
(9)\phantom{aaaaa}\textbf{return} $1$ \newline
	}
}
\medskip

	It remains to show that the program can be 
	implemented by a $(2,1)$-PTA $\A_{big}$
	with a suitable final control state $q_{big,fin}$.

	It is straightforward to initialize our two parametric clocks $x$ and $y$ in
	such a way that 
	one can enforce valuations $v$ that satisfy $v(x)-v(y)=N-1\text{ mod }N$:
	indeed, starting from the valuation 
	$\vec{0}$, we can wait one unit
	of time after which we reset $x$ but not $y$.

	We will use $O(s(n))$ suitable bit clocks for storing the variables $I$ and $J$
	respectively.

	Lines (3), (4) and (7) can easily directly be achieved by reading and writing the
	$O(s(n))$ many bits clocks reserved for storing $I$ and $J$. Line (5) boils
	down to incrementing $I$ when viewed as $s(n)+1$ bit integer and
	is thus obviously a polynomial space computable function 
	from $\N\times\{0,1\}^{s(n)}$ to $\{0,1\}^{s(n)}$ and hence computable
	using a suitable PTA based on Lemma~\ref{lemma PSPACE}.
	Line (6) is a function from
	$\N\times\{0,1\}^{s(n)+1}$ to 
	$\{0,1\}^{s(n)+1}$
	that can be implemented using a suitable PTA based on Lemma~\ref{lemma modulo}.

	As in the proofs of Lemma~\ref{lemma PSPACE} and Lemma~\ref{lemma modulo}
	we reset the two parametric clocks $x$ and $y$ once they have reached value $N$
	but only in case we are outside any of the gadget PTA corresponding to
	line (5) and line (6), respectively.
	Similarly we realize the implementation of the bit clocks for $I$ and $J$
	by resetting them once they have reached value $2$.
\end{proof}

Recall that we aim at implementing the program in 
Figure~\ref{fig serializability} by a 
$(2,1)$-PTA $\A$.
The initial part of $\A$ will consist 
of the gadget PTA $\A_{big}$
which will allow us to enforce 
an assignment of $p$ to some value 
$N> 2^{2^{s(n)}}$.
We first explain how to encode its variables and then discuss how to implement
the different lines of the program.

\medskip
{\bf Encoding the variables of the program in Figure~\ref{fig serializability}.}
Our PTA $\A$ will store in its control states
the current state $q$ of $D$ and the boolean variable $b$.
We cannot easily ``explicitly'' store the value of our variable $B$ in
binary as in the proof of Lemma~\ref{lemma PSPACE} via 
polynomially many bit clocks 
in such a way that, given the current valuation $v:\Omega \rightarrow\N$,
it suffices to simply inspect their $\widehat{v}$-value: 
indeed, there are only singly-exponentially many different combinations of such $\widehat{v}$-values,
yet $B$ is a number in $[0,2^{2^{s(n)}}]$ and 
thus of doubly-exponential magnitude.
We will rather store the value $B\in\N$ as the difference $v(x)-v(y)\text{ mod }N$
between our only two parametric clocks $x$ and $y$:
this is possible since $N> 2^{2^{s(n)}}$ by our initial gadget PTA $\A_{big}$.
However, when inspecting line (7) of Figure~\ref{fig serializability} we need to access certain bits of
the exponentially long bit string $w\cdot\bin_{2^{s(n)}}(B)$.
For this, we access $B$ in a different representation, namely in 
Chinese Remainder Representation that
we introduce next.
\newcommand{\CRR}{\mathsf{CRR}}
\begin{definition}[Chinese Remainder Representation]
	Let $p_i$ denote the $i$-th prime number and assume $\prod_{i=1}^mp_i>B$
	for some $m\in\N$.
	Then $\CRR_m(B)$ denotes the bit tuple $(b_{i,r})_{i\in[1,m],r\in[0,p_i-1]}$,
	where $b_{i,r}=1$ if $B\text{ mod }p_i=r$ and $b_{i,r}=0$ otherwise.
\end{definition}

Since $B$ will need to take
values in $[0,2^{2^{s(n)}}]$
and for every $k\in\N$ we have $\prod_{i=1}^{k}p_i>2^k$
there exists some $m\in O(\log(2^{2^{s(n)}}))=2^{\poly(n)}$
such that $\prod_{i=1}^m p_i>B$.
In other words, one can present $B$ as
\begin{eqnarray}
	\CRR_m(B)=(b_{i,r})_{i\in[1,m],r\in[0,p_i-1]}\quad\text{for some } m\in 2^{\poly(n)}\quad.
	\label{CRR B}
\end{eqnarray}
Since by the Prime Number Theorem the $i$-th prime $p_i$ is bounded by $O(i\log i)$
there exists some $\ell\in O(\log(m\log m))=O(\log(2^{\poly(s(n))}))=\poly(n)$
such that $\ell$ bits are sufficient to store in 
{\em binary}
precisely one of the primes $p_i$.
Thus, similarly
$O(\ell)=\poly(n)$ bits are sufficient
to store in {\em binary} precisely one of the pairs of the form $(i,r)$, where $i\in[1,m]$ and $r\in[0,p_i-1]$.
Moreover we have $|\CRR(B)|\in O(m^2\log m)=2^{\mathsf{poly}(s(n))}=2^{\poly(n)}$.

Observe that in line (7) of our program in Figure~\ref{fig serializability} 
we need to carry out $\LOGSPACE$ computations on our
exponentially long string $w\cdot\bin_{2^{s(n)}}(B)$.
Yet, if at all, we only have an on-the-fly mechanism for accessing the Chinese Remainder Representation
of $B$, notably still of exponential size in $n$. To have a chance to access concrete bits of $B$, we apply the following theorem that states that, given a number
in Chinese Remainder Representation, one can compute
in $\LOGSPACE$ 
its binary representation.

\begin{theorem}[Theorem 3.3. in \cite{ChDaLi01}]\label{thm CRR}
The following problem is computable in 
	$\mathsf{DLOGTIME}$-uniform $\mathsf{NC}^1$ (and thus in $\LOGSPACE$):

	INPUT: $\CRR_m(B)$ and $j\in[1,m]$

	OUTPUT: $\bit_j(B\text{ mod } 2^m)$
\end{theorem}

\medskip
{\bf Realization of line (7) in the program in Figure~\ref{fig serializability}.}
Let us assume that we have $B<2^{2^{s(n)}}$ and recall that we have
stored $B$ as the difference $v(x)-v(y)\text{ mod }N$ of our two parametric clocks $x$ and $y$,
assuming $v$ to be our current clock valuation.
Let us show how to compute
$\chi_U(w\cdot\bin_{2^{s(n)}}(B))$, where we recall that $U$ is a language in $\LOGSPACE$.
Let us fix some logarithmically space bounded deterministic Turing machine $\mathcal{M}$ for $U$.

For simulating $\mathcal{M}$ our PTA $\A$ will
use $O(\log(n+2^{s(n)}))=\poly(n)$ auxiliary bit clocks $\mathcal{J}$
to store in binary the position of the input head of $\mathcal{M}$
and further 
$O(\log(n+2^{s(n)}))=\poly(n)$
auxiliary bit clocks $\mathcal{W}$ in order to store the working tape $\mathcal{M}$.
Reading and writing on the working tape as well as updating the 
position of the input head can done analogously as in the proof of Lemma~\ref{lemma PSPACE}.
It only remains to show how to access the cell content 
$\bit_j(w\cdot\bin_{2^{s(n)}}(B))$
of the input head of $\mathcal{M}$, where we recall that $j$ 
itself is stored inside 
the above-mentioned bit clocks $\mathcal{J}$.

To compute $\bit_j(w\cdot\bin_{2^{s(n)}}(B))$ we apply
Theorem~\ref{thm CRR} and simulate in turn
a $\LOGSPACE$ machine $\mathcal{M}'$ whose
input is assumed to be
$$
\CRR(B)=(b_{i,r})_{i\in[1,m],r\in[0,p_i-1]} \text{ and $j\in[1,m]$},
$$
where we already have direct access to $j$ via the bit clocks $\mathcal{J}$
but need a special treatment in order to access the $b_{i,r}$ of $\CRR(B)$.
Importantly, during the to-be discussed simulation of $\mathcal{M}'$ we never modify the 
$\widehat{v}$-values associated with the bit clocks in
$\mathcal{J}$ and $\mathcal{W}$ that are being used in the (outermost) simulation of $\mathcal{M}$.
Before discussing the access to the $b_{i,r}$ let us first discuss 
the simulation of the working tape of $\mathcal{M}'$:
this can be achieved by using $O(\log(|\CRR(B)|+s(n)))=O(\log(m^2\cdot\log m+s(n)))=\poly(n)$
many auxiliary bit clocks $\mathcal{W}'$, say, where
reading and writing the working tape is done again as in Lemma~\ref{lemma PSPACE}.
It remains to discuss how to implement the input head in the simulation of
$\mathcal{M}'$.
As mentioned repeatedly above, input $j$ can directly be accessed by the
bit clocks $\mathcal{J}$.
However, accessing 
$\CRR(B)=(b_{i,r})_{i\in[1,m],r\in[0,p_i-1]}$
cannot be done explicitly but on-the-fly:
for this we reserve
$O(\ell)=O(s(n))=\poly(n)$ additional auxiliary bit clocks $\mathcal{J}'$, say, 
to store in binary a pair of indices $(i,r)$,
where $i\in[1,m]$ and $r\in[0,p_i-1]$.
Given the binary access to $(i,r)$ via the bit clocks $\mathcal{J}'$, 
one can compute via further suitable $\poly(\ell)=O(s(n))=\poly(n)$ bit clocks
$\mathcal{H}$, say, the binary
representation of the $i$-th prime number $p_i$ in space polynomial in $\ell$
(and thus in $n$) by Lemma~\ref{lemma PSPACE}: 
indeed, given $i\in[1,m]$ in binary, i.e. using $\ell=\poly(n)$ bits,
it is straightforward to compute
the $i$-th prime in space polynomial in $\ell$.
Having a binary resentation of $p_i$ via the bit clocks $\mathcal{H}$
one can finally compute $(v(x)-v(y)\text{ mod } N)\text{ mod }p_i$ via a gadget 
by Lemma~\ref{lemma modulo}.
Our $(2,1)$-PTA $\A$ can thus indeed compute $B\text{ mod } p_i$ and thus decide if
$r$ equals the latter, which in turn is nothing but computing the 
to-be-computed input bit $b_{i,r}$ of $\CRR(B)$ for
the simulation of $\mathcal{M}'$.

Concerning the implementation details of the simulation of $\mathcal{M}'$ 
it is important to remark (recalling Remark~\ref{remark wlog})
that both during the sub-computation computing the $i$-th prime
$p_i$ (using Lemma~\ref{lemma PSPACE})
as well as during the sub-computation computing
$B\text{ mod } p_i$ (using Lemma~\ref{lemma modulo})
one can guarantee that the $\widehat{v}$-values
associated with the bit clocks in $\mathcal{J},\mathcal{W},\mathcal{J}'$ and 
$\mathcal{W}'$ are never being modified.

\medskip
{\bf Realization of the remaining lines of the program in Figure~\ref{fig serializability}.}
Lines (4), (8) and (11) can be done directly by the control states of $\A$.
Line (5) boils down to resetting both $x$ and $y$
simultaneously.
Line (6) will be done by checking if for the second time ever (the first time was when $v(x)=v(y)=0$) we have that
$\bin_{2^{s(n)}}(B)=0^{2^{s(n)}}$, which in turn
can be done analogously (but in fact simpler) as our above-mentioned implementation of line (7). 
Line (9) is letting time elapse till the parametric clock $y$ has value $1$ (i.e. one time unit after it had value
$N$ and was reset), 
and then resetting it.
The latter implementation indeed correctly 
implements incrementation modulo $N$.

\section{From two-parametric timed automata with one parameter
to parametric one-counter automata}\label{ptatopoca section}
\renewcommand{\mod}{\mathsf{mod}}
\newcommand{\values}{{\textsc{Values}}}

Being introduced by Bundala and Ouaknine in \cite{BundalaO17}, 
we define parametric one-counter automata.
These are automata that can manipulate a counter that can be incremented or decremented, 
parametrically or not, compared against constants or parameters,
and with divisibility tests modulo constants. 
It is worth mentioning that the notion of parametric one-counter automata 
from \cite{BundalaO17} is slightly more expressive 
than ours, as we shall discuss further below.

After introducing parametric one-counter automata 
we state a theorem (Theorem~\ref{theorem ptatopoca}), proven essentially already in 
\cite{BundalaO17} --- again, however for a slightly more expressive model of parametric one-counter automata --- 
that states that {\sc $(2,1)$-PTA-Reachability} 
can be reduced in exponential time to the reachability 
problem of parametric one-counter automata over one parameter.
Since the actual proof of Theorem~\ref{theorem ptatopoca}
follows the approach of Bundala and Ouaknine from \cite{BundalaO17},
it can be found in the Appendix.

\subsection{Parametric one-counter automata}\label{poca def section}

Given a set of parameters $P$
we denote by $\Op(P)$ the {\em set of operations} over the set of parameters $P$, being of the form
$\Op(P)=\Op_{\pm}\cup\Op_{\pm P}\cup\Op_{\mod\ \!\!\N}\cup\Op_{\bowtie\N}\cup\Op_{\bowtie P}$,
		where
		\begin{itemize}
			\item $\Op_{\pm}=\{-1,0,+1\}$,
			\item $\Op_{\pm P}=\{+p,-p \mid p \in P\}$,
			\item $\Op_{\mod\ \!\!\N }=\{\mod\ c\mid c\in \N\}$, 
			\item $\Op_{\bowtie\N}=\{\bowtie c\mid\bowtie\in\{<,\leq,=,\geq,>\},c\in\N\}$, and
			\item $\Op_{\bowtie P}=\{\bowtie p\mid\bowtie\in\{<,\leq,=,\geq,>\},p\in P\}$.
		\end{itemize}

The {\em size} $|op|$  of an operation $op$ is defined as
$$
|op|=\begin{cases}\log(c)&\text{if $op = \mod\ c$ or $op=\bowtie c\ \text{with} \ c \in \N$}\\
	1 &\text{otherwise.}
	\end{cases}
$$
 
We denote by {\em updates} those operations that lie in 
$\Op_{\pm}\cup\Op_{\pm P}$ and by {\em tests} those 
operations that lie in 
$\Op_{\mod\ \!\!\N}\cup\Op_{\bowtie\N}\cup\Op_{\bowtie P}$.  
			Previously, such as in \cite{BundalaO17}
			or \cite{haase2012complexity},
			slightly different sets of operations have been used, such
			as operations to increment the counter by a constant
			represented in binary.
		Moreover, Bundala and Ouaknine
\cite{BundalaO17} include for the purpose of their construction some operations of the form $+[0,p]$ 
		that allow to nondeterministically add to the counter a value that lies in $[0,\mu(p)]$,
	where $\mu(p)$ is the parameter valuation of parameter $p$. 
			As we shall show in this section, when
			reducing the reachability problem for parametric timed automata
			with two parametric clocks and one parameter to parametric
			one-counter automata one does not require these $+[0,p]$-transitions.

A {\em parametric one-counter automaton }(POCA for short)
is a tuple $$\C=(Q,P,R, q_{init}, F),$$
where 
\begin{itemize}
        \item $Q$ is a non-empty finite {\em set of control states},
        \item $P$ is a non-empty finite {\em set of  parameters}
                that can take non-negative integer values,
	\item $R\subseteq Q\times\Op(P)\times Q$ is a finite {\em set of  rules},
        \item $q_{init}$ is an {\em initial control state}, and
	\item $F\subseteq Q$ is a {\em set of final control states}.
		\end{itemize}

The {\em size} of $\C$ is defined as
$$
|\C| \ = \ |Q|+|P|+|R|+\sum_{(q,op,q')\in R}|op|.
$$

Let $\Const(\C)$ denote the constants that appear in the operations 
$op\in\Op_{\mod\ \!\!\N}\cup\Op_{\bowtie\N}$ for some operation  $(q,op,q')$ in $R$.
By $\Conf(\C)=Q\times\Z$ we denote the set of {\em configurations} of $\C$.
We prefer however to denote a configuration			
 of $\Conf(\C)$ by $q(z)$ instead of $(q,z)$.

Being slightly non-standard we define configurations to take counter values
over $\Z$ rather than over $\N$ for notational convenience.
This does not cause any loss of generality as we allow guards that enable us to test if the value of 
the counter is greater or equal to zero.

\begin{center}
	\begin{figure}
\includegraphics[width=0.7\textwidth]{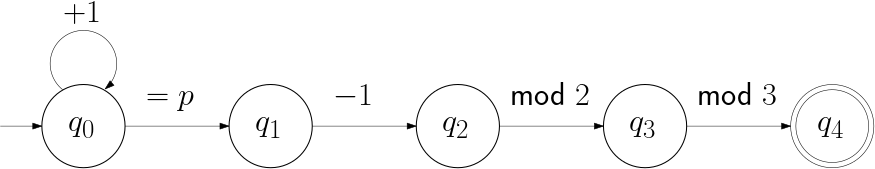}
	\caption{An example of a POCA. The automaton consists of four states and the set of
	parameters is $\{p\}$. The edges are represented by arrows labeled with the corresponding
	operations. 
		A parameter valuation $\mu:\{p\}\rightarrow\N$ 
		witnesses
		that reachability holds for
		the above POCA if, and only, if 
		$\mu(p)\equiv 5\text{ mod }6$.
		}
	\label{example poca}
	\end{figure}
\end{center}

\begin{samepage}

\begin{definition}[transition]{\label{def poca}}
	For every $op\in\Op(P)$, for every parameter valuation 
	$\mu:P\rightarrow\N$, for every POCA $\C$, and for every two configurations  $q(z)$ and $q'(z')$ in $\Conf(\C)$
	we define the {\em transition} $q(z)\xrightarrow{op,\mu}q'(z')$ if 
        there exists some $(q,op,q')\in R$ such that either of the following holds
\begin{enumerate}[(1)]
        \item $op=c \in\Op_{\pm}$ and $z'=z+c$, 
        \item $op\in\Op_{\pm P}$, and either 
                \begin{itemize}
                        \item $op=+p$ and $z'=z+\mu(p)$, or
                        \item $op=-p$ and $z'=z-\mu(p)$.
                \end{itemize}
	\item $op=\mod \ c\in\Op_{\mod\ \!\!\N}$, $z=z'$ and $z'\equiv 0\ \mod \ c$,
        \item $op=\bowtie c\in\Op_{\bowtie \N}$, $z=z'$ and $z'\bowtie c$, and 
         \item $op=\bowtie p\in\Op_{\bowtie P}$, $z=z'$ and $z'\bowtie\mu(p)$.
\end{enumerate}

\end{definition}

\end{samepage}

\newcommand{\I}{\mathcal{I}}

Let $\mu:P\rightarrow\N$ be a parameter valuation.
A {\em $\mu$-run} in $\C$ (from $q_0(z_0)$ to $q_n(z_n)$) is a sequence, possibly empty (i.e. $n=0$), of the form 
$$
\pi \quad =\quad q_0(z_0) \ \xrightarrow{op_0,\mu} \ q_1(z_1)  
\quad \cdots \quad \xrightarrow{op_{n-1},\mu} \ q_n(z_n)\quad.
$$

We say $\pi$ is {\em accepting} if $q_0 = q_{init}$, $z_0=0$,
and $q_n\in F$. 
We say {\em reachability holds} for the POCA $\C$ if
there exists an accepting $\mu$-run for some $\mu\in\N^P$. 
We refer to Figure~\ref{example poca} for an instance of a POCA for which reachability
holds.
For any two $c,d\in[0,n]$ we define
the {\em subrun $\pi[c,d]$ from $q_c(z_c)$ to $q_d(z_d)$} as the 
$\mu$-run
	$q_c(z_c)\ \xrightarrow{\pi_c,\mu} q_{c+1}(z_{c+1})\ \cdots\ \xrightarrow{\pi_{d-1},\mu} \  q_d(z_d).$
As expected, a {\em prefix} (resp. {\em suffix}) of $\pi$ is 
a $\mu$-run of the form $\pi[0,c]$
(resp. $\pi[d,n]$).

We define the {\em concatenation $ \pi_1 \pi_2$ of 
two $\mu$-runs $\pi_1$ and $\pi_2$} when the source configuration of $\pi_2$ is equal to the target configuration of $\pi_1$
as expcected.

We define $\Delta(\pi)=z_n-z_0$ as the {\em counter effect} of 
 the run $\pi$ and for each $i\in[0,n-1]$ let
$\Delta(\pi,i) = \Delta(\pi[i,i+1])$ to denote the 
counter effect of the $i$-th transition of $\pi$.
Its {\em length} is defined as $|\pi|=n$.

In the particular case where $P=\{p\}$ is a singleton for some parameter $p$ and $\mu(p)=N$, we prefer to write 
$q(z)\xrightarrow{op,N}q'(z')$
to denote $q(z)\xrightarrow{op,\mu}q'(z')$ and prefer to call a $\mu$-run an {\em $N$-run}.

We define $\values(\pi)=\{z_i\mid i\in[0,n]\}$
to denote the {\em set of counter values} of the configurations
of $\pi$.
We define a run $\pi$'s {\em maximum} as $\max(\pi)=\max(\values(\pi))$ and
 the {\em minimum} as $\min(\pi)=\min(\values(\pi))$. \newline

 The following theorem states an 
 exponential time reduction
 from {\sc $(2,1)$-PTA-Reachability}
 to the reachability problem of particular
 parametric one-counter automata over one parameter.

\begin{samepage}
\begin{theorem}\label{theorem ptatopoca} 
        The
        following is computable in exponential time:

	INPUT: A $(2,1)$-PTA $\A$.
	
	OUTPUT: A POCA $\C$ over one parameter\\
	such that
	\begin{enumerate}
\item for all $N \in \N$
all accepting $N$-runs $\pi$ in $\C$ satisfy 
$\values(\pi)\subseteq[0,4 \cdot \max(N,| \C|)]$, and
\item reachability holds for $\A$ if, and only if, 
	reachability holds for $\C$.
	\end{enumerate}
\end{theorem}
\end{samepage}

The proof of Theorem~\ref{theorem ptatopoca} can be found in Appendix~\ref{appendix ptatopoca}.

\section{Upper bounds}\label{upper section}

In this section we state the Small Parameter Theorem (Theorem~\ref{theorem upper}) which tells us that
for every POCA over one parameter and every
sufficiently large parameter value $N$,
accepting $N$-runs with counter values all in $[0,4N]$ 
can be turned into accepting $N'$-runs for some smaller $N'$.
After having stated the theorem we will show that 
together with Theorem~\ref{theorem ptatopoca} it implies
an $\EXPSPACE$ upper bound for 
{\sc $(2,1)$-PTA-Reachability}.

We provide an overview of the proof of the Small Parameter Theorem 
in Section~\ref{sec overview}, whose actual proof
will stretch over Sections~\ref{semirun section}, \ref{hill section}, \ref{5-6 section},
and \ref{application section}. 
\bigskip

\noindent
\begin{samepage}
For each POCA $\C=(Q,P,R,q_{init}, F)$ we define the following 
constants:\newline
\fbox{\parbox[t][2.2cm][c]{8cm}{
\begin{eqnarray*}
	\macro_{\C}&=&\LCM(\Const(\C))\\
	\Gamma_{\C}& =& \LCM(17 \cdot |Q|) \cdot \macro_{\C}\\
	\Upsilon_{\C}&=&17 \cdot |Q| \cdot \LCM(17 \cdot |Q|)\cdot
				\left(17 \cdot |Q|\cdot\macro_\C+2\right)\\
	M_{\C}& =& 30 \cdot (\Upsilon_{\C} + \Gamma_{\C}+1)
	\label{constant definitions}
\end{eqnarray*}
	}
}
\end{samepage}

\medskip

\noindent
Since for every non-empty finite set 
$U \subseteq \N\setminus\{0\}$ we have 
$\LCM(U) \leq \max(U)^{|U|}$, the following lemma is straightforward.

\begin{lemma}\label{constants bounded}
	The above constants are asymptotically bounded by 
	$2^{\mathsf{poly}(|\C|)}$.
\end{lemma}

The main result of this section is the following theorem.

\begin{theorem}[Small Parameter Theorem]\label{theorem upper}
	Let $\C=(Q,\{p\},R, q_{init}, F)$ be a POCA with one parameter $p$.
        If there exists an  
accepting
	$N$-run in $\C$ with values all in $[0, 4 N]$ 
	for
	some $N>M_{\C}$, then there exists an
accepting $(N-\Gamma_{\C})$-run in $\C$.
\end{theorem}

Let us first establish that this theorem is enough to prove the desired $\mathsf{EXPSPACE}$ upper bound.

\begin{samepage}
\begin{corollary}\label{corollary}
	{\sc$(2,1)$-PTA-Reachability} is in 
	$\mathsf{EXPSPACE}$.
\end{corollary}

\begin{proof}
	Given a $(2,1)$-PTA $\A$, we apply 	
	Theorem~\ref{theorem ptatopoca} 
        and translate $\A$ in exponential time into 
	a POCA $\C=(Q,P,R,q_0,F)$ with $P=\{p\}$,
	such that
\begin{enumerate}
\item 
all accepting $N$-runs $\pi$ in $\C$ 
	satisfy 
$\values(\pi)\subseteq[0, 4 \cdot \max(N,|\C|)]$,  and
\item
 reachability holds for $\A$ if, and only if, 
	reachability holds for $\C$.
\end{enumerate}

	We first claim that if there exists an accepting 
	$N$-run $\pi$ for $\C$, then there exists one 
	satisfying 
	$N\in[0,\max\{M_\C,|\C|\}]$ and 
	$\val(\pi)\subseteq[0,4\cdot\max\{M_\C,|\C|\}]$.
	All accepting $N$-runs $\pi$ of $\C$
	satisfy $\val(\pi)\subseteq[0,4\cdot\max\{N,|\C|\}]$
	by Point 1,
	so if $N>\max\{M_\C,|\C|\}$, then 
	$4N=4\cdot\max\{N,|\C|\}$ and hence there exists
	some accepting $(N-\Gamma_\C)$-run for $\C$
	by Theorem~\ref{theorem upper}.
	Remarking that
	in case $N>\max\{M_\C,|\C|\}$ we have
	$N-\Gamma_\C>M_\C-\Gamma_\C>0$,
	one can repeat the above argument for 
	$N-\Gamma_\C$ and possibly for $N-2\Gamma_\C$ and 
	so on, thus implying the desired existence.

	Thus by Point 2 it suffices to check in 
	exponential space
	in $|\A|$ whether there exists some accepting
	$N$-run $\pi$ for $\C$
	satisfying $\values(\pi)\subseteq[0,4N]$
	for some $N\in[0,\max\{M_\C,|\C|\}]$.
	Since $M_\C\in2^{\poly(|\C|)}=2^{2^{\poly(|\A|)}}$,
	the latter is simply a 
        reachability question in a doubly-exponentially large 
	finite graph all of whose vertices and edges 
	can be represented using exponentially many bits, and 
	thus decidable in exponential space.

\end{proof}

\end{samepage}

\subsection{Overview of the proof
of the Small Parameter Theorem}\label{sec overview}
For the proof of the Small Parameter Theorem
(Theorem~\ref{theorem upper}) we proceed as follows.
\begin{itemize}
\item In Section~\ref{semirun section} we introduce the 
	notion of $N$-semiruns.
		These generalize $N$-runs in that 
		only modulo tests need to hold,
		not however comparison tests.
		We define some natural operations on them, 
		like shifting them by some value or 
		cutting out certain infixes.
		In Subsection~\ref{bracket section} we prove two important lemmas on semiruns that will serve as base tools for 
		subsequent steps in the proof:
\begin{itemize}
\item The Depumping Lemma (Lemma~\ref{lemma zero}) will be our main tool to depump certains semiruns, in the following sense:
in case
 the difference between the number of $+p$-transitions and 
		$-p$-transitions is bounded for all infixes 
		and equal to $0$ for the whole semirun
		and furthermore the absolute
		counter effect of the semirun is
		sufficiently large, 
then one can build --- by 
		applying the above-mentioned operations ---
a new semirun whose absolute counter 
		effect is slightly smaller.
\item The Bracket Lemma (Lemma~\ref{bracket lemma}) states
	that in case the counter effect is sufficiently large and the counter values are all in $[0, 4 N]$,
then one can find an infix where the counter effect 
		is also large and moreover the difference 
		between the number of $+p$-transitions and 
		$-p$-transitions is bounded for all
		infixes and equal to $0$ for the whole semirun.
\end{itemize}

\item
In Section~\ref{hill section} we introduce the notion of hills
and valleys.
Hills are $N$-semiruns that start and end in configurations with 
		low counter values but where all intermediate configurations have counter values above the source 
		and target configuration. 
We introduce the dual notion of valleys.
The main contribution of the section is the following.
\begin{itemize}
\item The Hill and Valley Lemma (Lemma~\ref{lemma hill/valley})
	allows to transform
		$N$-semiruns that are hills (resp. valleys) into $(N-\Gamma_\C)$-semiruns with 
		the same source and target configuration.
\end{itemize}

\item
	Making use of all of the above 
		lemmas, we introduce
	in Section~\ref{5-6 section} the following lemma,
		which is a main technical ingredient
		in the proof of Theorem~\ref{theorem upper}.
\begin{itemize}
\item The $5/6$-Lemma (Lemma~\ref{lemma 5/6})
	states that $N$-semiruns with counter effect 
		smaller than $5/6 \cdot N$ 
		can be turned in into
		$(N-\Gamma_\C)$-semiruns.
\end{itemize}

\item
Finally, in Section~\ref{application section} 
we prove the Small Parameter Theorem 
(Theorem~\ref{theorem upper})
by carefully factorizing a potential $N$-run into 
subsemiruns that can be treated by the 
above lemmas.

\end{itemize}

\begin{center}
	\begin{figure}
		\hspace{2.7cm}
\includegraphics[width=0.5\textwidth]{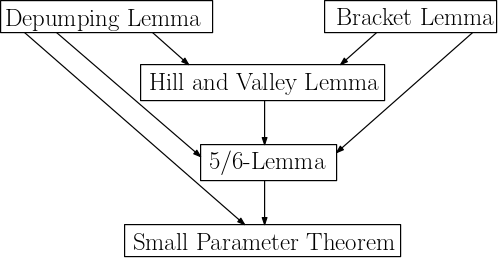}
	\caption{Illustration of the dependencies between the lemmas. The presence of an arrow going from a lemma to another means that the lemma in question is used inside the proof of the lemma the arrow points to.}
\label{figure orga}
	\end{figure} 
\end{center}

In Figure~\ref{figure orga} we give an overview 
of the dependencies of the above-mentioned lemmas.

\section{Semiruns, their bracket projection, and embeddings}\label{semirun section}

In this section we motivate and introduce the notion of semiruns
by loosening the conditions on runs, and
define basic operations on them.
These basic operations possibly change 
their counter values, length, or counter effect. 

The formalism of an $N$-run is a little bit too restrictive to define operations on them.
For instance, substracting $\macro_\C$  from all counter values
of an
$N$-run produces an object, where conditions (1),(2), and (3) of Definition~\ref{def poca} indeed hold --- as $\macro_\C = \LCM(\Const(\C))$ --- but where conditions (4) and (5) 
might not hold anymore, as comparison guards may be violated. 
Rather than certifying each time that
the application of an operation preserves
the property of being an $N$-run
we prefer to loosen the definition
in order to avoid tedious case distinctions.
This motivates the notion of 
semitransitions (resp. semiruns),
which are a generalization of transitions (resp. runs), 
in which the comparison tests need not to hold.

We introduce semiruns and operations on them in Section~\ref{subsection semiruns}.
Section~\ref{bracket section} introduces the bracket projection of
semiruns, the Depumping Lemma (Lemma~\ref{lemma zero}) and the
Bracket Lemma (Lemma~\ref{bracket lemma}).
Section~\ref{section embeddings} introduces the notion of 
embeddings, which provide a formal means to express
when a semirun can structurally be found as a subsequence
of another.
\subsection{Semiruns and operations on them}\label{subsection semiruns}
\begin{samepage}
	\begin{definition}[semitransition]\label{semitransitions}
Let $\C=(Q,P,R,q_{init}, F)$ be a POCA.
		For every operation 
		$op\in\Op(P)$, for every
		parameter valuation $\mu:P\rightarrow\N$,
 and for every two configurations $q(z)$ and $q'(z')$ in $\Conf(\C)$
 we define the {\em semitransition}
$q(z) \semi{op,\mu} q'(z')$  if 
        there exists some $(q,op,q')\in R$ such that 
		conditions (1),(2), and (3) 
		of Definition
		\ref{def poca} hold but
		where conditions (4) and
(5) are loosened
by the following conditions (4') and (5') respectively
\begin{enumerate}[(1)]
        \item $\ op=c \in\Op_{\pm}$, and $z'=z+c$, 
        \item $\ op\in\Op_{\pm P}$ and either 
                \begin{itemize}
                        \item $op=+p$ and $z'=z+\mu(p)$, or
                        \item $op=-p$ and $z'=z-\mu(p)$.
                \end{itemize}
	\item $\ op=\text{mod} \ c\in\Op_{\text{mod}\ \!\!\N}$, $z=z'$ and $z'\equiv 0\ \text{mod} \ c$,
\end{enumerate} \begin{enumerate}[(1')]
  \setcounter{enumi}{3}
	\item \!\!$op=\bowtie c\in \Op_{\bowtie\N}$ and $z=z'$, and
	\item $op=\bowtie p\in\Op_{\bowtie P}$, and $z=z'$.
\end{enumerate}
\end{definition}
\end{samepage}

Thus, in a nutshell,  when writing $q(z) \semi{op,\mu} q'(z')$ we do not require that the 
comparison tests against parameters or against
constants
 hold; however the updates and the modulo tests against constants 
must 
be respected.
This naturally gives rise to the definition of {\em $\mu$-semiruns} as expected.
Note that in particular every $\mu$-run is a $\mu$-semirun. 
The abbreviation $N$-semirun, $q(z)\semi{op,N}q'(z')$, the counter effect $\Delta$, $\values$,
$\min$, $\max$, subsemirun, prefix, suffix are defined as for runs.

Note that in particular every $N$-run is an $N$-semirun.
Importantly, note also that
semitransitions involving comparison tests
are still syntactically present in semiruns.
By a careful analysis, 
one can therefore possibly perform operations on 
$N$-semiruns in order to show that they 
are in fact $N$-runs.\newline

\begin{example}
The $2$-semirun 
$$
\pi \ = \ q_0(0) \ \semi{+1,2} \ q_1(1) \ \semi{+1,2} \ q_1(2) \ \semi{+1,2} \ q_{1}(3) \ \semi{\leq p,2} \ q_2(3) \ \semi{\mod\ 3,2} \ q_3(3)
$$
is not a $2$-run, as, in 
	$q_{1}(3) \semi{\leq p,2} q_2(3)$,
	condition (4) of Definition~\ref{def poca} 
	does not hold, 
	however condition (4') of Definition~\ref{semitransitions} does. \newline
\end{example}

\subsection*{Shifting and gluing of semiruns}

\noindent
Let us fix a POCA $\C$ and some $N$-semirun
\begin{eqnarray*}
	\pi  \ = \  q_0(z_0) \ \semi{\pi_0,N} \  q_1(z_1) \  \cdots
	 \  \semi{\pi_{n-2},N} \  q_{n-1}(z_{n-1}) \  \semi{\pi_{n-1},N} \ 
	q_n(z_n).
\end{eqnarray*} 

\noindent
We define the following operations, where we recall
that $\macro_\C=\LCM(\Const(\C))$:
\begin{itemize}
\item For $D \in \macro_\C \Z$, we define the {\em shifting} of $\pi$ by $D$ as
\begin{eqnarray*}
	\pi+D  \ = \  q_0(z_0+D) \ \semi{\pi_0,N} \  q_1(z_1+D) \ \cdots \  
	\semi{\pi_{n-1},N} \  q_n(z_n+D)\ .
\end{eqnarray*}
\noindent
Since there are no effective comparison tests and $D$ is an integer that is divisible by all constants appearing in modulo tests in $\C$, it is clear that $\pi+D$ is again an $N$-semirun.

\item For two configurations $q_i(z_i)$ and $q_j(z_j)$ with $0\leq i<j\leq n$ and
where $D=z_j-z_i\in\macro_\C \Z$ is a multiple of $\macro_\C $ and $q_i=q_j$, we define the {\em gluing} of the configurations as
\begin{eqnarray*}
\pi-[i,j]  \ = \  q_0(z_0)  \  \ \cdots
 \ \semi{\pi_{i-1},N}
 \ 
 q_i(z_i) \  \semi{\pi_{j},N} \  q_{j+1}(z_{j+1}-D) 
	 \ \cdots \ \semi{\pi_{n-1},N} \  q_n(z_n-D).
\end{eqnarray*}

\end{itemize}

When gluing the leftmost and rightmost configurations of pairwise non-intersecting intervals $I_1=[a_1,b_1],\ldots, I_k=[a_k, b_k] \subseteq[0,n]$,
assuming $b_i < a_{i+1}$ for all $1 \leq i < k$, and $q_{a_i}=q_{b_i}$ and $z_{b_i} - z_{a_i} \in \macro_\C \Z$ for all $1 \leq i \leq k$,
we will use $\pi-I_1-I_2\cdots- I_k$ to denote the result
corresponding to gluing each interval successively while shifting the others accordingly, 
instead of writing the more tedious
 $\pi^{(k)}$, where
\begin{eqnarray*}
	\pi^{(1)} &=&  \pi - [a_1,b_1],\\ 		
	\pi^{(2)} &=&  \pi^{(1)} - [a_2 - (|I_1|-1),
	b_2 - (|I_1|-1)],\\ 
		&\cdots&\\							\pi^{(k)} &=& \pi^{(k-1)} - [a_k - 
		\sum_{1 \leq j < k} (|I_j|-1), 
					b_k - 
		\sum_{1 \leq j < k} (|I_j|-1) ]\ . 
\end{eqnarray*}

\subsection{The bracket projection of semiruns}\label{bracket section}

In this section we define a projection $\phi$ of 
semitransitions
$\tau = q(z) \semi{op,N} q'(z')$ to a word over the binary 
alphabet $\{\boldsymbol{[},\boldsymbol{]}\}$, 
where transitions with $op=+p$ are mapped to $\boldsymbol{[}$,  transitions with $op=-p$ are mapped
to $\boldsymbol{]}$, 
and all other transitions are mapped to the empty word $\varepsilon$.
The projection $\phi$ is naturally extended
to a morphism from semiruns to $\{\boldsymbol{[},\boldsymbol{]}\}^*$.
In this section we will show the following lemmas.
\begin{itemize}
\item The Depumping Lemma (Lemma~\ref{lemma zero})
	states that for each $N$-semirun whose
		$\phi$-projection has bounded
		bracketing properties and that has
		a counter effect whose absolute value is 
		sufficiently large 
		there exists another
		$N$-semirun with a counter 
		effect whose absolute value is slightly smaller.
		This latter resulting $N$-semirun
		has a particular form in that it
		can be obtained from the original $N$-semirun
		by applying the above-mentioned
		operations of shifting and gluing: notably,
		the subsemiruns that are being glued themselves
		have a $\phi$-projection that has
		bounded bracketing properties.
\item The Bracket Lemma (Lemma~\ref{bracket lemma}) 
	states that if an $N$-semirun 
		has all its counter values in $[0, 4 N]$,
has an absolute counter effect that is sufficiently large
		and has a $\phi$-projection 
		satisfies a suitable threshold
		condition on the number of occurrences
		of $\boldsymbol{[}$ and $\boldsymbol{]}$,
		that there is a subsemirun where the absolute
		counter effect is 
		also large and whose $\phi$-projection
		has bounded bracketing properties.
\end{itemize}

\noindent
Formally, we define a mapping $\phi$ such that
for every semitransition 
$\tau = q(z) \semi{op,N} q'(z')$,
$$
\phi(\tau) \ = \ 
\begin{cases}
	\ \boldsymbol{[} & \text{if $op=+p$}\\
	\ \boldsymbol{]} & \text{if $op=-p$}\\
	\ \varepsilon & \text{otherwise}.
\end{cases}
$$

\noindent  
Note that an $N$-semirun $\pi$ can contain several 
$+p$-transitions and $-p$ transitions.
We introduce the notation $\phi(\pi,i) = \phi(\pi[i,i+1])$ to denote the $\phi$-projection of the $i$-th transition of $\pi$ for all $i \in [0,|\pi|-1]$. The mapping $\phi$ is naturally extended to a morphism on semiruns to words over the binary 
alphabet $\{\boldsymbol{[},\boldsymbol{]}\}$ as 
expected:
$\phi(\pi) = \phi(\pi, 0) \phi(\pi, 1) \ \cdots \ \phi(\pi, |\pi|-1)$.\newline

We are particularly interested in $N$-semiruns whose 
projection by $\phi$
contains as many opening as closing brackets and 
only a few pending ones (when read from left to right).
To make this formal, for all $k \in\N$ we define 
the regular language
\begin{eqnarray*}
	\Lambda_k=\left\{w\in\{\boldsymbol{[},\boldsymbol{]}\}^* 
	: |w|_{\boldsymbol{[}}=
	|w|_{\boldsymbol{]}}, \forall u,v \in 
	\{\boldsymbol{[},\boldsymbol{]}\}^*. \ 
	uv=w \implies |u|_{\boldsymbol{[}}-|u|_{\boldsymbol{]}}
	\in[-k, k]\right\}.
	\end{eqnarray*}

We are interested in analyzing $N$-semiruns with counter 
values in $[0, 4 N]$. 
Bounding the counter values like this limits
the number of $+p$ (resp. $-p$) that can appear in a row. This will be the basis in the Bracket Lemma which amounts
to showing the existence of subsemiruns 
whose $\phi$-projection is in $\Lambda_{8}$.

The now following Depumping Lemma 
will enable us to reduce the counter effect of 
$N$-semiruns whose $\phi$-projection is in $\Lambda_8$. 
It is worth remarking that $\Gamma_\C \ll \Upsilon_\C$,
recalling the definition of our constants on 
page~\pageref{constant definitions}.

\begin{lemma}[Depumping Lemma]\label{lemma zero}
	For all $N$-semiruns $\pi$ satisfying 
	$\phi(\pi)\in\Lambda_{8}$
	and $|\Delta(\pi)|>\Upsilon_\C$
	there exists an $N$-semirun $\pi'$ such that either
	\begin{itemize}
	\item $\Delta(\pi)>\Upsilon_\C$ and $\Delta(\pi')=\Delta(\pi)-\Gamma_\C$, or
	\item  $\Delta(\pi)<-\Upsilon_\C$ and $\Delta(\pi')=\Delta(\pi)+\Gamma_\C$.
	\end{itemize}
	Moreover, $\pi'=\pi-I_1-I_2 \ \cdots \ -I_k$ for pairwise 		\textit{disjoint} 
							intervals $I_1,\ldots,I_k\subseteq [0,|\pi|]$
	such that we have 
							$\phi(\pi[I_i])\in\Lambda_{16}$
					for all $i\in[1,k]$, 
							and either $\Delta(\pi[I_i])>0$ 
							for all $i\in[1,k]$ or
							$\Delta(\pi[I_i])<0$ for all $i\in[1,k]$.
\end{lemma}

 \begin{proof}
Let 
 $\pi = q_0(z_0) \semi{\pi_0,N} q_1(z_1) \semi{\pi_1,N} \quad \cdots\quad
\semi{\pi_{n-1},N } q_n(z_n) $
	be an $N$-semirun
such that $\phi(\pi) \in \Lambda_{8}$. 
We will assume without loss of generality that $\Delta(\pi) > \Upsilon_\C$. 
	The dual case when
	$\Delta(\pi) < - \Upsilon_\C$ can be proven
	analogously. 

\newcommand{\pot}{\mathsf{pot}}
\noindent
	For every position $i\in[0,n]$ 
let us define
$$\lambda(i)=|\phi(\pi[0,i])|_{\boldsymbol{[}}
-
|\phi(\pi[0,i])|_{\boldsymbol{]}}
\qquad\text{and }\qquad
\pot(i)=z_i-z_0-\lambda(i)\cdot N\qquad.
$$
Note that since $\phi(\pi)\in\Lambda_{8}$
we have for all $i\in[0,n]$,
\begin{eqnarray}
	\lambda(i)\in[-8,8],\qquad\qquad  \label{Lambda lambda}
\end{eqnarray}
and moreover
\begin{eqnarray}
	\phi(\pi[0,i])\in\Lambda_{8}
 \Longleftrightarrow\lambda(i)=0.
\end{eqnarray}

\noindent
	\begin{samepage}
We note the following important properties of $\pot$,
\begin{enumerate}
	\item $|\pot(i-1) - \pot(i)| \leq 1$ for all $i\in[1,n]$,
	\item $\pot(0)=0$,
	\item for all $0\leq i<j\leq n$,
		if $\lambda(i)=\lambda(j)$, then $\pot(j)-\pot(i)=z_j - z_i$, and
	\item $\pot(n)=z_n-z_0=\Delta(\pi)$ since $\lambda(0)=\lambda(n)=0$.
\end{enumerate}
	\end{samepage}

	The following claim states that
	if in a subsemirun the $\pot$ increases
	sufficiently large, then one can
	find a subsemirun therein that can potentially be glued.
\begin{claim}
For each subsemirun $\pi[a,b]$ 
that satisfies 
	$\pot(b)-\pot(a)> 17\cdot|Q|\cdot\macro_\C$
there exist positions 
$a\leq s<t\leq b$,
such that 
\begin{itemize}
	\item $q_{s}=q_{t}$,
	\item $\lambda(s)=\lambda(t)$, and
	\item $z_{t}-z_s=d\macro_\C$
		for some $d\in[1,17 \cdot |Q|]$.
\end{itemize}
\end{claim}
\begin{proof}[Proof of the Claim.]
	Since 
	by assumption 
	$\pot(b)-\pot(a)>
	17 \cdot |Q|\cdot\macro_\C$,
by the pigeonhole principle and Point $1$ above,
	there exist two indices 
	$a\leq s<t\leq b$ such that
	$q_{s}=q_{t}$, 
	$\lambda(s)\in[-8,8] $ and $\lambda(t)\in[-8,8]$ are equal,
	and $\pot(t)-\pot(s)=d\macro_\C$
	for some $d\in[1,17 \cdot |Q|]$.
	By Point 3 above, from $\lambda(t)=\lambda(s)$, 
	it follows $z_{t}-z_s=\pot(t)-\pot(s)=d\macro_\C$.\\
	\noindent
	{\em (End of the proof of the Claim)}
\end{proof}

\noindent
Since $\pot(i)-\pot(i-1)\leq 1$ for all $i\in[1,n]$
by Point 1 above and
\begin{eqnarray*}
	\pot(n)-\pot(0) &=& z_n - z_0   \\
			&=& \Delta(\pi) \\
			&>& \Upsilon_\C \\
			&\stackrel{\text{page~\pageref{constant definitions}}}{=}&  
			17 \cdot |Q| \cdot \LCM(17 \cdot |Q|)\cdot
				\left(17 \cdot |Q|\cdot\macro_\C+2\right),
\end{eqnarray*}
by the pigeonhole principle, there exist at least 
$$
17 \cdot |Q| \cdot \LCM(17 \cdot |Q|)
$$
pairwise disjoint subsemiruns
$\pi[a,b]$ satisfying $\pot(b)-\pot(a)>
17\cdot|Q|\cdot\macro_\C
$.
Let $$L=  \LCM(17 \cdot |Q|),$$
and let 
$\pi[a_1,b_1],\ldots,\pi[a_{17 \cdot |Q| \cdot L},b_{17 \cdot |Q| \cdot L}]$ be an enumeration
of these latter subsemiruns.
We apply the above Claim to all of these $\pi[a_i,b_i]$:
there exist positions $a_i\leq s_i\leq t_i\leq b_i$
such
that $\lambda(s_i)=\lambda(t_i)$, $q_{s_i}=q_{t_i}$,
and $z_{t_i}=z_{s_i}+d_i\macro_\C$ for some $d_i\in[1,17 \cdot |Q|]$.
From $\lambda(s_i)=\lambda(t_i)$ and (\ref{Lambda lambda})  it follows
$\phi(\pi[s_i,t_i])\in\Lambda_{16}$.
Recall that $\Gamma_\C = \LCM(17 \cdot|Q|) \cdot \macro_\C = L \cdot \macro_\C$,
cf. page~\pageref{constant definitions}.
By the pigeonhole principle, among
these $17 \cdot |Q| \cdot L$ pairwise disjoint subsemiruns $\pi[a_i,b_i]$,
there exists some $d\in[1, 17 \cdot |Q|]$ such that there are 
$L/d$ many different $\pi[a_i,b_i]$ all satisfying $d_i=d$.
Let $\pi[a_{i_1},b_{i_1}],\ldots,\pi[a_{i_{L/d}},b_{i_{L/d}}]$
be an enumeration of these latter $\pi[a_i,b_i]$.
Note that for all of these $\pi[a_i,b_i]$ 
we have $\Delta(\pi[s_{i_j},t_{i_j}])=d\cdot\macro_\C$.
Since moreover $q_{s_{i_j}}=q_{t_{i_j}}$
we know that, 
for all $j\in[1,L/d]$,
the gluing 
$\pi-[s_{i_j},t_{i_j}]$
is an $N$-semirun with 
$\Delta(\pi-[s_{i_j},t_{i_j}])
=\Delta(\pi)-d \macro_\C$.
Thus, 
$$
\pi'\quad=\quad\pi-[s_{i_1},t_{i_1}]-\ldots-[s_{i_{L/d}},t_{i_{L/d}}]
$$
is an $N$-semirun 
satisfying $\Delta(\pi')=\Delta(\pi)-d\cdot(L/d) \cdot \macro_\C=\Delta(\pi)-\Gamma_\C$
as required.
\end{proof}

Let us now introduce the Bracket Lemma, which states that 
in case the absolute value of the 
counter effect of an $N$-semirun is sufficiently large, 
the counter values are all in $[0, 4N ]$
and a majority condition holds on the
number of occurrences of $\boldsymbol{[}$ and $\boldsymbol{]}$
in its $\phi$-projection, that there is
a subsemirun where the counter effect is also
large and that moreover has good bracketing properties 
(in the sense of the Depumping Lemma).
Roughly speaking, it based on the idea that if the values of a semirun are all
in $[0, 4N ]$, there cannot be five $+p$-transitions in a row.
Technically speaking, the Bracket Lemma 
can be applied to $(N-\Gamma_\C)$-semiruns, where $N$ is sufficiently
large: the reason is that the Bracket Lemma 
will later be applied to $N$-semiruns in which some of
the $+p/-p$-transitions have already been modified (``by hand'')
to have an effect $(N-\Gamma_\C)/-(N-\Gamma_\C)$ instead
of $N/-N$.

\begin{samepage}
\begin{lemma}[Bracket Lemma]\label{bracket lemma}
	For all $N > M_\C$, all
	$(N-\Gamma_\C)$-semiruns $\pi$ satisfying 
	$\values(\pi)\subseteq[0,4N]$,
	$\Delta(\pi)<-\Upsilon_\C$ (resp. $\Delta(\pi)>\Upsilon_\C$)
	and 
	where $\phi(\pi)$
	contains at least as many occurrences of $\boldsymbol{[}$ as occurrences of $\boldsymbol{]}$ 
	(resp. at least as many occurrences of $\boldsymbol{]}$ as occurrences of $\boldsymbol{[}$)
	there exists a subsemirun $\pi[c,d]$ 
	satisfying
	$\phi(\pi[c,d])\in\Lambda_{8}$ and
	$\Delta(\pi[c,d])<-\Upsilon_\C$ (resp. $\Delta(\pi[c,d])>\Upsilon_\C$).
\end{lemma}
 \begin{proof}
	We only prove the case where  $\Delta(\pi)<-\Upsilon_\C$
	and $\phi(\pi)$
	contains at least as many occurrences of $\boldsymbol{[}$ as of $\boldsymbol{]}$.
	The dual case when 
	$\Delta(\pi)>\Upsilon_\C$
	and $\phi(\pi)$
	contains at least as many $\boldsymbol{]}$ as of $\boldsymbol{[}$
	can be proven analogously.

	As in the proof of Lemma~\ref{lemma zero}, 
	for any word $u\in\{\boldsymbol{[},\boldsymbol{]}\}^*$ 
	let
		$\lambda(u)=|u|_{\boldsymbol{[}}-
		|u|_{\boldsymbol{]}}.$
	For the rest of the proof assume by contradiction 
	that there is no
	such subsemirun $\pi[c,d]$ 
	satisfying 
	$\Delta(\pi[c,d])<-\Upsilon_\C$ and
	$\phi(\pi[c,d])\in\Lambda_{8}$,
	or, equivalently, that every subsemirun
	$\pi[c,d]$ with
	$\phi(\pi[c,d])\in\Lambda_{8}$
	satisfies
	$\Delta(\pi[c,d])\geq -\Upsilon_\C$.

	\noindent
	For all $k\geq 0$ let
	$$
		\Psi_k=\{w\in\{\boldsymbol{[},
		\boldsymbol{]}\}^*\mid\forall uv=w:
		\lambda(u)\in[-k,k]\}
	$$
	denote the set of all words over the alphabet
	$\{\boldsymbol{[},\boldsymbol{]}\}$,
	where for each prefix the absolute difference between the 
	number of occurrences of 
	$\boldsymbol{[}$ and of $\boldsymbol{]}$ is at most $k$.
	Note that 
	\begin{eqnarray}\label{psi lambda}
		\Lambda_{k}=\Psi_{k}\cap
		\lambda^{-1}(0).
	\end{eqnarray} 

	Under the above assumptions on $\pi$, for the sake of contradiction,
	we have three claims on properties on the image of $\phi$ applied
	to $\pi$ and subsemiruns thereof.

	\medskip

	\noindent
	{\em Claim 1.}
	$\phi(\pi)\in\Psi_{4}$.
	\noindent

	\begin{proof}[Proof of Claim 1.]
		Let us write $\pi=\pi[0,n]$.
		Assume by contradiction that
	$\phi(\pi)\not\in\Psi_{4}$.
		Let $u$ be a shortest prefix of 
		$\phi(\pi)$ such that
		$\lambda(u)\not\in[-4,4]$.
	Let us first consider the case when
		$\lambda(u)>4$.

		By definition of $u$ 
		we have 
		$\lambda(u)=4+1=5$
		and there are indices 
		$0\leq t_1<\cdots<t_{5}<n$
		such that 
	\begin{itemize}
	\item 	$
		\phi(\pi,{t_1})=
		\ldots=
		\phi(\pi,{t_{5}})=\boldsymbol{[}$, and
	\item $\phi(\pi[t_i+1,t_{i+1}])
		\in\Lambda_{4}$ for all $i\in[1,4]$.
	\end{itemize}

		Recall that by our assumption
		every subsemirun $\pi[c,d]$ of
		$\pi$ with
		$\phi(\pi[c,d])\in\Lambda_{8}$
		satisfies $\Delta(\pi[c,d])\geq-\Upsilon_\C$.
		Since 
		$\bigcup_{i\in[1,8]} \Lambda_i=
		\Lambda_{8}$ it follows
		$\Delta(\pi[t_i+1,t_{i+1}])\geq-\Upsilon_\C$
		for all $i\in[1,4]$.
		Moreover, bearing in mind that $\pi$
		is an $(N-\Gamma_\C)$-semirun, 
		we obtain 
		$\Delta(\pi,{t_i})=N-\Gamma_\C$.
		Altogether, as $N > M_\C$ by assumption, we obtain
		\begin{eqnarray*}
			\Delta(\pi[t_1,t_{5}+1])
			&\geq&-4 \cdot \Upsilon_\C+5\cdot(N- \Gamma_\C)\\
			&>&
			4N+N-5\cdot(\Upsilon_\C+ \Gamma_\C)\\
			&>& 4N+ M_\C -5\cdot(\Upsilon_\C+ \Gamma_\C)\\
			&>&{4N},
		\end{eqnarray*}
	where the last inequality follows from $M_\C$'s definition on page~\pageref{constant definitions},
		hence
		contradicting 
		$\values(\pi)\subseteq[0,4N]$.

	Let us now consider the case
		when $\lambda(u)<-4$.
		Again, by definition of $u$,
		we have 
		$\lambda(u)=-5$.
		There are hence
		indices $0\leq t_1<\ldots<t_{5}<n$
		such that
		$$\phi(\pi,{t_1})=
		\ldots=
		\phi(\pi,{t_{5}})=\boldsymbol{]},$$
		and moreover $\phi(\pi[0,t_1])\in\Lambda_{4}$
		and $\phi(\pi[t_i+1,t_{i+1}])\in\Lambda_{4}$
		for all $i\in[1,4]$.
		By assumption  
		$\phi(\pi)$ contains
		at least as many $\boldsymbol{[}$
		as $\boldsymbol{]}$.
			Therefore 
		there must exist $5$ further positions
		$t_1',\ldots,t_{5}'$ in $\pi$
		satisfying
		$0\leq t_1<\ldots<t_{5}<t_1'<t_2'<\ldots<t_{5}'<n$ 
		such that 
		$$\phi(\pi,{t_1'})=
		\ldots=\phi(\pi,{t_{5}'})=
		\boldsymbol{[}
		$$
		and $\phi(\pi[t_i'+1,t_{i+1}'])\in \Lambda_{4}$ 
		for all $i\in[1,4]$.
		Again taking into account our
		assumption that $\Delta(\pi[c,d])\geq -\Upsilon_\C$ for all subsemiruns
		$\pi[c,d]$ with 
		$\phi(\pi[c,d])\in\Lambda_{8}$,
		it follows as above,
		that $\Delta(\pi[t_1',t_{5}'+1])>4N$,
		contradicting 
		$\values(\pi)\subseteq[0,4N]$.
	\end{proof}

	\noindent
	{\em Claim 2.}
	$\phi(\pi[a,b])\in\Psi_{8}$
	for all subsemiruns $\pi[a,b]$ of $\pi$.

	\noindent
	\begin{proof}[Proof of Claim 2.]
		This is an immediate consequence
		of Claim 1.
		Indeed, any subsemirun
		$\pi[a,b]$ of $\pi$ 
		satisfying 
		$\phi(\pi[a,b])\not\in\Psi_{8}$
		gives rise to a prefix $u$ of
		$\phi(\pi)$ such 
		that $u\not\in\Psi_4$ and hence
		$\phi(\pi)\not\in\Psi_4$.
	\end{proof}

\medskip

	\noindent
	{\em Claim 3.}
	For all subsemiruns $\pi[a,b]$ of $\pi$,
	if $\lambda(\phi(\pi[a,b]))>0$, then $\Delta(\pi[a,b])>\Upsilon_\C$.

	\noindent
	\begin{proof}[Proof of Claim 3.]
		
		We prove the statement by induction on
		$\lambda(\phi(\pi[a,b]))$.

	For the induction base, assume
		$\lambda(\phi([a,b]))=1$.
		Thus, there exists a position $t\in[a,b]$
		such that $\phi(\pi,t)=\boldsymbol{[}$ and
		$\lambda(\pi[a,t])=\lambda(\phi(\pi[t+1,b]))=0$.
		By Claim 2 and (\ref{psi lambda})
		we have $\phi(\pi[a,t]),\phi(\pi[t+1,b])
		\in\Lambda_8$. Thus, $\Delta(\pi[a,t]),\Delta(\pi[t+1,b])>-\Upsilon_\C$ by our assumption.
Hence, we obtain
		\begin{eqnarray*}
			\Delta(\pi[a,b])&=&\Delta(\pi[a,t])+
			\Delta(\pi,t)+\Delta(\pi[t+1,b])\\
			&\geq&
			-\Upsilon_\C+(N-\Gamma_\C)-\Upsilon_\C\\
			&>&M_\C-2\Upsilon_\C-\Gamma_\C\\
			&>&\Upsilon_\C,
		\end{eqnarray*}
		where the last strict inequality
		follows from definition of
		$M_\C$ on page~\pageref{constant definitions}.

	Assume $\lambda(\phi(\pi[a,b]))>1$. 
		Consider the smallest position $t \in [a,b]$
		such that $\lambda(\phi(\pi[a,t]))=0$
		and 
		$\phi(\pi,t)=\boldsymbol{[}$.
		By Claim 2 and (\ref{psi lambda})
		it follows that $\phi(\pi[a,t])\in\Lambda_{8}$ 
		and hence 
		$\Delta(\pi[a,t])\geq-\Upsilon_\C$ by our assumption.
		Moreover, $\lambda(\phi(\pi[t+1,b]))=\lambda(\phi(\pi[a,b]))-1$.
		We can thus apply induction hypothesis to $\pi[t+1,b]$ 
		and obtain 
	\begin{eqnarray*}
		\Delta(\pi[a,b])&=&
			\Delta(\pi[a,t])+\Delta(\pi,t)+ \Delta(\pi[t+1,b])\\
			&>&\Delta(\pi[a,t])+\Delta(\pi,t)+ \Upsilon_\C\\
			&\geq& -\Upsilon_\C+(N-\Gamma_\C)+\Upsilon_\C\\
			&>& {M_\C - \Gamma_\C},\\
			&>& {\Upsilon_\C},
	\end{eqnarray*}
		where the first strict inequality
		follows from induction hypothesis on  $\pi[t+1,b]$  and 
		the last strict inequality follows from 
		definition of $M_\C$ on page~\pageref{constant definitions}.
	\end{proof}

	We will now contradict our initial assumption
	that there is no subsemirun $\pi[c,d]$ satisfying 
	$\phi(\pi[c,d])\in\Lambda_{8}$
	and $\Delta(\pi[c,d])<-\Upsilon_\C$
	by making use of the above claims.

	Since $\pi$ itself satisfies $\Delta(\pi)<-\Upsilon_\C$,
	it follows $\phi(\pi)\not\in\Lambda_{8}=\Psi_8\cap\lambda^{-1}(0)$ by our assumption and (\ref{psi lambda}).
	But since $\phi(\pi)\in\Psi_8$ by Claim 2,
	it follows $\lambda(\phi(\pi)) \neq 0$ .

	As
	$\phi(\pi)$
	contains at least as many occurrences of $\boldsymbol{[}$ as occurrences of $\boldsymbol{]}$
	by assumption,
	 $\phi(\pi)$
	must contain strictly more occurrences
	of $\boldsymbol{[}$ than of $\boldsymbol{]}$, i.e. $\lambda(\phi(\pi))>0$.
	By Claim 3 it follows
	$\Delta(\pi) > \Upsilon_\C$, contradicting our assumption that
	$\Delta(\pi) < - \Upsilon_\C$. 
\end{proof}

\end{samepage}

\subsection{Embeddings of semiruns}\label{section embeddings}

\noindent
The Small Parameter Theorem (Theorem~\ref{theorem upper}) turns 
$N$-runs with values in $[0,4N]$ into $(N-\Gamma_\C)$-runs.
In proving this, we prefer to view $N$-runs as $N$-semiruns.
Indeed, we first view any $N$-run as an $N$-semirun
and then apply certain of the above-mentioned
operations on them to obtain some $(N-\Gamma_\C)$-semirun.
However, we would then like to claim that the resulting $(N-\Gamma_\C)$-semirun is in
fact an $(N-\Gamma_\C)$-run as desired, in particular the comparison tests need to hold.
To do so, we introduce a notion when an
$N$-semirun can be embedded into an $M$-semirun
(possibly $N\not=M$) in the sense
that operations are being preserved, source and target control states are being
preserved, and that with respect to some 
line $\ell\in\Z$ the counter
value of each configuration of the embedding has the same orientation with respect to $\ell$
as the counter value of the configuration it corresponds to.

\begin{definition}[$\ell$-embedding]\label{def embedding}
	Let $\ell\in\Z$.
        An $N$-semirun
        $$
	\sigma \quad = \quad s_0(y_0)\semi{\sigma_0,N}s_1(y_1) \ \cdots \  \semi{\sigma_{n-1},N}s_n(y_n)
	$$ 
	is an {\em $\ell$-embedding}
        of an $M$-semirun
        $$
	\pi\quad =\quad q_0(z_0)\semi{\pi_0,M}q_1(z_1) \ \cdots \ \semi{\pi_{m-1},M} q_m(z_m)
	$$
if $s_0=q_0$, $s_n=q_m$ and  
there exists an order-preserving injective mapping 
			$\psi:[0,n]\rightarrow[0,m]$ such that
                        \begin{itemize}
                                \item $\sigma_i=\pi_{\psi(i)}$ for all $i\in[0,n-1]$, and
				\item $\ell\bowtie y_i$ if, and only if, $\ell\bowtie z_{\psi(i)}\
        \text{for all}\ \bowtie\in\{<,=,>\}\text{ and all }i\in[0,n]$.
        \end{itemize}
	Moreover we say $\sigma$ is 
	\begin{itemize}
		\item {\em max-falling (w.r.t $\pi$)} if $\max(\sigma)\leq\max(\pi)$, and
		\item {\em min-rising (w.r.t. $\pi$)} if $\min(\sigma)\geq\min(\pi)$.
	\end{itemize}
\end{definition}

\begin{center}
	\begin{figure} 
\includegraphics[width=0.5\textwidth]{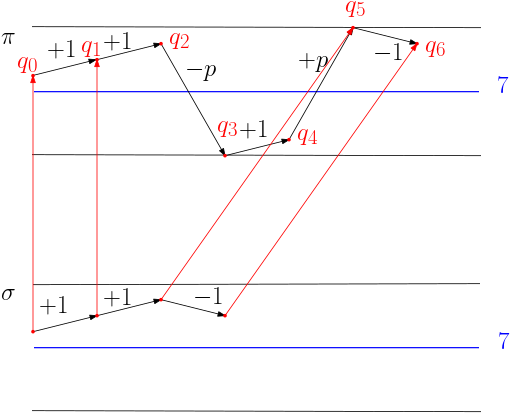} 
\includegraphics[width=0.5\textwidth]{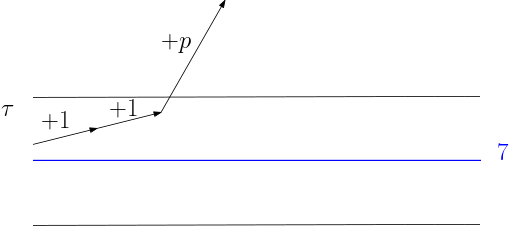}
		\caption{Example of a semirun $\sigma$ that could possibly be an embedding 
		of the semirun $\pi$ and a semirun $\tau$ that cannot.}  
	\label{embedding example}
	\end{figure}
\end{center}

\begin{samepage}
\begin{example}
	Consider the semiruns $\pi,\sigma$ and $\tau$ in Figure~\ref{embedding example},
	where neither 
	concrete counter values nor the control states of 
	$\sigma$ and $\tau$ are mentioned.
	The semirun $\sigma$ can possibly be 
	a $7$-embedding of $\pi$ (if its source control control is $q_0$ and its
	target control state is $q_6$).
	However, $\tau$ cannot be a $7$-embedding of $\pi$.
	Indeed, for every possible $\psi$ such that $\tau_2 = +p = \pi_{\psi(2)}$,
	the counter value of $\tau$ at position $2$ is strictly larger than $7$, whereas
	the counter value of 
	$\pi$ at position $\psi(2)$ is strictly below $7$.
\end{example}
\end{samepage}

The following remark is implictly being used in subsequent sections.

\begin{remark}\label{remark embeddings}
Embeddings possess some useful properties that all follow immediately from definition.
\begin{itemize}
	\item {\em Transitivity. } 
		Let $\pi$, $\rho$ and $\sigma$ be semiruns such that
	$\pi$ is an $\ell$-embedding of $\rho$
	and $\rho$ is an $\ell$-embedding of $\sigma$.
	Then $\pi$ is an $\ell$-embedding of $\sigma$.
	Moreover, if $\pi$ was max-falling (resp. min-rising) w.r.t. $\rho$ 
	and $\rho$ was max-falling (resp. min-rising) w.r.t. $\sigma$,
		then $\pi$ is max-falling (resp. min-rising) w.r.t. $\sigma$.
\item {\em Closure under concatenation. }
	
	Let $\pi$ be an $N$-semirun from $q(x)$ to $r(y)$ and
	let $\rho$ be $N$-semirun from $r(y)$ to $s(z)$.
	Moreover, let $\pi'$ be an $N'$-semirun from $q(x')$ 
		to $r(y')$
	that is an $\ell$-embedding of $\pi$
	and let $\rho'$ be an $N'$-semirun from $r(y')$ to $s(z')$ that
	is an $\ell$-embedding of $\rho$.
		Then $\pi'\rho'$ is an $\ell$-embedding of $\pi\rho$.
	If furthermore, $\pi'$ was max-falling (resp. min-rising) w.r.t. $\pi$ 
	and $\rho'$ was max-falling (resp. min-rising) w.r.t. $\rho$,
	then $\pi' \rho'$ is max-falling (resp. min-rising) w.r.t. $\pi \rho$.
\item {\em Shifting distant embeddings. }
	Let $D\in \macro_\C \Z$ be a multiple of $\macro_\C$, let $\pi$ be a semirun and
	let $\rho$ be an $\ell$-embedding of $\pi$ such that
	for all configurations $q(z)$ in $\rho$ we have $|z - \ell| > |D|$.
	Then both $\rho + D$ and $\rho - D$ are $\ell$-embeddings of $\pi$.
\end{itemize}
\end{remark}

\section{On hills and valleys}\label{hill section}

In this section we introduce
the notions of hills and valleys. Hills are
semiruns that start and end
in configurations with low counter values but where all intermediate configurations have
counter values above these source and target configurations, 
and where moreover $+p$-transitions (resp. $-p$-transitions) are
followed (resp. preceded) by semiruns with 
absolute counter effect 
larger than $\Upsilon_\C$
(we refer to Figure~\ref{hill example} for an illustration of the concept).
We also introduce the dual notion of valleys.
We then prove that an $N$-semirun that is either a hill or a valley
can be turned into a $(N-\Gamma_\C)$-semirun with the same source and target configuration
that is an embedding.
This lowering process serves as a building block in the proof of the $5/6$-Lemma (Lemma~\ref{lemma 5/6}).

\begin{samepage}
\begin{definition}[Hills and Valleys]
	An $N$-semirun 
	$$q_0(z_0) \ \semi{\pi_0,N} \ q_1(z_1) \ \semi{\pi_{1},N} q_2(z_2)\quad\ \cdots \quad\ \semi{\pi_{n-1},N} \ q_n(z_n)$$ is a 
	\begin{itemize}
		\item {\em $B$-hill} if 
			\begin{itemize}
				\item $z_0,z_n<B$,
				\item $z_i\geq B$ for all $i\in[1,n-1]$,
				\item $\pi_i=-p$ implies $z_i > z_0+\Upsilon_\C$ for all $i\in[0,n-1]$,
					and
				\item $\pi_i=+p$ implies $z_{i+1} > z_n+\Upsilon_\C$ for all $i\in[0,n-1]$.
			\end{itemize}
		\item {\em $B$-valley} if 
			\begin{itemize}
				\item $z_0,z_n>B$,
				\item $z_i\leq B$ for all $i\in[1,n-1]$, 
				\item $\pi_i=-p$ implies $z_{i+1} < z_n-\Upsilon_\C$ for all
					$i\in[0,n-1]$, and
				\item $\pi_i=+p$ implies $z_i < z_0-\Upsilon_\C$ for all $i\in[0,n-1]$.
			\end{itemize}
	\end{itemize}
\end{definition}
\end{samepage}

\begin{center}
	\begin{figure}
			\hspace{2.7cm}
\includegraphics[width=0.58\textwidth]{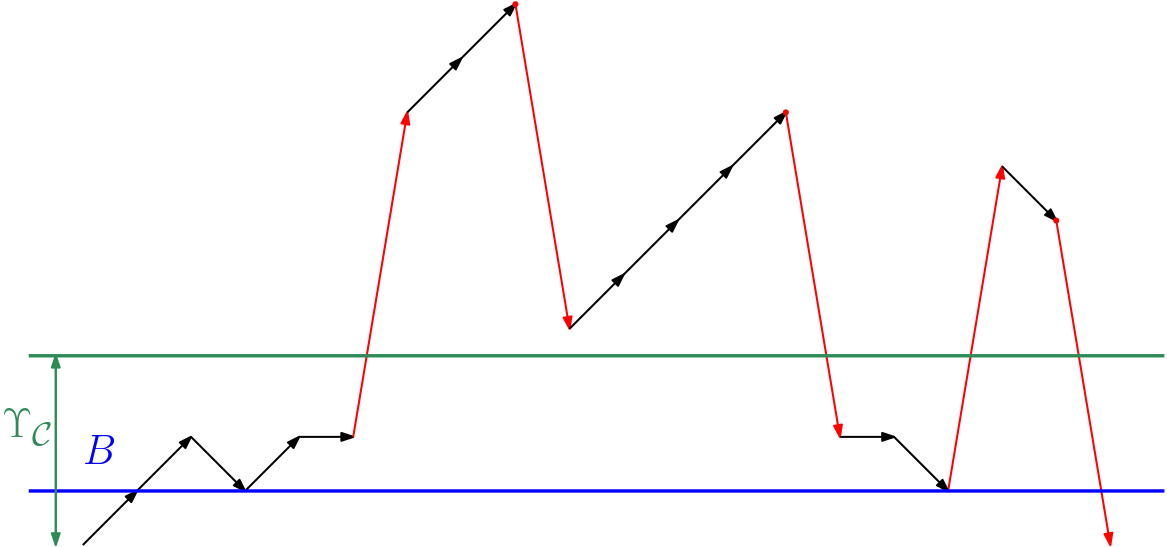}
	\caption{Illustration of a $B$-hill. }
\label{hill example}
	\end{figure}
\end{center}

\medskip

\noindent
The Hill and Valley Lemma states that an $N$-semirun $\pi$
that is either a $B$-hill or a $B$-valley can be turned into an $(N-\Gamma_\C$)-semirun
with the same source and target configuration that is moreover both a min-rising and max-falling 
$B'$-embedding of $\pi$, where $B'$ is close to $B$.

\begin{lemma}[Hill and Valley Lemma]\label{lemma hill/valley}
	For all $N, B \in \N$, all $N$-semiruns
	$\pi$ from $q_0(z_0)$ to $q_n(z_n)$ with $N > M_\C$ and $\values(\pi)\subseteq[0,4N]$
	such that moreover $\pi$ is either a $B$-hill or a $B$-valley,
	there exists an $(N-\Gamma_\C)$-semirun 
	from $q_0(z_0)$ to $q_n(z_n)$ that  
	is both a min-rising and max-falling $(B-\Upsilon_\C-\Gamma_\C-1)$-embedding of $\pi$ 
	(in case $\pi$ is a $B$-hill), 
	or both a min-rising and max-falling $(B+\Upsilon_\C+\Gamma_\C+1)$-embedding of $\pi$ 
	(in case $\pi$ is a $B$-valley).
\end{lemma}
We remark that the resulting $(N-\Gamma_\C)$-semirun satisfies
further properties --- these are being discussed in Section~\ref{HV further}.

Before proving the Hill and Valley Lemma let us explain why the finding of the resulting
embedding is delicate.
Let us fix any $N$-semirun 
\begin{eqnarray*}
	\pi \quad =\quad \ q_0(z_0) \ \semi{\pi_0,N} \ q_1(z_1) \ \semi{\pi_1,N} \quad\ \cdots \quad\ \semi{\pi_{n-1},N} \ q_n(z_n)
\end{eqnarray*}
from $q_0(z_0)$ to $q_n(z_n)$
with $\values(\pi) \subseteq [0, 4N]$
and $N > M_\C$.
Let us moreover assume that $\pi$ is a $B$-hill for some $B\in\N$.
We need to show the existence of 
some $(N-\Gamma_\C)$-semirun from $q_0(z_0)$ to $q_n(z_n)$ that is 
moreover both a {min-rising and max-falling} {$(B-\Upsilon_\C-\Gamma_\C-1)$-embedding} of $\pi$.

We are particularly interested in those transitions $\tau$ with absolute counter effect
$|\Delta(\tau)| = N $, i.e. transitions with operation $+p$ or $-p$
that we will denote as {\em unlowered} $+p$-transitions and $-p$-transitions respectively.
Note that if there is no such transition in $\pi$, then $\pi$ is already an
$(N - \Gamma_\C)$-semirun.
Let us therefore assume there is at least one transition with absolute counter effect $N$ in $\pi$.
For obtaining only an $(N-\Gamma_\C)$-semirun it would simply suffice
to lower the absolute counter effect of 
these
transitions by $\Gamma_\C$.
Indeed, if the transition $\tau = q(z) \semi{+p,N} q'(z')$ 
is an $N$-semirun, then the {\em lowered transition}
$ \widehat{\tau} = q(z) \semi{+p,N-\Gamma_\C} q'(z' - \Gamma_\C)$ 
is an
$(N-\Gamma_\C)$-semirun.
Dually, 
if $\tau = q(z) \semi{-p,N} q'(z')$ is an $N$-semirun, then 
$ \widehat{\tau} = q(z) \semi{-p,N-\Gamma_\C} q'(z' + \Gamma_\C)$
is an 
$(N-\Gamma_\C)$-semirun.

Thus, applying such a lowering to all transitions of $\pi$ whose absolute counter effect is $N$
yields an $(N-\Gamma_\C)$-semirun with target configuration shifted by a multiple of $\Gamma_\C$, 
according to the operations
 seen in Subsection~\ref{semirun section}.
However, the Hill and Valley Lemma not only requires the resulting semirun to be an
$(N-\Gamma_\C)$-semirun but also to have same source and target configurations as the original semirun
(and to be a min-rising and max-falling
$(B-\Upsilon_\C-\Gamma_\C-1)$-embedding).
Hence, simply lowering all transitions with a large counter effect as described above
is not enough to prove the result as the following example illustrates.
Let us assume an $N$-semirun $\pi$ containing
precisely one transition $\tau$ whose absolute
counter effect is $N$, say $\pi_j=+p$ for some position $j$.
That is,
$$
	\pi\quad =\quad \ q_0(z_0) \ \semi{\pi_0,N}  \quad\ \cdots \quad\ \ q_j(z_j) \ \semi{+p,N}
	q_{j+1}(z_{j+1})  \quad\ \cdots \quad\ \semi{\pi_{n-1},N} \ q_n(z_n)\quad.
$$
If we replace directly this $j$-th transition by a transition with $\Delta(\tau')=N-\Gamma_\C$, and, starting with 
the $(j+1)$-th configuration, shift all following counter values by $-\Gamma_\C$,  
we indeed obtain an $(N-\Gamma_\C)$-semirun
$$
	q_0(z_0) \ \semi{\pi_0,N-\Gamma_\C} \ \cdots \ \ q_j(z_j) \ \semi{+p,N-\Gamma_\C} \ q_{j+1}(z_{j+1}-\Gamma_\C) \quad \ \cdots \quad\ \semi{\pi_{n-1},N-\Gamma_\C} \ q_n(z_n-\Gamma_\C)\quad.
$$
However, this $(N-\Gamma_\C)$-semirun does not have the same source and target configuration 
as the original semirun, as the target configuration's counter value has been shifted by $-\Gamma_\C$. 
Worse yet, if our initial $N$-semirun $\pi$ were to possess several $+p$-transitions, then
the accumuluated counter value shifts could potentially yield that the resulting $(N-\Gamma_\C)$-semirun
is not a $(B-\Upsilon_\C-\Gamma_\C-1)$-embedding of $\pi$: indeed, such a shifted semirun could contain
intermediate configurations with counter values less than $B-\Upsilon_\C-\Gamma_\C-1$.

In order to account for those transitions whose absolute counter effect is $N$ that have
already been lowered or not we will introduce the notion of hybrid semiruns, which can be seen
as sequences of $N$-semiruns and $(N-\Gamma_\C)$-semiruns whose source and target configurations
are suitably connected.

		\begin{samepage}
			\begin{definition}\label{definition hybrid semirun}
A {\em hybrid semirun} is a sequence 
				$\eta=\alpha^{(0)}\beta^{(1)}\alpha^{(1)}\cdots\beta^{(k)}
				\alpha^{(k)}$,
where 
		\begin{itemize}
			\item each $\alpha^{(i)}$ is an 
				$(N-\Gamma_\C)$-semirun (possibly empty)
				of the form
				$$\alpha^{(i)}\qquad=\qquad 
				p_0(y_0)\ \semi{\alpha_0^{(i)},N-\Gamma_\C}\
				p_1(y_1)\qquad\cdots\qquad
				\semi{\alpha_{m_i}^{(i)},N-\Gamma_\C}\ p_{m_i}(y_{m_i})$$
			\item each $\beta^{(i)}$ is a single transition
				with $|\Delta(\beta^{(i)})|=N$,		
			\item the target configuration of $\alpha^{(i-1)}$ is the source configuration
				of $\beta^{(i)}$ for all $i\in[1,k]$, and 
			\item the source configuration of $\alpha^{(i)}$ is the target configuration
				of $\beta^{(i)}$ for all $i\in[1,k]$.
		\end{itemize}
			We call $k$ the {\em breadth} of $\eta$.
			\end{definition}

\begin{remark}\label{init hybrid}
	In case our initial $N$-semirun $\pi$ contains $k$ transitions of absolute counter
			effect $N$, we observe that $\pi$ can naturally be viewed as an initial hybrid semirun
			of breadth $k$.
\end{remark}

			Several of the notions (such as counter effect, length and maximum)
			that we have defined for runs and semiruns can naturally be 
			extended to hybrid 
			semiruns.
			As expected, the projection $\phi(\eta)$ is defined as
			$\phi(\eta)=\phi(\alpha^{(0)})\phi(\beta^{(1)})
			\phi(\alpha^{(1)})\cdots \phi(\beta^{(k)})\phi(\alpha^{(k)})$.
			We moreover introduce the particular projection $\phi_\restriction$ of $\phi$
			restricted to the $\alpha^{(i)}$, i.e.
			${\phi}_\restriction(\eta)=\phi(\alpha^{(0)})\phi(\alpha^{(1)})
			\cdots\phi(\alpha^{(k)})$.

			Moreover, we view the $\alpha^{(i)}$ themselves as 
			sequences (not as atomic objects) of length $m_i$
			and the $\beta^{(i)}$ as sequences of length one. Using this convention,
			the notions of prefixes, infixes 
			and suffixes are as expected.
	More importantly, we extend naturally the notion of 
			(\textcolor{black}{max-falling and min-rising}) 
	$\ell$-embedding to 
	hybrid semiruns as in Definition~\ref{def embedding} when treating them as such sequences.

We prove the Hill and Valley Lemma (Lemma~\ref{lemma hill/valley}) in Section~\ref{HV proof}.
We summarize important further consequences of the proof in Section~\ref{HV further}.

\subsection{Proof of the Hill and Valley Lemma}
\label{HV proof}

Let us fix any $N$-semirun 
\begin{eqnarray*}
	\pi \quad =\quad \ q_0(z_0) \ \semi{\pi_0,N} \ q_1(z_1) \ \semi{\pi_1,N} \quad\ \cdots \quad\ \semi{\pi_{n-1},N} \ q_n(z_n)
\end{eqnarray*}
from $q_0(z_0)$ to $q_n(z_n)$
with $\values(\pi) \subseteq [0, 4N]$
and $N > M_\C$.
Let us moreover assume that $\pi$ is a $B$-hill for some $B\in\N$.
The case when $\pi$ is a $B$-valley can be proven analogously.

For reasons of simplicity we separate the proof into two cases, namely if there is
a $+p$-transition or $-p$-transition whose source and target configurations
have counter values that are both at most $B+\Upsilon_\C+\Gamma_\C$ or not.
		Section~\ref{Case A} deals with the latter case, 
		Section~\ref{Case B} with the former.
		It is worth mentioning that Section~\ref{Case B} depends on Section~\ref{Case A}.

		\subsubsection{$\pi$ does not contain any $\pm p$-transition whose
		source and target configuration both have counter value at most $B+\Upsilon_\C+\Gamma_\C$}\label{Case A}

\end{samepage}

In the following let us denote by $\mathcal{L}$ the {\em critical level}, i.e. the constant
$$
		\mathcal{L}= B+\Gamma_\C.
$$
Moreover, for a hybrid semirun 
$\eta=\alpha^{(0)}\beta^{(1)}\alpha^{(1)}\cdots\beta^{(k)}\alpha^{(k)}$, 
		for every $\beta^{(j)}$ that is an unlowered $+p$-transition,
		we define the {\em critical descending infix with respect to 
		$\beta^{(j)}$}
		as the shortest prefix 
		(when viewed as a sequence, as mentioned above) of 
		$\alpha^{(j)}\beta^{(j+1)}\alpha^{(j+1)}\cdots \beta^{(k)}\alpha^{(k)}$
		that ends in a configuration with counter value at most $\mathcal{L}$.
		In particular, this critical descending infix could possibly
		end in a configuration inside some (strict prefix of) $\alpha^{(i)}$, where $i\in[j,k]$.
Dually, for every $\beta^{(j)}$ that is an unlowered $-p$-transition,
		we define the {\em critical ascending infix 
		with respect to $\beta^{(j)}$}
		as the shortest suffix of $\alpha^{(0)}\beta^{(1)}\cdots \alpha^{(j-1)}$
		that starts in a configuration with counter value at most $\mathcal{L}$.
		The following remark is central.\\[-0.2cm]

		\begin{remark}{\label{remark critical}}
			For a hybrid semirun 
$\eta=\alpha^{(0)}\beta^{(1)}\alpha^{(1)}\cdots\beta^{(k)}\alpha^{(k)}$, if
			some unlowered $-p$-transition (resp. $+p$-transition) $\beta^{(j)}$ 
			appears in the critical 
			descending infix (resp. critical ascending infix) of some 
			unlowered $+p$-transition (resp. $-p$-transition)
			$\beta^{(i)}$, then so does $\beta^{(i)}$ appear in the 
			critical ascending infix (resp. critical
			descending infix) of $\beta^{(j)}$.
		\end{remark}

Viewing our initial semirun $\pi$ as a hybrid semirun, we will now introduce two phases that 
successively lower unlowered $+p$-transitions
and unlowered $-p$-transitions yielding hybrid semiruns that
retain an approximation invariant (Definition~\ref{def approximates}).

In phase one, we 
 are interested in unlowered $+p$-transitions.
We want to progressively lower these,
going from right to left. Moreover, we want to inspect the critical descending infix in order to obtain 
successive min-rising and max-falling embeddings with the same source and target configuration. 
In case the rightmost unlowered $+p$-transition has the property that its critical descending
infix contains some unlowered $-p$-transition we lower the leftmost such directly,
together with the $+p$-transition.
Otherwise, we want to make use of the Bracket Lemma 
(Lemma~\ref{bracket lemma}) and the Depumping Lemma (Lemma~\ref{lemma zero})
in order to retain some nice bracketing properties.

Having successively lowered all unlowered $+p$-transitions in phase one,
we finally lower the remaining unlowered $-p$-transitions in phase two.
For these we take their critical ascending infix and their $\varphi_\restriction$-projection
into account, again yielding some carefully chosen bracketing property.

The following definition formalizes the above-mentioned bracketing property.

		\begin{samepage}
		\begin{definition}{\label{def approximates}}
		A hybrid semirun {\em $\eta$ approximates $\pi$ with respect 
			to level $\ell\in\Z$} if 
	\begin{enumerate}
		\item $\eta=\alpha^{(0)}\beta^{(1)}\alpha^{(1)}\cdots\beta^{(k)}\alpha^{(k)}$ is a 
			hybrid semirun of some breadth $k$,
		\item$ \pi$ can be factorized
			as $\pi=\chi^{(0)}\zeta^{(1)}\chi^{(1)}\cdots\zeta^{(k)}\chi^{(k)}$, 
			where the $\zeta^{(i)}$ 
			are transitions with operation either $+p$ or $-p$,
		\item $\eta$ is a min-rising and max-falling $\ell$-embedding of $\pi$,
		\item $\alpha^{(i)}$ is a max-falling $\ell$-embedding of 
			$\chi^{(i)}$ for all $i\in[0,k]$ 
			with the same source and target configuration as $\chi^{(i)}$,
		\item every prefix of $\phi_\restriction(\gamma^{(i)})$ contains at least as many 
			occurrences of $\boldsymbol{[}$ as of $\boldsymbol{]}$,
			where $\gamma^{(i)}$ is the critical descending infix of $\beta^{(i)}$
			for all $i\in[1,k]$ for which $\beta^{(i)}$ has operation $+p$, and
		\item every suffix of $\phi_\restriction(\gamma^{(i)})$ 
			contains at least as many 
			occurrences of $\boldsymbol{]}$ as of $\boldsymbol{[}$, where 
			$\gamma^{(i)}$ is the critical ascending
			infix of $\beta^{(i)}$ for all $i\in[1,k]$ for which $\beta^{(i)}$ 
			has operation $-p$.
	\end{enumerate}
		\end{definition}
		\end{samepage}

By completing phase one and then phase two we will show the existence 
of a hybrid semirun that approximates $\pi$ with respect to level $B$ and does not contain any
unlowered $+p$-transition nor any unlowered $-p$-transition (and is hence
an $(N-\Gamma_\C)$-semirun). 
Observe first that by Point 4 any such hybrid semirun $\eta$ has the same source and
target configuration as $\pi$. Second, any such $\eta$
 is in particular a
 min-rising and max-falling 
 $(N-\Upsilon_\C-\Gamma_C-1)$-embedding
of $\pi$ since $\pi$ is assumed to be a $B$-hill. Thus, the lemma follows.
We will obtain the desired $(N-\Gamma_\C)$-semirun and variants thereof
by first systematically lowering $+p$-transitions from the rightmost to the leftmost in phase one
and secondly systematically lowering the possibly remaining $-p$-transitions from the leftmost to the rightmost
in phase two. 
We denote such a process --- whose details are given below --- by the so-called {\em $(+p,-p)$-lowering
process}. 
As mentioned in Remark~\ref{alternate process} we will also define a dual variant, namely
the {\em $(-p,+p)$-lowering process}: here phase one will consist of systematically
lowering the $-p$-transitions from the leftmost to the rightmost, whereas phase two 
will systematically lower the possibly remaining $+p$-transitions from the rightmost to the leftmost.

Remark~\ref{alternate process} finally discusses a variant of a $(+p,-p)$-lowering process 
(resp. $(-p,+p)$-process) which ends in a hybrid semirun that contains precisely
one unlowered transition.

Let us discuss the $(+p,-p)$-lowering process in detail.

\medskip
\noindent
{\bf Phase one of the $(+p,-p)$-lowering process: Lowering $+p$-transitions}

\label{phase one}
We can view our initial $N$-semirun $\pi$ as a hybrid semirun 
$\eta^{(0)}$ of breadth $k_0$, i.e. 
$$
		\eta^{(0)}\quad =\quad \alpha^{(0,0)}\beta^{(0,1)}\alpha^{(0,1)}\cdots 
		\beta^{(0,k_0)}\alpha^{(0,k_0)}\quad.
$$
In phase one we will inductively show the existence of a sequence of hybrid semiruns
$\eta^{(0)}, \eta^{(1)},\ldots,\eta^{(r)}$, where each $\eta^{(i)}$ has breadth $k_i$ and 
approximates $\pi$ with respect to level $B$,
$\eta^{(r)}$ does not contain any unlowered $+p$-transition, and $k_{i-1}>k_i$ for all $i\in[1,r]$.
Let us assume that we have inductively already defined the sequence $\eta^{(0)},\ldots,\eta^{(i-1)}$
of hybrid semiruns for some $i\geq 1$ and where $\eta^{(i-1)}$ has breadth 
$k_{i-1}>0$ and approximates $\pi$
with respect to level $B$
and contains at least one unlowered $+p$-transition.
Towards extending the sequence we need to show the existence of some hybrid
semirun $\eta^{(i)}$ of breadth $k_{i}<k_{i-1}$ that approximates $\pi$ with respect to level $B$.

Let $\eta^{(i-1)}=\alpha^{(i-1,0)}\beta^{(i-1,1)}\alpha^{(i-1,1)}\cdots\beta^{(i-1,k_{i-1})}
\alpha^{(i-1,k_{i-1})}$.
Let $j\in[1,k_{i-1}]$ be maximal such that $\beta^{(i-1,j)}$ is an unlowered $+p$-transition.
For defining $\eta^{(i)}$ we make the following case distinction.
\begin{enumerate}
	\item The critical descending infix with respect to 
		the $+p$-transition $\beta^{(i-1,j)}$ 
		contains at least one 
	unlowered $-p$-transition.
That is, the critical descending infix is of the form 
		$$\alpha^{(i-1,j)}\beta^{(i-1,j+1)}\alpha^{(i-1,j+1)}\cdots\beta^{(i-1,h)}\xi,$$
		where $\xi$ is a prefix (possibly empty) of $\alpha^{(i-1,h)}$,
		$\beta^{(i-1,j+1)}$ is an unlowered
	$-p$-transition and where $h\geq j+1$. 
	We refer to Figure~\ref{phase one case one} for an illustration.
		Our desired hybrid semirun $\eta^{(i)}$ is obtained from 
		$\eta^{(i-1)}$ by simply 
		lowering both $\beta^{(i-1,j)}$ and $\beta^{(i-1,j+1)}$, 
		i.e. replacing $\beta^{(i-1,j)}$ by $\widehat{\beta^{(i-1,j)}}$ 
		satisfying $\Delta(\widehat{\beta^{(i-1,j)}})=N-\Gamma_\C$
		and replacing $\beta^{(i-1,j+1)}$ by a suitable $\widehat{\beta^{(i-1,j+1)}}$
		satisfying $\Delta(\widehat{\beta^{(i-1,j+1)}})=-N+\Gamma_\C$
		and moreover suitably shifting the part after $\widehat{\beta^{(i-1,j)}}$
		and until (including) $\widehat{\beta^{(i-1,j+1)}}$ by $-\Gamma_\C$.
		More precisely, the part $\alpha^{(i,j-1)}$ in $\eta^{(i)}$ 
		is chosen to be of the form
		$$\alpha^{(i,j-1)}=\alpha^{(i-1,j-1)}\widehat{\beta^{(i-1,j)}}
		\left(\alpha^{(i-1,j)}-\Gamma_\C\right)
		\left(\widehat{\beta^{(i-1,j+1)}}-\Gamma_\C\right)\alpha^{(i-1,j+1)}.$$
		Moreover, observe that $\alpha^{(i,j-1)}$
	and the infix 
		$\alpha^{(i-1,j-1)}\beta^{(i-1,j)}
		\alpha^{(i-1,j)}\beta^{(i-1,j+1)}\alpha^{(i-1,j+1)}$
		of $\eta^{(i-1)}$
connect the same source and target configurations.
		Thus, it easily follows that $\eta^{(i)}$ also approximates $\pi$ 
		with respect to level $B$.
		Finally, observe that the breadth of $\eta^{(i)}$ equals $k_{i-1}-2$.

\begin{center}
\begin{figure}
\includegraphics[width=0.567\textwidth]{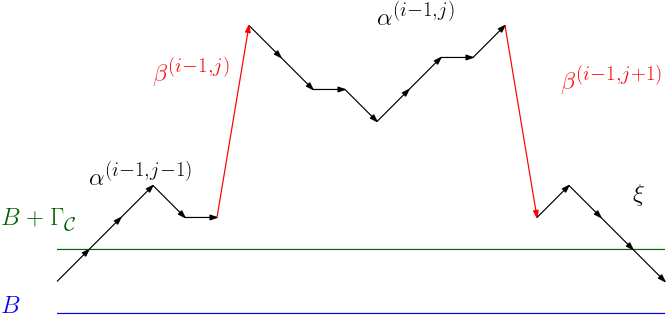}
	\caption{Illustration of phase one case 1, i.e. the unlowered $+p$-transition 
	$\beta^{(i-1,j)}$ can be lowered by lowering it with the leftmost unlowered $-p$-transition on
	its critical descending infix, i.e. $\beta^{(i-1,j+1)}$.}
\label{phase one case one}
\end{figure}
\end{center}
\item The critical descending infix with respect to 
the $+p$-transition $\beta^{(i-1,j)}$ does not contain any 
	unlowered $-p$-transition. It follows that the critical descending infix 		with respect to $\beta^{(i-1,j)}$ is a
		non-empty prefix $\xi$ of $\alpha^{(i-1,j)}$.
	We refer to Figure~\ref{phase one case two} for an illustration.
		Recall that $\eta^{(i-1)}$ approximates $\pi$ with respect to level $B$.
	Firstly, since by assumption $\values(\pi)\subseteq[0,4N]$, it follows 
	from Point 3 of Definition~\ref{def approximates} that 
		$\values(\xi)\subseteq[0,4N]$.
	Secondly, from Point 5 of Definition~\ref{def approximates} every prefix of 
		$\phi(\alpha^{(i-1,j)})$
contains at least as many occurrences of $\boldsymbol{[}$ as of $\boldsymbol{]}$. 
		Hence, the latter must also hold for 
	every prefix of $\phi(\xi)$.
	Thirdly, since by the case of this subsection the target configuration of every transition with operation $+p$ in
	$\pi$ has counter value strictly larger than $B+\Upsilon_\C+\Gamma_\C$, it follows
	from Points 2 and 4 of Definition~\ref{def approximates} that the target configuration
	of $\beta^{(i-1,j)}$ ends in a configuration with counter value strictly larger than 
	$B+\Upsilon_\C+\Gamma_\C$. Since $\xi$ is the critical descending infix  with respect to 
	$\beta^{(i-1,j)}$ (in particular ending in a configuration with counter value at most 
	$B+\Gamma_\C)$, it follows $\Delta(\xi)<-\Upsilon_\C$.
	Hence one can apply Lemma~\ref{bracket lemma} to the $(N-\Gamma_\C)$-semirun $\xi$ 
	yielding an infix 
	$\xi[c,d]$ satisfying $\phi(\xi[c,d])\in\Lambda_8$ and $\Delta(\xi[c,d])<-\Upsilon_\C$.
	Applying Lemma~\ref{lemma zero} to $\xi[c,d]$ implies the existence
	of an $(N-\Gamma_\C)$-semirun $\xi'=\xi[c,d]-I_1-I_2\cdots-I_s$
	satisfying $\Delta(\xi')=\Delta(\xi[c,d])+\Gamma_\C$
	and where $I_1,\ldots,I_s$ are pairwise disjoint intervals of positions in $\xi[c,d]$
	such that moreover $\phi(\xi[c,d][I_t])\in\Lambda_{16}$ and $\Delta(\xi[c,d][I_t])<0$
	for all $t\in[1,s]$. 
	Assume that $\xi=\xi[0,m]$ consisted of $m$ transitions; thus in particular $c,d\in[0,m]$.
	By combining the above properties it immediately follows that 
	$$\xi''=\xi[0,c]\xi'\left(\xi[d,m]+\Gamma_\C\right)$$
	is an $(N-\Gamma_\C)$-semirun with $\Delta(\xi'')=\Delta(\xi)+\Gamma_\C$
	and that $\xi'' - \Gamma_\C$
	is a max-falling $B$-embedding of $\xi$.
	We define the desired $\eta^{(i)}$ to be obtained from $\eta^{(i-1)}$ by 
	lowering $\beta^{(i-1,j)}$ to $\widehat{\beta^{(i-1,j)}}$ satisfying 
	$\Delta(\widehat{\beta^{(i-1,j)}})=\Delta(\beta^{(i-1,j)})-\Gamma_\C$ 
	and moreover replacing $\xi$ by $\xi''-\Gamma_\C$. 
	Observe that $\eta^{(i)}$ and $\eta^{(i-1)}$ only differ in the infix $\alpha^{(i,j-1)}$
	of $\eta^{(i)}$. The latter is hence of the form
$$
\alpha^{(i,j-1)}=\alpha^{(i-1,j-1)}\widehat{\beta^{(i-1,j)}}\left(\xi''-\Gamma_\C\right)
\alpha^{(i-1,j)}[m,|\alpha^{(i-1,j)}|].
$$
By construction $\eta^{(i-1)}$'s infix  
$$
\alpha^{(i-1,j-1)}\beta^{(i-1,j)}\alpha^{(i-1,j)}
$$
has the same source and target configuration as the part $\alpha^{(i,j-1)}$ of $\eta^{(i)}$.
	Since moreover 
	\begin{itemize}
		\item $\phi(\xi[c,d][I_t])\in\Lambda_{16}$ contains
	precisely as many occurrences of $\boldsymbol{[}$ as of $\boldsymbol{]}$
	and $\Delta(\xi[c,d][I_t])<0$ for each $t\in[1,s]$ and
\item $\Delta(\xi'')=\Delta(\xi)+\Gamma_\C$
	\end{itemize}
	it follows that indeed
	$\eta^{(i)}$ approximates $\pi$ with respect to level $B$.
	Finally, observe that the breadth of $\eta^{(i)}$ is 
	$k_{i-1}-1$.

\begin{center}

\begin{figure}
\includegraphics[width=0.56\textwidth]{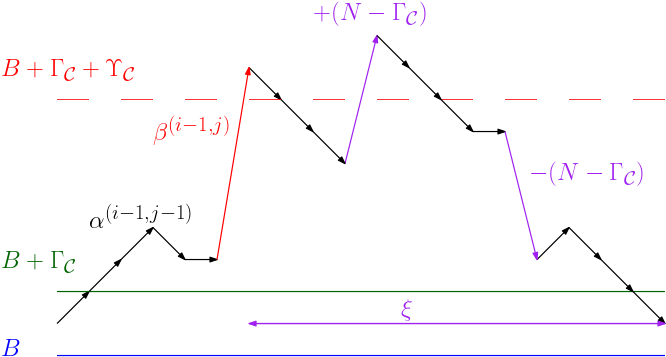}
\caption{Illustration of phase one case 2, i.e. the suffix of the 
	to be lowered $+p$ transition 
	$\beta^{(i-1,j)}$ does not contain any unlowered $-p$-transition, i.e.
	any transition with counter effect $-N$, 
	inside its critical descending infix.}
\label{phase one case two}
\end{figure}
\end{center}
\end{enumerate}

Recall that in phase one we have repeatedly lowered unlowered $+p$-transitions from right to left.
In doing so we have hereby possibly lowered certain $-p$-transitions.
The final hybrid semirun $\eta^{(r)}$ of phase one notably does not contain any
unlowered $+p$-transition.
However, $\eta^{(r)}$ may still contain unlowered $-p$-transitions.
Lowering these will be subject of phase two. Yet, these unlowered $-p$-transitions
will be lowered rather  from leftmost to rightmost 
(instead of from rightmost to leftmost as in phase one).

\medskip
\noindent
{\bf Phase two of the $(+p,-p)$-lowering process: Lowering $-p$-transitions that remain after phase one}

Recall that $\mathcal{L}=B+\Gamma_\C$ denotes our critical level.
Also recall that $\eta^{(r)}$ is the final hybrid semirun in the sequence 
$\eta^{(0)},\ldots,\eta^{(r)}$ 
of phase one and approximates $\pi$ with respect to level $B$.
Note that by construction $\eta^{(r)}$ does not contain any unlowered 
$+p$-transition.
That is, all unlowered transitions of 
$\eta^{(r)}$ have operation $-p$ and there
are as many of them as the breadth of $\eta^{(r)}$.
Setting $\eta^{(0)'}=\eta^{(r)}$, phase two consists in showing the existence
of a sequence of hybrid semiruns $\eta^{(1)'},\ldots, \eta^{(t)'}$ 
all of which
do not contain any unlowered $+p$-transition and in which each 
$\eta^{(i)'}$ has breadth $k_i'$ satisfying $k_i'<k_{i-1}'$,
where each $\eta^{(i)'}$ approximates $\pi$ with respect to level $B$,
and 
finally $\eta^{(t)'}$ is of breadth $0$ 
(and is therefore already an $(N-\Gamma_\C)$-semirun).

Let us inductively assume that we have already defined the sequence
$\eta^{(0)'},\ldots,\eta^{(i-1)'}$ for some $i\geq 1$
and that the breadth $k_{i-1}'$ of $\eta^{(i-1)'}$ 
satisfies $k_{i-1}'>0$.

Let $\eta^{(i-1)'}=\alpha^{(i-1,0)'}\beta^{(i-1,1)'}\alpha^{(i-1,1)'}\cdots
\beta^{(i-1,k_{i-1})'}\alpha^{(i-1,k_{i-1})'}$.
There is only one possible case for this phase
since the critical ascending infix with respect to the
leftmost unlowered $-p$-transition
$\beta^{(i-1,1)'}$ does not contain any unlowered $+p$-transition 
since $\eta^{(i-1)'}$ does not.
The construction of $\eta^{(i)'}$, as well as the proof that $\eta^{(i)'}$ approximates $\pi$ with respect to level $B$,
is completely dual to the proof of the second case of phase one and 
therefore omitted.

\begin{example}\label{Example Lowering Process}
Figure~\ref{Figure Lowering Process} illustrates an example of an application
of the $(+p,-p)$-lowering process. The topmost figure on the left is the starting hybrid semirun $\pi$. 
We begin the process by lowering the two unlowered $+p$-transitions each by compensating them with a $-p$-transition, 
	then 
	we enter phase two with one unlowered $-p$-transition remaining, which we lower and compensate 
	by shifting and cutting out portions 
	inside the critical ascending infix by applying the Depumping Lemma.
\end{example}

\begin{remark}\label{remark valley}
	In case $\pi$ is a $B$-valley instead of a $B$-hill there is a dual 
	variant of the $(+p,-p)$-lowering process. 
	The critical level would be adjusted to $\mathcal{L}=B-\Gamma_\C$, 
	for unlowered $+p$-transitions one would define the critical descending 
	infix to be the shortest suffix of 
	$\alpha^{(0)}\beta^{(1)}\cdots\alpha^{(j-1)}$ that 
	starts in a configuration with counter value at least $\mathcal{L}$,
	whereas for unlowered $-p$-transitions one would 
	define the critical
	ascending infix to be the shortest prefix of 
	$\alpha^{(j)}\beta^{(j+1)}\alpha^{(j+1)}\cdots\beta^{(k)}\alpha^{(k)}$
	that ends in a configuration with counter value at least $\mathcal{L}$.
	The definition when a hybrid semirun approximates $\pi$ with respect to 
	level $B$
	would be defined analogously as in Definition~\ref{def approximates},
	but where in Point 4 $\alpha^{(i)}$ is rather required to be a
	min-rising $B$-embedding of $\chi^{(i)}$
	with the same source and target configuration as $\chi^{(i)}$,
	Point 5 (resp. Point 6) of Definition~\ref{def approximates} would
	rather require that every suffix (resp. every prefix) of $\phi_\restriction(\gamma^{(i)})$ 
	contains at least as many occurrences of $\boldsymbol{[}$ (resp. $\boldsymbol{]}$)
	as of $\boldsymbol{]}$ (resp. $\boldsymbol{[}$).
\end{remark}

\begin{center}
	\begin{figure}
		\hspace{0.3cm}
\includegraphics[width=0.43\textwidth]{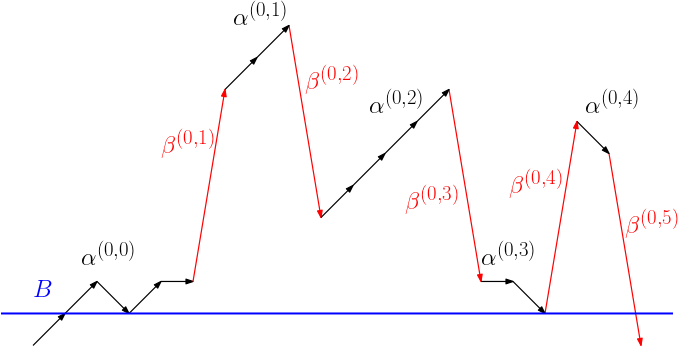}
		\hspace{1cm}
\includegraphics[width=0.43\textwidth]{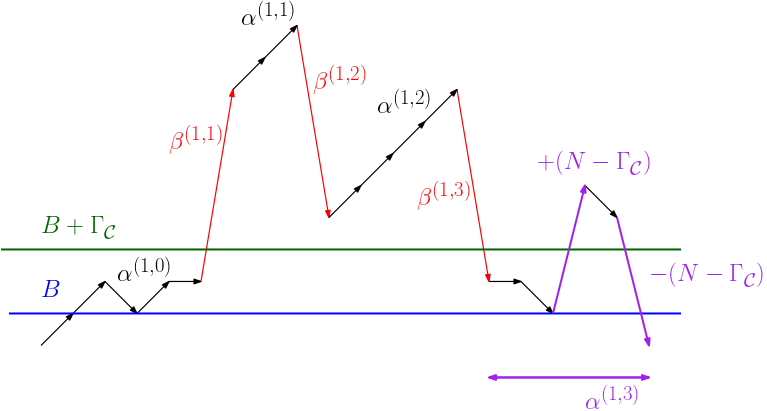}

		\hspace{0.3cm}
\includegraphics[width=0.43\textwidth]{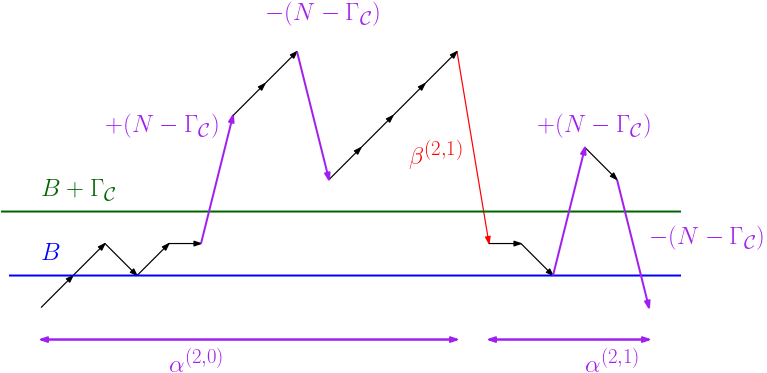}
		\hspace{1cm}
\includegraphics[width=0.43\textwidth]{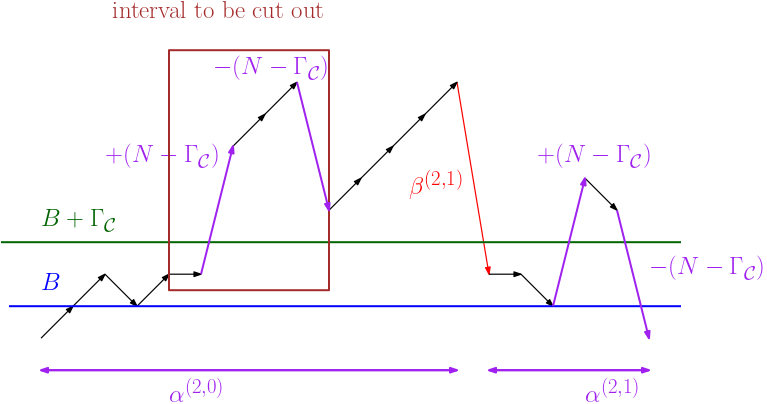}

		\vspace{0.5cm}

		\hspace{0.3cm}
\includegraphics[width=0.43\textwidth]{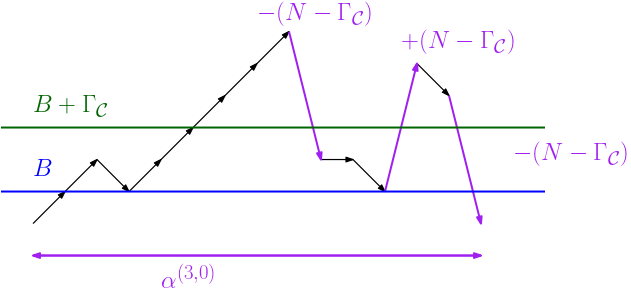}

		\caption{Illustration of the $(+p,-p)$-lowering process from Example~\ref{Example Lowering Process}
		to be read from upper left to lower right.}\label{Figure Lowering Process}
	\end{figure}
\end{center}

\begin{remark}\label{alternate process}
	Consider the following variants of the $(+p,-p)$-lowering process for our $B$-hill $\pi$
	(dual variants can be formulated in the case when $\pi$ is $B$-valley):
	\begin{enumerate}
		\item Consider the dual {\em $(-p,+p)$-lowering process}:
			In phase one we lower the $-p$-transitions from the leftmost to the
			rightmost and in phase two lower the $+p$-transitions from the rightmost
			to the leftmost.
			That is, such a $(-p,+p)$-lowering process produces a sequence of hybrid semiruns
			$$
			\eta^{(0)},\eta^{(1)},\ldots,\eta^{(s)},\eta^{(0)'},\eta^{(1)'},\ldots,\eta^{(t)'},
			$$
			that all approximate $\pi$ with respect to level $B$
			where $\eta^{(0)}=\pi$, $\eta^{(i)}$ is obtained 
			from $\eta^{(i-1)}$
			by lowering the leftmost unlowered 
			$-p$-transition of $\eta^{(i)}$,
			$\eta^{(0)'}=\eta^{(s)}$, 
			$\eta^{(i+1)'}$ is obtained from $\eta^{(i)'}$ by lowering
			the rightmost unlowered $+p$-transition, and 
			finally $\eta^{(t)'}$ has breadth $0$.
		\item 	Consider again the sequence of hybrid semiruns 
			$$
			\eta^{(0)},\eta^{(1)},\ldots,\eta^{(s)},\eta^{(0)'},\eta^{(1)'},
			\ldots,\eta^{(t)'},
			$$ 
			of the $(+p,-p)$-process (dually $(-p,+p)$-process):
			\begin{enumerate}
				\item If $t>0$, then observe that 
					$\eta^{(t-1)'}$ and breadth $1$ and contains
					precisely one unlowered transition, namely an
					unlowered $-p$-transition (dually $+p$-transition).
				\item If however $t=0$, then we claim that every prefix (dually suffix) of 
					$\phi(\eta^{(0)'})=\phi(\eta^{(s)})$ 
					contains at least as many occurrences of 
					$\boldsymbol{[}$ (dually occurrences of $\boldsymbol{]}$)
					as occurrences of $\boldsymbol{]}$ (dually occurrences of 
					$\boldsymbol{[}$).
					Indeed, it follows immediately from the fact that each 
					$\eta^{(i)}$
					is obtained from $\eta^{(i-1)}$ by lowering a $+p$-transition
					(dually a $-p$-transition) either by shifting 
				infixes and
					cutting out certain infixes $\zeta'$
					for which $\phi(\zeta')$ contains as many occurrences 
					of $\boldsymbol{[}$ as of $\boldsymbol{]}$,
					or by lowering a $+p$-transition 
					(dually $-p$-transition) together with an unlowered $-p$-transition 
					(dually $+p$-transition) to the right (dually to the left).
				\item 	If $\phi(\pi)$ a priori 
					contains strictly more occurrences of 
			$\boldsymbol{]}$ than of $\boldsymbol{[}$ one can --- 
					by applying the $(+p,-p)$-lowering process --- obtain a sequence
					of hybrid semiruns
					$$
					\eta^{(0)},\eta^{(1)},\ldots,\eta^{(s)},
					\eta^{(0)'},\eta^{(1)'},\ldots,\eta^{(t)'},
					$$
					where $t>0$, all $\eta^{(i)}$ and $\eta^{(j)'}$ 
					approximate $\pi$ with respect to level $B$
					(and are therefore, as remarked above
					by bearing in mind that $\pi$ is $B$-hill,
					in particular both min-rising and max-falling 
					$(B-\Upsilon_\C-\Gamma_\C-1)$-embeddings
					of $\pi$ with the same source and target configuration
					as $\pi$)
					and where the breadth of $\eta^{(t-1)'}$ is $1$.
			Dually, if $\phi(\pi)$ contains strictly more occurrences of 
					$\boldsymbol{[}$ than of 
					$\boldsymbol{]}$ one can --- 
					by applying the $(-p,+p)$-lowering process --- obtain a sequence
					of hybrid semiruns
					$$
					\eta^{(0)},\eta^{(1)},\ldots,
					\eta^{(s)},\eta^{(0)'},\eta^{(1)'},\ldots,\eta^{(t)'},
					$$
					where $t>0$, all $\eta^{(i)}$ and $\eta^{(j)'}$ 
approximate $\pi$
					with respect to level $B$
					(and are therefore both min-rising and max-falling 
					$(B-\Upsilon_\C-\Gamma_\C-1)$-embeddings
					of $\pi$ with the same source and target configuration
					as $\pi$)
					and where the breadth of $\eta^{(t-1)'}$ is $1$.
			\end{enumerate}
			
	\end{enumerate}
\end{remark}

\subsubsection{$\pi$ contains a $\pm p$-transition whose
		source and target configuration both have counter value at most $B+\Upsilon_\C+\Gamma_\C$}\label{Case B}

The presence of a $+p$-transition 
(resp. $-p$-transition) 
$q_i(z_i)\semi{\pi_i,N}q_{i+1}(z_{i+1})$ 
for which we have $\max\{z_i,z_{i+1}\}\leq B+\Upsilon_\C+\Gamma_\C$
implies

$z_{i+1} - (B+\Gamma_\C) \leq \Upsilon_\C$ 
(resp.  $z_{i} - (B+\Gamma_\C) \leq \Upsilon_\C$),
so the core problem is that in both cases
it is not possible to apply the Bracket Lemma
(Lemma~\ref{bracket lemma}) 
in the critical descending (resp. ascending) infix of
 such an unlowered transition.
We thus have to find another way to compensate for lowering such transitions.

We next claim that firstly, any 
$+p$-transition 
whose configurations both have a counter value at most 
$B+\Upsilon_\C+\Gamma_\C$
must be the first transition of $\pi$ and secondly, 
any $-p$-transition with the same property
must be the last transition of 
$\pi$.
Indeed, every 
$+p$-transition $q_i(z_i)\semi{\pi_i,N}q_{i+1}(z_{i+1})$
that is not the first transition (i.e. $i>0$)
satisfies $z_i\geq B$ as $\pi$ is a $B$-hill.
As a consequence, we have 
$z_{i+1}\geq B+N > B+ M_\C > B + \Gamma_\C + \Upsilon_\C$,
where the last inequality follows from $M_\C$'s definition on page~\pageref{constant definitions}.
Dually, if there exists a $-p$-transition 
$q_i(z_i)\semi{\pi_i,N}q_{i+1}(z_{i+1})$
with $z_i\leq B+\Upsilon_\C+\Gamma_\C$
it must be the last transition
$q_{n-1}(z_{n-1})\semi{\pi_{n-1},N}q_n(z_n)$ of $\pi$.

To finalize the proof it thus suffices to distinguish 
whether both the first transition of $\pi$ is a
$+p$-transition with counter values at most
$(B+\Upsilon_\C+\Gamma_\C)$ and
the last transition of $\pi$ is a $-p$-transition
with counter values at most
$(B+\Upsilon_\C+\Gamma_\C)$,
or this holds for precisely one of them.
We thus distinguish these two cases, however in opposite order.

\medskip

\noindent
{\em Case 1. 
The first transition $q_0(z_0)\semi{\pi_0,N}q_1(z_1)$ is a 
$+p$-transition with counter values at most $(B+\Upsilon_\C+\Gamma_\C)$
and the last transition $q_{n-1}(z_{n-1})\semi{\pi_{n-1},N}q_{n-1}(z_{n-1})$ is not a
 $-p$-transition with counter values at most $(B+\Upsilon_\C+\Gamma_\C)$, or 
the first transition $q_0(z_0)\semi{\pi_0,N}q_1(z_1)$ is not a 
$+p$-transition with counter values at most $(B+\Upsilon_\C+\Gamma_\C)$
and the last transition $q_{n-1}(z_{n-1})\semi{\pi_{n-1},N}q_{n-1}(z_{n-1})$ is a
 $-p$-transition with counter values at most $(B+\Upsilon_\C+\Gamma_\C)$.
}

We only treat the case when 
the first transition $q_0(z_0)\semi{\pi_0,N}q_1(z_1)$ is a 
$+p$-transition with counter values at most $(B+\Upsilon_\C+\Gamma_\C)$
and the last transition $q_{n-1}(z_{n-1})\semi{\pi_{n-1},N}q_{n-1}(z_{n-1})$ is not a
 $-p$-transition with counter values at most $(B+\Upsilon_\C+\Gamma_\C)$,
 since the opposite case can be proven analogously.

Starting with $\eta^{(0)}=\pi$ we apply phase one of the $(+p,-p)$-lowering process to $\pi$
yielding a sequence of hybrid semiruns that all approximate $\pi$ with respect to level $B$
(and thus in particular 
--- bearing in mind that $\pi$ is $B$-hill ---
approximates $\pi$ with respect to level $B-\Gamma_\C-1$)
$$
\eta^{(0)},\eta^{(1)},\ldots,\eta^{(s-1)}
$$
in which (as above) $\eta^{(i)}$ is obtained from 
$\eta^{(i-1)}$ by lowering the 
rightmost unlowered $+p$ of $\eta^{(i-1)}$ however
only until reaching the hybrid semirun $\eta^{(s-1)}$
that contains precisely one unlowered $+p$-transition, namely
the first transition 
$q_0(z_0)\semi{\pi_0,N}q_1(z_1)$ of $\pi$, which
has counter values at most $(B+\Upsilon_\C+\Gamma_\C)$ by assumption.
It is important but straightforward to verify that despite the case we are in, 
it holds that $\eta^{(i)}$ approximates $\pi$ with respect to level $B$ 
(and also 
--- bearing in mind that $\pi$ is a $B$-hill ---
with respect to level $B-\Gamma_\C-1)$ for all $i\in[1,s-1]$.

Next, we will define a sequence of hybrid semiruns
$\eta^{(s)}=\eta^{(0)'},\eta^{(1)'},\ldots,\eta^{(t)'}$ 
in which $\eta^{(t)'}$ will be the desired $(N-\Gamma_\C)$-semirun
as required by the lemma.
For first defining $\eta^{(s)}=\eta^{(0)'}$ we make a case distinction for lowering the
only $+p$-transition 
$q_0(z_0)\semi{\pi_0,N}q_1(z_1)$ of $\eta^{(s-1)}$, which happens
to 
have counter values at most $(B+\Upsilon_\C+\Gamma_\C)$ by assumption.
For this assume $\eta^{(s-1)}$ has the following form
$$
\eta^{(s-1)}=\alpha^{(s-1,0)}\beta^{(s-1,1)}\alpha^{(s-1,1)}\cdots\beta^{(s-1,k_{s-1})}
\alpha^{(s-1,k_{s-1})},
$$
where we recall that $\beta^{(s-1,1)}$ equals $q_0(z_0)\semi{\pi_0,N}q_1(z_1)$.
Observe that the critical descending infix of $\beta^{(s-1,1)}$
could possibly be empty, for instance if $z_1\leq B+\Gamma_\C$.
We now make the following case distinction.
\begin{itemize}
	\item In case the critical descending infix of $\beta^{(s-1,1)}$ contains
		an unlowered $-p$-transition we define $\eta^{(s)}$ to be obtained
		from $\eta^{(s-1)}$ by lowering $\beta^{(s-1,1)}$ with
		the leftmost unlowered $-p$-transition inside the critical
		descending infix as above.
		Thus, $\eta^{(s)}$ no longer contains any unlowered $+p$-transition.
		It is again straightforward to verify that 
		$\eta^{(s)}$ approximates $\pi$ with respect to level $B$.
		Setting $\eta^{(0)'}=\eta^{(s)}$ we then construct the sequence
		$\eta^{(s)}=\eta^{(0)'},\eta^{(1)'},\ldots,\eta^{(t)'}$
		as usual, i.e. each $\eta^{(i)'}$ approximates $\pi$ with respect to level $B$ 
		and is obtained from $\eta^{(i-1)'}$ 
		by lowering the leftmost unlowered $-p$-transition
		and where eventually the breadth of $\eta^{(t)'}$ is $0$.
		Thus, as desired, the final $\eta^{(t)'}$ is an $(N-\Gamma_\C)$-semirun
		that is a min-rising and max-falling
		$B$-embedding of $\pi$ that has the same
		source and target configuration as $\pi$.
		Since $\eta^{(t)'}$ has the same source and target configuration as $\pi$ 
		it follows that $\eta^{(t)'}$ is also a 
		min-rising and max-falling
		$(B-\Upsilon_\C-\Gamma_\C-1)$-embedding of $\pi$
		as required by the lemma.
	\item In case the critical descending infix of $\beta^{(s-1,1)}$
		does not contain any unlowered $-p$-transition, we consider
		the shortest prefix $\zeta$ of the remaining suffix 
		$$\alpha^{(s-1,1)}\cdots\beta^{(s-1,k_{s-1})}\alpha^{(s-1,k_{s-1})}$$
		that ends in a configuration with counter value 
		at most $$\mathcal{L}'=z_1-\Upsilon_\C-1$$
		(where we recall that as above each $\alpha^{(i,j)}$ 
		is viewed as a sequence
		of transitions). 
		Indeed, we claim that $\zeta$ exists and moreover
		satisfies $\Delta(\zeta)<-\Upsilon_\C$. 
		Firstly, as $\pi$ is a $B$-hill by assumption, we 
		have that $z_1-z_n>\Upsilon_\C$. 
		Secondly, since $\eta^{(s-1)}$
		approximates $\pi$ with respect to level $B$ we have that 
		$\eta^{(s-1)}$ ends in a configuration
		with counter value $z_n$.
		Thus, 
		$\Delta(\alpha^{(s-1,1)}\cdots\beta^{(s-1,k_{s-1})}
		\alpha^{(s-1,k_{s-1})})=z_n-z_1<
		-\Upsilon_\C$ which implies that the prefix $\zeta$ exists
		and satisfies $\Delta(\zeta)<-\Upsilon_\C$.
		We make the following final case distinction.
		\begin{itemize}
			\item In case $\zeta$ contains an unlowered $-p$-transition,
		it must contain the leftmost unlowered $-p$-transition, namely 
				$\beta^{(s-1,2)}$.
		Similar as for the critical descending infix, 
				we define $\eta^{(s)}$ to be obtained from 
				$\eta^{(s-1)}$ by lowering $\beta^{(s-1,1)}$
				together with $\beta^{(s-1,2)}$.
				Here it is important to note that $\eta^{(s)}$ is not necessarily a 
				$(B-\Gamma_\C)$-embedding
				of $\pi$ since we cannot rule out the existence of configurations appearing 
				in $\alpha^{(s-1,1)}$ that have counter value 
				$B$.	Since
				$\eta^{(s-1)}$ was a $B$-embedding of the $B$-hill $\pi$
		with the same source and target configuration
				it follows however that $\eta^{(s)}$ is a
				$(B-\Gamma_\C-1)$-embedding of $\pi$.
				Hence, $\eta^{(s)}$ approximates $\pi$ with respect to 
				level $B-\Gamma_\C-1$.
				Thus, $\eta^{(s)}$ no longer contains any unlowered $+p$-transitions,
				however, possibly contains
				unlowered $-p$-transitions.
				Recalling that $\eta^{(0)'}=\eta^{(s)}$ we define 
				each of the remaining
				$\eta^{(i)'}$ to be obtained from $\eta^{(i-1)'}$
				as usual but by retaining that each $\eta^{(i)'}$ approximates
				$\pi$ with respect to level $B-\Gamma_\C-1$
				(instead of level $B$).
				By construction, $\eta^{(0)'}$ has breadth $0$ and thus is
				an $(N-\Gamma_\C)$-semirun that is a min-rising and max-falling
				$(B-\Gamma_\C-1)$-embedding
				and hence
				--- bearing in mind that $\pi$ is a $B$-hill ---
				in particular a 
				min-rising and
				max-falling
		$(B-\Upsilon_\C-\Gamma_\C-1)$-embedding
				of $\pi$ with the same source and target configuration as $\pi$.
			\item In case $\zeta$ does not contain any unlowered $-p$-transition
				it follows that $\zeta$ is a prefix of $\alpha^{(s-1,1)}$,
				thus contains neither unlowered $+p$-transitions nor
				unlowered $-p$-transitions but possibly lowered ones.
				By an analogous reasoning as Point 2 of Remark~\ref{alternate process}
				every occurrence of a lowered 
				$-p$-transition in $\alpha^{(s-1,1)}$ is preceded by 
				a unique corresponding lowered
				$+p$-transition again in $\alpha^{(s-1,1)}$.
				Thus, every 
				prefix of 
				$\phi(\zeta)$ contains 
				at least as many
				occurrences of 
				$\boldsymbol{[}$ as of 
				$\boldsymbol{]}$. 
				Recalling that $\Delta(\zeta)<-\Upsilon_\C$ we can hence
				apply the Bracket Lemma (Lemma~\ref{bracket lemma})
				and the Depumping Lemma (Lemma~\ref{lemma zero})
				to $\zeta$ as in phase one.
				The final $\eta^{(s)}$ is obtained from $\eta^{(s-1)}$
				by suitably shifting subsemiruns and cutting out certain subsemiruns 
				whose $\phi$-projection contains the same number of
				occurrences of $\boldsymbol{[}$ as of 
				$\boldsymbol{]}$.
				Similar as argued in the previous point it follows that 
				$\eta^{(s)}$ approximates
				$\pi$ with respect to level $B-\Gamma_\C-1$.
				Setting again $\eta^{(0)'}=\eta^{(s)}$ we define the 
				sequence of hybrid semiruns
				$\eta^{(0)'},\eta^{(1)'},\ldots,
				\eta^{(t)'}$ that all approximate $\pi$ with respect
				to level $B-\Gamma_\C-1$ analogously as done in the previous point.
				Again $\eta^{(t)'}$ is an
				$(N-\Gamma_\C)$-semirun that is a min-rising and max-falling
				$(B-\Gamma_\C-1)$-embedding
				(and hence in particular a 
	min-rising and max-falling
	$(B-\Upsilon_\C-\Gamma_\C-1)$-embedding
				of $\pi$ with the same source and target configuration as $\pi$.
			\end{itemize}
\end{itemize}

\begin{remark}\label{alternate process case 1}
	Our case (where the first transition
	of our $B$-hill
	is a $+p$-transition with counter values at most $B+\Upsilon_\C+\Gamma_\C$  
	and the last transition is not a $-p$-transition with counter values
	at most $(B+\Upsilon_\C+\Gamma_\C)$) allows the following ``penultimate'' 
	process variants for our $B$-hill $\pi$ (dual variants can be formulated in the case when $\pi$ is $B$-valley):
\begin{enumerate}
	\item 
	The adjusted process here in Case 1 bears similar properties to
	those of the $(+p,-p)$-lowering process seen in Remark~\ref{alternate process}.
	Specifically, 
		if $\phi(\pi)$ contains strictly more occurrences of 
			$\boldsymbol{]}$ than of $\boldsymbol{[}$ one can obtain a sequence
			of hybrid semiruns
					$$
		\eta^{(0)},\eta^{(1)},\ldots,\eta^{(s)},\eta^{(0)'},\eta^{(1)'},
		\ldots,\eta^{(t-1)'}, $$
					where all $\eta^{(i)}$ and $\eta^{(j)'}$ approximate $\pi$
					with respect to level $B-\Gamma_\C-1$
					(and are therefore 
					--- bearing in mind that $\pi$ is a $B$-hill ---
					in particular 
min-rising and max-falling
		$(B-\Upsilon_\C-\Gamma_\C-1)$-embeddings 
					of $\pi$
					with the same source and target configuration as $\pi$)
					and where
					$\eta^{(t-1)'}$ has breadth $1$ and contains precisely one 
					unlowered $-p$-transition.
	\item Dually, 
		the $(-p,+p)$-lowering process mentioned in Remark~\ref{alternate process},
		when applied to Case 1,
		is such that if
	 $\phi(\pi)$ contains strictly more occurrences of 
			$\boldsymbol{[}$ than of $\boldsymbol{]}$ one can obtain a sequence
			of hybrid semiruns
					$$
		\eta^{(0)},\eta^{(1)},\ldots,\eta^{(s)},\eta^{(0)'},\eta^{(1)'},\ldots,
		\eta^{(t-1)'},
					$$
					where all $\eta^{(i)}$ and $\eta^{(j)'}$ 
approximate $\pi$
					with respect to level $B-\Gamma_\C-1$
					(and are therefore 
					--- bearing in mind that $\pi$ is a $B$-hill ---
					in particular 
min-rising and max-falling $(B-\Upsilon_\C-\Gamma_\C-1)$-embeddings 
					of $\pi$
					with the same source and target configuration as $\pi$)
					and where $\eta^{(t-1)'}$ 
	has breadth $1$ and contains precisely one unlowered $+p$-transition.
\end{enumerate}
\end{remark}

\medskip

\noindent
{\em Case 2. The first transition $q_0(z_0)\semi{\pi_0,N}q_1(z_1)$ is a
$+p$-transition with counter values at most $(B+\Upsilon_\C+\Gamma_\C)$
and the last transition $q_{n-1}(z_{n-1})\semi{\pi_{n-1},N}q_{n-1}(z_{n-1})$ is a 
$-p$-transition with counter values at most $(B+\Upsilon_\C+\Gamma_\C)$.
}

By our case we have that $z_0,z_n\leq B+\Upsilon_\C+\Gamma_\C-N\leq
B+\Upsilon_\C+\Gamma_\C-M_\C< B-\Upsilon_\C-\Gamma_\C-1$,
where the last inequality follows the definition of our constants
on page~\pageref{constant definitions}.

Since $\pi$ is a $B$-hill $\pi$ is also a $(B-\Upsilon_\C-\Gamma_\C-1)$-hill.
Moreover, obviously there are no $+p$-transitions nor $-p$-transitions in $\pi$
whose source and target configuration both have counter value 
at most $B-1$. 
Phrased differently, setting $B'=B-\Upsilon_\C-\Gamma_\C-1$, we view $\pi$ as a $B'$-hill
that does not contain any $+p$-transitions
nor $-p$-transitions
whose source and target configuration have a counter value at most $(B'+\Upsilon_\C+\Gamma_\C)$.
We can hence apply the $(+p,-p)$-lowering process to $\pi$
as described in the case of in Section~\ref{Case A}
for $B'$ instead of $B$,
		thus yielding the sequence 
		$\eta^{(0)}, \eta^{(1)}, \ldots, \eta^{(s)}$ and 
		$\eta^{(0)'},\eta^{(1)'}, \ldots \eta^{(t)'}$ 
		of hybrid semiruns that approximate
		$\pi$ with respect to level $B'$ and are
		therefore min-rising and max-falling
		$B'$-embeddings of $\pi$: note that we use the 
		fact that they are are indeed $B'$-embeddings
		as the construction in Section~\ref{Case A}
		guarantees rather than that they
		are $(B'-\Upsilon_\C-\Gamma_\C-1)$-embeddings.
		The final $\eta^{(t)'}$ is of breadth 
		$0$ and is hence 
		a min-rising and max-falling
		$(N-\Gamma_\C)$-semirun
		that is a
		$(B-\Upsilon_\C-\Gamma_\C-1)$-embedding of $\pi$
		with the same source and target configuration
		as $\pi$,
		as required by the lemma.

\subsection{The Hill and Valley Lemma, dependent on
the number of occurrences $+p$-transitions as of $-p$-transitions}
\label{HV further}
A closer look at the proof of Lemma~\ref{lemma hill/valley}
	reveals that majority of occurrences
	of $+p$-transitions (resp. $-p$-transitions)
	implies the respective majority 
	is preserved in the resulting $(N-\Gamma_\C)$-semirun.

\begin{remark}\label{remark hill/valley}
	The resulting $\eta^{(t)'}$ obtained from the 
	$B$-hill (resp. $B$-valley) $\pi$ satisfies the following:
		\begin{itemize}
		\item If $\phi(\pi)$ contains at least as many occurrences of $\boldsymbol{[}$ as
			of $\boldsymbol{]}$, then so does the resulting $(N-\Gamma_\C)$-semirun
				$\phi(\eta^{(t)'})$ satisfying Lemma~\ref{lemma hill/valley}.
		\item If $\phi(\pi)$ contains at least as many occurrences of $\boldsymbol{]}$ as
			of $\boldsymbol{[}$, then so does 
				the resulting $(N-\Gamma_\C)$-semirun
				$\phi(\eta^{(t)'})$ satisfying Lemma~\ref{lemma hill/valley}.
	\end{itemize}
\end{remark}

The following final remark stresses the fact that when our $B$-hill (resp. $B$-valley) $\pi$ 
contains a number of occurrences of $+p$-transitions
different from the number of $-p$-transitions, 
the lowering processes described in the previous section yield
a penultimate hybrid semirun all but one of whose $+p$-transitions and $-p$-transitions
are lowered.
It is an immediate consequence of
Point 2.c in Remark~\ref{alternate process} and
Remark~\ref{alternate process case 1}.
\begin{remark}\label{hill bracket}
	Let $\pi$ be an $N$-semirun
that is a $B$-hill (dually a $B$-valley).
\begin{enumerate}
	\item If $\phi(\pi)$ contains strictly more occurrences of 
			$\boldsymbol{]}$ than of $\boldsymbol{[}$ one can obtain a sequence
			of hybrid semiruns
					$$
		\eta^{(0)},\eta^{(1)},\ldots,\eta^{(s)},\eta^{(0)'},\eta^{(1)'},\ldots,
		\eta^{(t-1)'},
					$$
					in which 
					$\eta^{(t-1)'}$ is a min-rising and max-falling
					$(B-\Upsilon_\C -\Gamma_\C-1)$-embedding
					(dually $(B+\Upsilon_\C +\Gamma_\C+1)$-embedding)
					of $\pi$
					with the same source and target configuration as $\pi$
					and where $\eta^{(t-1)'}$  
					has breadth $1$ and 
					contains precisely one unlowered $-p$-transition.
	\item Analogously, if $\phi(\pi)$ contains strictly more occurrences of 
		$\boldsymbol{[}$ than of $\boldsymbol{]}$ one can obtain a sequence
			of hybrid semiruns
					$$
		\eta^{(0)},\eta^{(1)},\ldots,\eta^{(s)},\eta^{(0')},\eta^{(1)'},\ldots,
		\eta^{(t-1)'},
					$$
					in which 
					$\eta^{(t-1)'}$ is a min-rising and max-falling
					$(B-\Upsilon_\C -\Gamma_\C-1)$-embedding
					(dually $(B+\Upsilon_\C +\Gamma_\C+1)$-embedding)
					of $\pi$
					with the same source and target configuration as $\pi$
					and where $\eta^{(t-1)'}$  
					has breadth $1$ and 
					contains precisely one unlowered $+p$-transition.

\end{enumerate}
\end{remark}

\section{The 5/6-Lemma}\label{5-6 section}
In this section, we introduce
the $5/6$-Lemma (Lemma~\ref{lemma 5/6}), stating that any $N$-semirun 
		with counter effect 
		smaller than $5/6 \cdot N$ can be turned into an $(N-\Gamma_\C)$-semirun
		that is moreover an $\ell$-embedding for all $\ell$ that are in
		distance at most $5/6\cdot N$ from the counter values of the source and target configuration.
		It will be the main technical ingredient
		in the proof of the Small Parameter Theorem (Theorem~\ref{theorem upper}).
This section is devoted to proving the lemma, hereby making extensive use of the Hill and Valley Lemma 
(Lemma~\ref{lemma hill/valley}), the Depumping Lemma (Lemma~\ref{lemma zero}),
and the Bracket Lemma (Lemma~\ref{bracket lemma}) introduced in previous sections.

Recall that we have fixed a 
POCA $\C=(Q,P,R,q_{init}, F)$ with $P=\{p\}$ 
along with the constants
$Z_C,\Gamma_\C,\Upsilon_\C,M_\C$ on page~\pageref{constant definitions}.

	\medskip

	\noindent
	Let us first introduce the $5/6$-Lemma.

\begin{lemma}[$5/6$-Lemma]{\label{lemma 5/6}}
        For all $N>M_\C$ and 
all $\ell\in\Z$ 
	and all $N$-semiruns $\pi$ from 
	$q_0(z_0)$ to $q_n(z_n)$ with
	$\values(\pi)\subseteq[0,4 N]$ satisfying
        $\max(z_0,z_n,\ell)-\min(z_0,z_n,\ell)\leq 5/6 \cdot N$ 
	there exists
	an $(N-\Gamma_\C)$-semirun $\pi'$ from $q_0(z_0)$ 
	to $q_n(z_n)$ 
	that is an $\ell$-embedding 
	of $\pi$ such that $\values(\pi')\subseteq[\min(\pi)-\Gamma_\C,\max(\pi)+\Gamma_\C]$.
\end{lemma}

\noindent
Towards proving Lemma~\ref{lemma 5/6} let us fix 
\begin{itemize}
	\item some $N>M_\C$,
	\item some $\ell\in\Z$,
	\item some $N$-semirun
$
\pi  \  = \  q_0(z_0)\semi{\pi_0,N}q_1(z_1) \ \cdots \ \semi{\pi_{n-1},N}q_n(z_n)
$
from $q_0(z_0)$ to $q_n(z_n)$ 
satisfying $\values(\pi)\subseteq[0,4 N]$ and
$\max(z_0,z_n,\ell)-\min(z_0,z_n,\ell)\leq 5/6 \cdot N$.
\end{itemize}

In order to prove the 5/6-Lemma we need to show the existence of some $(N-\Gamma_\C)$-semirun $\pi'$
from $q_0(z_0)$ to $q_n(z_n)$
that is both an $\ell$-embedding of $\pi$ 
with $\values(\pi')\subseteq[\min(\pi)-\Gamma_\C,\max(\pi)+\Gamma_\C]$.

\medskip

For this, let us define following two constants

\newcommand{\Bmin}{B_{\min}}
\newcommand{\Bmax}{B_{\max}}
\label{Bmin Bmax def}
\fbox{\parbox[t][1cm][c]{12cm}{
	$B_{\min}=\min(z_0,z_n,\ell)-\Upsilon_\C-2\Gamma_\C-1\quad$ and 
			$\quad B_{\max}=\max(z_0,z_n,\ell)+\Upsilon_\C+2\Gamma_\C+1$
	}
}
\vspace{0.2cm}

and observe that 
\begin{eqnarray}
	\Bmax-\Bmin&=&\max(z_0,z_n,\ell)-\min(z_0,z_n,\ell)+2\Upsilon_\C+4\Gamma_\C+2\nonumber\\
	&\leq&5/6\cdot N+2\Upsilon_\C+4\Gamma_\C+2\nonumber\\
	&\leq&5/6\cdot N+ M_\C / 6 \label{5/6 upper}\\
	&<&N,\label{upper N}
\end{eqnarray}
where the penultimate inequality follows from the definitions
of our constants on page~\pageref{constant definitions}.

\label{B constants}

\vspace{0.1cm}
\medskip

\noindent
We are particularly interested in subsemiruns of $\pi$ that start
and end in configurations with counter values in $[\Bmin+1,\Bmax-1]$.
To categorize such subsemiruns into different types, we introduce the notion of crossing and doubly-crossing
transitions.

\begin{samepage}
\begin{definition}
	A transition $q_i(z_i)\semi{\pi_i}q_{i+1}(z_{i+1})$ 
	is called {\em crossing} if either
		\begin{itemize}
		\item $\pi_i = +p$ 
				 and we have
					$z_i<\Bmax\leq z_{i+1}$ or 
					$z_i\leq \Bmin < z_{i+1}$, or
		\item $\pi_i = -p$ 
					and we have $z_i>\Bmin\geq z_{i+1}$ or
					$z_i\geq \Bmax> z_{i+1}$.
			\end{itemize}
	If even moreover
	$z_i\leq\Bmin\leq\Bmax\leq z_{i+1}$ or
	$z_i\geq\Bmax\geq\Bmin\geq z_{i+1}$
	we call $\pi_i$ {\em doubly-crossing}.
\end{definition}
\end{samepage}

We already refer to Figure~\ref{type example}, where subsemiruns of
a certain type (to be defined below) are depicted, some of whose transitions
crossing transitions, some of whose are even doubly-crossing transitions.

Next, we introduce three particular types of subsemiruns of $\pi$ 
starting and ending in configurations with counter values in $[\Bmin+1,\Bmax-1]$.

\begin{definition}[Type I, II and III subsemiruns of $\pi$]
	A subsemirun $\pi[a,b]$ of $\pi$ with source and target configuration
	in $Q\times[\Bmin+1,\Bmax-1]$ is
	\begin{itemize}
		\item of {\em Type I} if $\values(\pi[a,b])\subseteq[\Bmin+1,\Bmax-1]$,
		\item of {\em Type II} if
			\begin{itemize}
			\item $\values(\pi[a+1,b-1])\cap[\Bmin+1,\Bmax-1]=\emptyset$, and
			\item $\pi[a,b]$ does not contain any doubly-crossing transitions, 
			\end{itemize}
		\item of {\em Type III} if
			\begin{itemize}
			\item $\values(\pi[a+1,b-1])\cap[\Bmin+1,\Bmax-1]=\emptyset$, and
			\item $\pi[a,b]$ contains at least one doubly-crossing
			transition.
			\end{itemize}
	\end{itemize}
\end{definition}

\begin{remark}\label{remark Type III}
All crossing transitions in a Type III semirun,
except possibly the first or the last transition, are doubly-crossing.
\end{remark}

Figure~\ref{type example} shows an example of a Type II and of a Type III subsemirun.

\begin{center}
	\begin{figure}
		\hspace{0.3cm}
\includegraphics[width=0.44\textwidth]{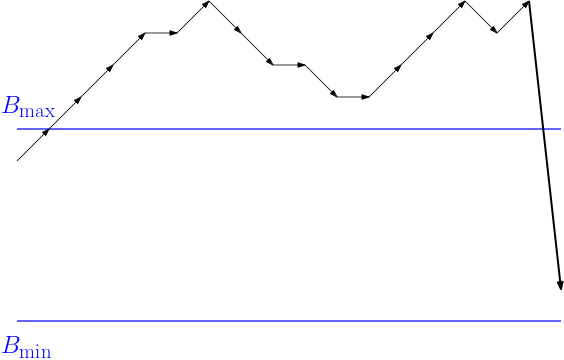}
		\hspace{0.4cm}
		\raisebox{-1.94ex}{\includegraphics[width=0.44\textwidth]{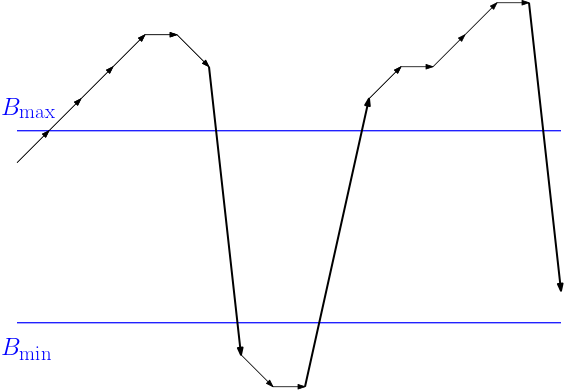}}
\caption{On the left, an example of a Type II subsemirun. On the right, an example of a Type III subsemirun.
		Bold transitions are crossing, and the first two bold transitions
		of the figure on the right
 are moreover doubly crossing.}
\label{type example}
	\end{figure}
\end{center}

\begin{samepage}
The following lemma factorizes $\pi$ into Type I, Type II and Type III subsemiruns, bearing
	in mind that both the source and target configuration of $\pi$ have
	a counter value in $[\Bmin+1,\Bmax-1]$. 
\begin{lemma}\label{lemma type}
The $N$-semirun $\pi$ can be factorized into Type I, Type II, and Type III subsemiruns.
\end{lemma}
\begin{proof}
Let us first factorize $\pi$ as
\begin{eqnarray}
	\pi &=& \pi[c_1,d_1]\pi[d_1,c_2]\pi[c_2,d_2]\pi[d_2,c_3]\quad\cdots\quad\pi[c_t,d_t],
 \label{fact type}
\end{eqnarray}
	where
	\begin{itemize}
		\item $\pi[c_i,d_i]$ are Type I and {\em maximal} (possibly empty),
			i.e. $\pi[c_i,d_i]$ is of Type I 
			but neither $\pi[c_i-1,d_i]$ nor
			$\pi[c_i,d_i+1]$ is of Type I 
			for all $i\in[1,t]$, and
		\item $\values(\pi[d_i+1,c_{i+1}-1])\cap[\Bmin+1,\Bmax-1]=\emptyset$
			for all $i\in[1,t-1]$.
	\end{itemize}
	Now it suffices to show that each subsemirun $\pi[d_i,c_{i+1}]$ is either of Type II or of Type III. For this, let us make a case distincton on whether $\pi[d_i,c_{i+1}]$ contains a doubly-crossing transition or not.

	For the first case, namely that $\pi[d_i,c_{i+1}]$ does contain a doubly-crossing
	transition, since $\values(\pi[d_i+1,c_{i+1}-1])\cap[\Bmin+1,\Bmax-1]=\emptyset$, 
	we have that
	$\pi[d_i,c_{i+1}]$ is of Type III by definition.

	For the second case, namely that $\pi[d_i,c_{i+1}]$ does not contain any doubly-crossing
	transition,
	since $\values(\pi[d_i+1,c_{i+1}-1])\cap[\Bmin+1,\Bmax-1]=\emptyset$ we have that
	$\pi[d_i,c_{i+1}]$ is of Type II by definition.
\end{proof}
\end{samepage}

\begin{center}
	\begin{figure}
\includegraphics[width=0.8\textwidth]{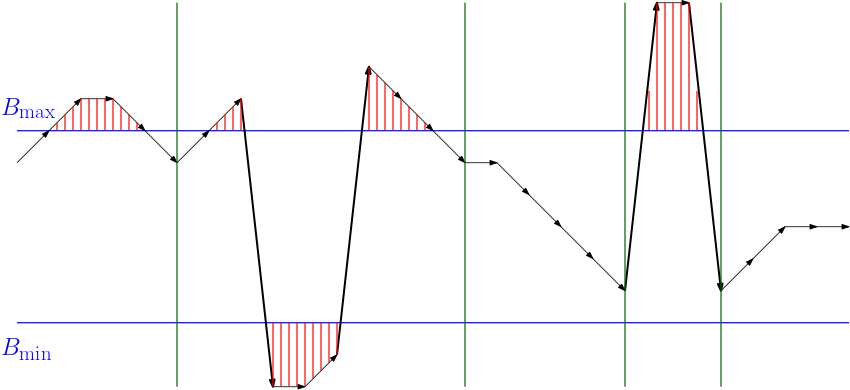}
		\caption{In this figure, we provide an example factorization of a semirun $\pi$. A semirun $\pi$ is divided into five subsemiruns, separated by 
		vertical lines. The third and fifth subsemiruns are of Type I, 
		the first and fourth subsemiruns are of Type II, 
		and the second one is of Type III.}
\label{type factor}
	\end{figure}
\end{center}

	By Remark~\ref{remark embeddings}
	in order to prove the 
	existence of the desired $(N-\Gamma_\C)$-semirun
	it suffices to show it
	for Type I, Type II and Type III subsemiruns of $\pi$.

	Since Type I subsemiruns neither contain any $+p$-transition nor any 
	$-p$-transition by (\ref{upper N}),
	they are already $(N-\Gamma_\C)$-semiruns.

	Let us now discuss the situation for Type II subsemiruns of $\pi$.
	If a Type II subsemirun is already a Type I subsemirun we are done as above.
	In case a Type II subsemirun $\rho$ is not of Type I we first claim that 
	$\rho$ is either a $\Bmin$-valley or a $\Bmax$-hill.
	Indeed, if $\rho$ is of Type II but not of Type I one can factorize $\rho$ 
	as 
$$
\rho\quad=\quad p_0(x_0) \quad \semi{\rho_0,N} \quad p_1(x_1)\quad\cdots\quad \semi{\rho_{m-1},N}\quad p_m(x_m),
$$
	where 
	\begin{samepage}
	\begin{enumerate}
		\item $m\geq 2$,
		\item $x_0,x_m\in[\Bmin+1,\Bmax-1]$, and
		\item either $x_i\in[0,\Bmin]$ for all $i\in[1,m-1]$ or $x_i\in[\Bmax,4N]$ for all $i\in[1,m-1]$,
	\end{enumerate}
		\end{samepage}
	where Point 3 follows from the absence of doubly-crossing transitions.

	First assume that $x_i\in[\Bmax,4N]$ for all $i\in[1,m-1]$.
	In this case any $+p$-transition (resp. $-p$-transition) ends (resp. starts) in a configuration with counter value strictly larger than $\Bmin+N$. 
	Due to the definition of our constants on page~\pageref{constant definitions} we have

\begin{eqnarray*}
	x_0+\Upsilon_\C, \ x_n+\Upsilon_\C&<&\Bmax+\Upsilon_\C\\
	&=&\Bmin+(\Bmax-\Bmin)+\Upsilon_\C \\
	&\leq&\Bmin + 5/6\cdot N+3\Upsilon_\C+4\Gamma_\C+2\\
	&<&\Bmin + 5/6\cdot N+ M_\C / 6\\
	 &<&\Bmin+N,
\end{eqnarray*}

	hence $\rho$ is a $\Bmax$-hill.

	Secondly, in case $x_i\in[0,\Bmin]$ for all $i\in[1,m-1]$ it can analogously be shown
	that $\rho$ is $\Bmin$-valley.
	
	The existence of the desired 
	$(N-\Gamma_\C)$-semirun 
	$\rho'$
	that is an $\ell$-embedding of the Type II semirun
	$\rho$ 
	with the same source and target configuration as $\rho$
	follows immediately from the following claim, which 
	itself (with a short justificaton below)
	is a consequence of the Hill and 
	Valley Lemma (Lemma~\ref{lemma hill/valley}); 
	thanks to the fact that the Hill and Valley Lemma
	guarantees the resulting $(N-\Gamma_\C)$-semiruns
	to be min-rising and max-falling, we can 
	even guarantee 
	$\values(\rho')\subseteq[\min(\rho),\max(\rho)]$.

\begin{claim}\label{claim embedding}
	For every $N$-semirun $\rho$ that is either
	a $B$-hill with $B \geq \ell+ \Upsilon_\C + \Gamma_\C +1$
or a $B$-valley with $B \leq \ell - \Upsilon_\C - \Gamma_\C-1$, 
	there exists an $(N-\Gamma_\C)$-semirun
	that is both a min-rising and max-falling $\ell$-embedding of $\rho$ with same source
	and target configuration as $\rho$.
\end{claim}
That the Hill and Valley Lemma produces an $\ell$-embedding
that has the same source and target configuration is
important here.
	Indeed, generally speaking if $\rho$ is any $B$-hill 
	and $\rho'$ is any $k$-embedding of $\rho$
	with the same source and target configuration
	as $\rho$ and where $k<B$, then $\rho'$ is also
	a $k'$-embedding of $\rho$ for all $k'<k$.
	Dually, if $\rho$ is any $B$-valley 
	and $\rho'$ is a $k$-embedding of $\rho$ with
	with the same source and target
	configuration as $\rho$ and where $k>B$,
	then $\rho'$ is also
	an $k'$-embedding of $\rho$ for all $k'>k$.

	\bigskip

	For the rest of this section it now 
	suffices to prove that
	for every Type III subsemirun $\rho$ of $\pi$
	there exists an $(N-\Gamma_\C)$-semirun 
	$\rho'$ that 
	is an $\ell$-embedding of $\rho$
	with the same source and target configuration as $\rho$
	and that moreover satisfies
	$\values(\rho')\subseteq[\min(\rho)-\Gamma_\C,
	\max(\rho)+\Gamma_\C]$.

\subsection{Lowering Type III subsemiruns}\label{section typeII}

For the rest of the section let us fix a Type III subsemirun $\rho$ of $\pi$. Let us factorize $\rho$ 
by its crossing semitransitions, 
i.e. as
\begin{eqnarray*}
	\rho \quad=&& \alpha^{(0)} \beta^{(1)} \alpha^{(1)} \quad\cdots\quad\beta^{(n)} \alpha^{(n)},
\end{eqnarray*}

\noindent
where $\beta^{(1)},\ldots,\beta^{(n)}$ is an enumeration of the crossing semitransitions
 of $\rho$ 
 and each $\alpha^{(i)}$ is a (possibly empty) $N$-subsemirun of $\rho$. 
It is worth mentioning that, indeed abusing notation, for the rest of this section 
we refer to $n$ as the number of crossing semitransitions of $\rho$, rather than
the number of transitions of our original $N$-run $\pi$.

An example factorization is shown in Figure~\ref{type example},
where the crossing transitions are depicted in bold. 
 We remark that the only crossing transitions
 of $\rho$ that are not doubly-crossing can possibly
 only be the first or the last one (or both).

We intend to now factorize $\rho$, if possible, into hills and valleys. 
In order to do this let us first introduce the notions of $B$-hill candidate
and $B$-valley candidate.

\begin{definition}
	Let 
	$$\chi\quad=\quad p_0(x_0)\semi{\chi_0,N}p_1(x_1)\quad\cdots\quad\semi{\chi_{m-1},N}p_m(x_m)$$
be an $N$-semirun.
	We say $\chi$ is a {\em $B$-hill candidate} if $x_0,x_m<B$ and $x_i\geq B$ for all $i\in[1,m-1]$,
	respectively a {\em $B$-valley candidate}
	if $x_0,x_m>B$ and $x_i\leq B$ for all $i\in[1,m-1]$.
\end{definition}
Note that every $B$-hill is a $B$-hill candidate 
but not vice versa, since being a $B$-hill moreover requires 
$+p$-transitions to end at a configuration with counter value strictly larger than $x_n+\Upsilon_\C$
and $-p$-transitions to start at a configuration with counter value strictly larger than $x_0+\Upsilon_\C$.
A similar remark applies to $B$-valleys and $B$-valley candidates.

For the rest of this section we assume {\em without loss of generality}
that the crossing transition 
$\beta_1$ is a $+p$-transition. 
The case when $\beta_1$ is $-p$-transition
can be proven analogously.

It follows that if the number $n$ of crossing transitions is even, then
there is a unique factorization 
\begin{eqnarray}
\rho\quad=\quad \alpha^{(0)}\sigma^{(1)}
\alpha^{(2)}\sigma^{(2)}\alpha^{(4)}\sigma^{(3)}\alpha^{(6)}\quad\cdots\quad\sigma^{(n/2)}\alpha^{(n)},
	\label{hill factorization}
\end{eqnarray}
where $\sigma^{(i)}=\beta^{(2i-1)}\alpha^{(2i-1)}\beta^{(2i)}$ is a $\Bmax$-hill candidate,
$\beta^{(2i-1)}$ is a $+p$-transition and $\beta^{(2i)}$ is a $-p$-transition for all $i\in[1,n/2]$.
Indeed, this immediately follows from the definition of crossing transitions
and the fact that $\alpha^{(2i-1)}$ does not contain any configuration with counter value 
strictly less than $\Bmax$.

Therefore our proof makes first a case distinction on the parity of the number $n$ of crossing transitions.

\medskip

\noindent
{\em Case A: The number of crossing transitions $n$ is even.}

Our proof next makes a case distinction on the number of 
$\Bmax$-hill candidates $\sigma^{(i)}$ in the factorization~(\ref{hill factorization})
that are in fact $\Bmax$-hills.

\medskip

\noindent
{\em Case A.1: All of the $\Bmax$-hill candidates 
$\sigma^{(i)}$ 
in (\ref{hill factorization}) are in fact $\Bmax$-hills.}\\[-0.2cm]

Since each $\sigma^{(i)}=\beta^{(2i-1)}\alpha^{(2i-1)}\beta^{(2i)}$ from (\ref{hill factorization}) 
 is a $\Bmax$-hill, to each $\sigma^{(i)}$ we can apply Claim~\ref{claim embedding}
 and obtain
an $(N-\Gamma_\C)$-semirun $\widehat{\sigma^{(i)}}$ that is both a min-rising and max-falling $\ell$-embedding
of $\sigma^{(i)}$ with the same source and target configuration as $\sigma^{(i)}$.
Thus, it remains to show the same for $\alpha^{(2i)}$ for each $i\in[0,n/2]$.
We do this separately for $\alpha^{(0)},\alpha^{(n)}$ and finally
for those $\alpha^{(2i)}$, where $i\in[1,n/2-1]$.

Let us first show it for $\alpha^{(0)}$.
The proof for $\alpha^{(n)}$ is completely analogous.
If $\alpha^{(0)}$ is empty (which would imply that $\beta^{(1)}$ is crossing but not doubly-crossing),
there is nothing to show. 
Let us therefore assume that $\alpha^{(0)}$ is not empty. It
follows that $\beta^{(1)}$ must be a doubly-crossing 
$+p$-transition by  Remark~\ref{remark Type III}.
Since $\beta^{(1)}$ is the first crossing transition (even doubly-crossing) 
and a $+p$-transition and moreover
$\rho$ is of Type III 
one can factorize $\alpha^{(0)}$ as 
$$
\alpha^{(0)}\quad=\quad\alpha^{(0,0)}\sigma^{(0,1)}\alpha^{(0,1)}\quad\cdots\quad\sigma^{(0,k)},
$$
where $\alpha^{(0,j)}$ satisfies $\values(\alpha^{(0,j)})\subseteq[\Bmax-N,\Bmin+1]$
for all $j\in[0,k]$
and $\sigma^{(0,j)}$ is a $(\Bmax-N-1)$-valley candidate for all $j\in[1,k]$.
It immediately follows that each $\alpha^{(0,j)}$ does not contain any $+p$-transition
nor any $-p$-transition, and is hence already an $(N-\Gamma_\C)$-semirun.
Finally we claim that each $\sigma^{(0,j)}$ is in fact a $(\Bmax-N-1)$-valley.
Indeed, firstly the target configuration of each $+p$-transition in $\sigma^{(0,j)}$ 
has a counter value at most $\Bmin$ 
and hence a source configuration with counter value at most 
$\Bmin-N<\Bmax-N-1-\Upsilon_\C$, where the inequality follows
from definition of $\Bmin$ and $\Bmax$ from page~\pageref{B constants}.
Secondly and analogously, the source configuration of each $-p$-transition in $\sigma^{(0,j)}$
has a counter value of at most $\Bmin$ and hence
a target configuration with counter value at most $\Bmin-N<\Bmax-N-1-\Upsilon_\C$.
Thus, to each $\sigma^{(0,j)}$ we can apply Claim~\ref{claim embedding} to obtain 
an $(N-\Gamma_\C)$-semirun $\widehat{\sigma^{(0,j)}}$ that is a min-rising 
and max-falling
$\ell$-embedding of $\sigma^{(0,j)}$ with the
same source and target configuration as $\sigma^{(0,j)}$.
Hence, by appropriately concatenationg the $\alpha^{(0,j)}$ with the
$\widehat{\sigma^{(0,j)}}$ we obtain the desired
$(N-\Gamma_\C)$-semirun that is an $\ell$-embedding of $\alpha^{(0)}$.

It now only remains and suffices to show that
for each $\alpha^{(2i)}$ with $i\in[1,n/2-1]$, 
that there exists an $(N-\Gamma_\C)$-semirun that is a
min-rising and max-falling $\ell$-embedding of $\alpha^{(2i)}$ with same
source and target configuration as $\alpha^{(2i)}$.
Again by Remark~\ref{remark Type III} any such $\alpha^{(2i)}$ 
succeeds $\sigma^{(i)}=\beta^{(2i-1)}\alpha^{(2i-1)}\beta^{(2i)}$
and thus succeeds the doubly-crossing transition $\beta^{(2i)}$ and 
analogously
preceeds
$\sigma^{(i+1)}=\beta^{(2i+1)}\alpha^{(2i+1)}\beta^{(2i+2)}$ and
thus precedes the doubly-crossing $\beta^{(2i+1)}$.
Therefore, analogously as done for $\alpha^{(0)}$, one can factorize $\alpha^{(2i)}$
as
$$
\alpha^{(2i)}\quad=\quad\alpha^{(2i,0)}\sigma^{(2i,1)}\alpha^{(2i,1)}\quad\cdots\quad
\sigma^{(2i,k)}\alpha^{(2i,k)}
$$
for some $k\geq 0$, where $\alpha^{(2i,j)}$ satisfies
$\values(\alpha^{(2i,j)})\subseteq[\Bmax-N,\Bmin]$
(and is thus already an $(N-\Gamma_\C)$-semirun)
for all $j\in[0,k]$
and $\sigma^{(2i,j)}$ is a $(\Bmax-N-1)$-valley candidate
that is in fact a $(\Bmax-N-1)$-valley
for all $j\in[1,k]$.
\medskip

\noindent
{\em Case A.2: All but one of the $\Bmax$-hill candidates 
$\sigma^{(i)}$ 
in
(\ref{hill factorization}) are in fact $\Bmax$-hills.}

Recall the factorization
$$\rho\quad=\quad \alpha^{(0)}\sigma^{(1)}
\alpha^{(2)}\sigma^{(2)}\alpha^{(4)}\sigma^{(3)}\alpha^{(6)}\quad\cdots\quad\sigma^{(n/2)}\alpha^{(n)}$$
from (\ref{hill factorization}) 
where each
$\sigma^{(i)}=\beta^{(2i-1)}\alpha^{(2i-1)}\beta^{(2i)}$ is a $\Bmax$-hill candidate
for all $i\in[1,n/2]$.

First let us show that, in case one such $\Bmax$-hill candidate is not a $\Bmax$-hill, then it must be either $\sigma^{(1)}$  or $\sigma^{(n/2)}$. 
For every of the remaining $j\in[2,n/2-1]$ we have that $\sigma^{(j)}=\beta^{(2j-1)}\alpha^{(2j-1)}\beta^{(2j)}$ 
is such that $\beta^{(2j-1)}$ and $\beta^{(2j)}$ are both doubly-crossing,
implying that $\sigma^{(j)}$ is a $\Bmax$-hill:
indeed, both the source and target configuration of $\sigma^{(j)}$ have a counter value at most $\Bmin<\Bmax-\Upsilon_\C$,
which is sufficient since every $+p$-transition (resp. $-p$-transition) of $\sigma^{(j)}$ ends (resp. starts)
in a configuration with counter value at least $\Bmax$.

Let us assume without loss of generality that 
$\sigma^{(1)}=\beta^{(1)}\alpha^{(1)}\beta^{(2)}$ is the only $\Bmax$-hill candidate that is not a $\Bmax$-hill,
the case when $\sigma^{(n/2)}$ is not a $\Bmax$-hill can be treated analogously.

Clearly, by the above reasoning, either $\beta^{(1)}$ or $\beta^{(2)}$ must not be doubly-crossing.
Without loss of generality let us assume that the first crossing transition $\beta^{(1)}$ is not doubly-crossing
(for $\beta^{(2)}$ not doubly-crossing implies $n=2$ by definition of Type III
and Remark~\ref{remark Type III};
this case is thus included in the dual case when 
$\sigma^{(n/2)}$ is not a $\Bmax$-hill
and the last crossing transition $\beta^{(n)}$ is not doubly-crossing).

Remarking that $\alpha^{(0)}$ must be empty by our case, one can now factorize our $N$-semirun $\rho$ as
$$
\rho\quad=\quad \left( \beta^{(1)}\alpha^{(1)}\right)\beta^{(2)}\left(\alpha^{(2)}\sigma^{(2)}\alpha^{(4)}\quad\cdots\quad\sigma^{(n/2)}\alpha^{(n)}\right),
$$
where $\beta^{(2)}$ is a $-p$-transition.
For finishing this case we will proceed as follows.
\begin{enumerate}
	\item Firstly we show the existence of an $(N-\Gamma_\C)$-semirun that is
both a min-rising and max-falling $\ell$-embedding of $\beta^{(1)}\alpha^{(1)}$ with the
same source and target configuration as $\beta^{(1)}\alpha^{(1)}$.
\item Secondly, let us assume that $\beta^{(2)}$ is an $N$-semirun 
	from $q(x)$ to $q'(y)$, say, and moreover that
	$\alpha^{(2)}\sigma^{(2)}\alpha^{(4)}\cdots \ \sigma^{(n/2)}\alpha^{(n)}$ is an $N$-semirun
		from $q'(y)$ to $q''(z)$, say.
		Noting that $\beta^{(2)}=q(x)\semi{-p,N}q'(y)$, we explicitly lower $\beta^{(2)}$
	into the $(N-\Gamma_\C)$-semirun $q(x)\semi{-p,N-\Gamma_\C}q'(y+\Gamma_\C)$,
		which is --- since $\beta^{(2)}$ is doubly-crossing --- obviously both a min-rising and max-falling
		$\ell$-embedding of $\beta^{(2)}$ from $q(x)$ to $q'(y+\Gamma_\C)$.
	Finally, we show the existence of an $(N-\Gamma_\C)$-semirun that 
		is an
		$\ell$-embedding of 
		$\alpha^{(2)}\sigma^{(2)}\alpha^{(4)}\cdots \ \sigma^{(n/2)}\alpha^{(n)}$
		from $q'(y+\Gamma_\C)$ to 
		$q''(z)$ all of whose counter values lie
		in $[\min(\rho)-\Gamma_\C,\max(\rho)+\Gamma_\C]$.
\end{enumerate}

Let us first show Point 1.
Since $\beta^{(1)}$ is not doubly-crossing
and $\sigma^{(1)}=\beta^{(1)}\alpha^{(1)}\beta^{(2)}$ is a $\Bmax$-hill candidate that is not a $\Bmax$-hill,
the only reason for the latter is the existence
of a $-p$-transition $\tau$
such that the counter values of the source configuration of $\tau$ and $\sigma^{(1)}$ have an absolute difference 
at most $\Upsilon_\C$. 
Such a transition $\tau$ must have a source configuration with counter value in the interval 
$[\Bmax, \Bmax + \Upsilon_\C-1]$ and $\sigma^{(1)}$ must then necessarily have a source configuration with a counter
value in the interval $[\Bmax-\Upsilon_\C,\Bmax-1]$. 
Such a violation can only happen for 
$\tau = \beta^{(2)}$.
As a consequence, the target configuration 
$q(x)$ of $\alpha^{(1)}$ has a counter value inside $[\Bmax,\Bmax+\Upsilon_\C-1]$.
Recalling that $\beta^{(1)}$ is not doubly-crossing
one can (analogously as has been done in Case A.1) factorize $\beta^{(1)}\alpha^{(1)}$ as
$$
\beta^{(1)} \alpha^{(1)}\quad=\quad \chi^{(1)}\xi^{(1)}\cdots\chi^{(k)}\xi^{(k)},
$$
for some $k\geq 0$, where each $\chi^{(i)}$ is a $(\Bmin+N+1)$-hill 
and each $\xi^{(i)}$ satisfies 
$\values(\xi^{(i)})\subseteq[\Bmax,\Bmin+N]$.
Again, analogously as has been done in Case A.1, to each of the $\chi^{(i)}$ we can 
apply Claim~\ref{claim embedding} to turn them into a suitable
min-rising and max-falling $(N-\Gamma_\C)$-semirun
that is an $\ell$-embedding of $\chi^{(i)}$ with the same source and target configuration 
as $\chi^{(i)}$,
whereas each of the $\xi^{(i)}$ are already $(N-\Gamma_\C)$-semiruns since they do not 
contain any $+p$-transitions nor $-p$-transitions.
\medskip

Let us finally show Point 2.
Consider the remaining factorization
\begin{eqnarray}\label{rho suffix}
\gamma\quad=\quad\alpha^{(2)}\sigma^{(2)}\alpha^{(4)}\quad\cdots\quad\sigma^{(n/2)}\alpha^{(n)}
\end{eqnarray}
from $q'(y)$ to $q''(z)$.
We need prove the existence of an $(N-\Gamma_\C)$-semirun 
from $q'(y+\Gamma_\C)$ to $q''(z)$ that is an $\ell$-embedding of $\gamma$
and whose counter values lie in $[\min(\gamma)-\Gamma_\C,\max(\gamma)+\Gamma_\C]$.

	We first claim that $\Delta(\gamma)>\Upsilon_\C$. 
	Since $x\in[\Bmax,\Bmax+\Upsilon_\C-1]$
		it follows that the counter value $y$ of $\gamma$'s source
		configuration $q'(y)$ satisfies
	$y\in[\Bmax-N,\Bmax+\Upsilon_\C-1-N]$.
	Moreover, the target configuration $q''(z)$ of $\gamma$ is the target configuration
	of our Type III $N$-semirun $\rho$, thus $z\in[\Bmin+1,\Bmax-1]$.
	Hence by the definition of our constants on page~\pageref{constant definitions} we have
\begin{eqnarray*}
	\Delta(\gamma)&\geq& \Bmin+1-(\Bmax+\Upsilon_\C-1-N)\\
	&>& N-(\Bmax-\Bmin)-\Upsilon_\C\\
	&\stackrel{(\ref{5/6 upper})}{\geq}&N-(5/6 \cdot N+2\Upsilon_\C+4\Gamma_\C+2)-\Upsilon_\C\\
	&=&N/6-2\Upsilon_\C-4\Gamma_\C-2-\Upsilon_\C\\
	&>&M_\C/6-2\Upsilon_\C-4\Gamma_\C-2-\Upsilon_\C\\
	&=& M_\C/6-4\Upsilon_\C-4\Gamma_\C-2+\Upsilon_\C\\
	&>& (M_\C/6-4(\Upsilon_\C+\Gamma_\C+1))+\Upsilon_\C\\
	&>&\Upsilon_\C.
\end{eqnarray*}

Recall that each $\sigma^{(i)}$ is a $\Bmax$-hill for all $i\in[2,n/2]$
by our case.
Analogously, as has been done in Case A.1, for each $i\in[1,n/2]$ one can 
factorize $\alpha^{(2i)}$ as
$$
\alpha^{(2i)}\quad=\quad\ \alpha^{(2i,0)}\sigma^{(2i,1)}\quad\cdots\quad\sigma^{(2i,k_i)}\alpha^{(2i,k_i)},
$$
where $\alpha^{(2i,j)}$ satisfies $\values(\alpha^{(2i,j)})
\subseteq[\Bmax-N,\Bmin+1]$ for each $j\in[0,k_i]$
and $\sigma^{(2i,j)}$ is a $(\Bmax-N-1)$-valley for each 
$j\in[1,k_i]$: more precisely for the final 
$\alpha^{(n)}$ we have
$\values(\alpha^{(n)})\subseteq[\Bmax-N,\Bmin+1]$, 
however for all $i\in[1,n/2-1]$
we have $\alpha^{(2i)}\subseteq[\Bmax-N,\Bmin]$.

It is important to remark that each $\alpha^{(2i,j)}$ is 
already an $(N-\Gamma_\C)$-semirun
since it does not contain any $+p$-transition nor any $-p$-transition. 
The following remark summarizes the 
factorization of $\gamma$.
\begin{remark}
Our $N$-semirun $\gamma$ from $q'(y)$ to $q''(z)$  can be written as
	\begin{eqnarray}\label{gamma factorization}
\gamma\quad=\quad\alpha^{(2)}\sigma^{(2)}\alpha^{(4)}\quad\cdots\quad\sigma^{(n/2)}\alpha^{(n)}\quad=\quad
\alpha^{(2)} \left(\prod_{i=2}^{n/2}\sigma^{(i)} \ \alpha^{(2i,0)} \left(\prod_{j=1}^{k_i}
\sigma^{(2i,j)}\alpha^{(2i,j)}\right)
\right),
	\label{gamma fact}
\end{eqnarray}
	
\noindent
where
\begin{enumerate}		
\item each $\sigma^{(i)}$ is a $\Bmax$-hill,
\item each $\sigma^{(2i,j)}$ is a $(\Bmax-N-1)$-valley, 
\item each $\alpha^{(2i,j)}$ is already an
	$(N-\Gamma_\C)$-semirun, 
\item $\Delta(\gamma)>\Upsilon_\C$, and
\item every configuration in $\gamma$ (except for possibly the target
	configuration $q''(z)$) has a counter value 
	in $[0,\Bmin]\cup[\Bmax,4N]$ whose absolute difference
		with $\ell$ is thus strictly larger than $\Gamma_\C$
		(recall the definition of $\Bmin$ and $\Bmax$ of page~\pageref{Bmin Bmax def}).
\end{enumerate}
\end{remark}

We can now obtain suitable $(N-\Gamma_\C)$-semiruns with the same source
and target configuration for any of the above hills and valleys
by applying the Hill and Valley Lemma.
\begin{remark}\label{remark widehat}
By applying the Hill and Valley Lemma 
	(Lemma \ref{lemma hill/valley}) 
we obtain the following.
	\begin{enumerate}
		\item For each of the $\Bmax$-hills $\sigma^{(i)}$ there exists
			an $(N-\Gamma_\C)$-semirun $\widehat{\sigma^{(i)}}$
			that is both a min-rising and max-falling
			$(\Bmax-\Upsilon_\C-\Gamma_\C-1)$-embedding of $\sigma^{(i)}$
			from the same source and target configuration as $\sigma^{(i)}$.
			In particular, since $\sigma^{(i)}$ is a $\Bmax$-hill
			and $\Bmax-\Upsilon_\C-\Gamma_\C-1>\ell+\Gamma_\C$
			it follows that $\widehat{\sigma^{(i)}}$ is in fact an 
			$\ell$-embedding of $\sigma^{(i)}$ all of whose 
			configurations have a counter value
			whose absolute difference with $\ell$ is strictly
			larger than $\Gamma_\C$ 
			(except for the exotic case when $i=n/2$ and
			$\gamma$ in fact ends with $\sigma^{(i)}$, and hence the
			last configuration of $\sigma^{(i)}$ happens to be the last configuration
			$q''(z)$ of $\gamma$; recalling that $z\in[\Bmin+1,\Bmax-1]$).
		\item For each of the $(\Bmax-N-1)$-valleys $\sigma^{(2i,j)}$
			there exists an $(N-\Gamma_\C)$-semirun
			$\widehat{\sigma^{(2i,j)}}$ that is both a min-rising and max-falling
			$(\Bmax-N-1+\Upsilon_\C+\Gamma_\C+1)$-embedding of $\sigma^{(2i,j)}$
			with the same source and target configuration 
			as $\sigma^{(2i,j)}$. 
			In particular, since $\sigma^{(2i,j)}$ is a $(\Bmax-N-1)$-valley
			and 
			$$
				\Bmax-N+\Upsilon_C+\Gamma_\C\stackrel{(\ref{upper N})}{<}
				\Bmin+ \Upsilon_C+\Gamma_\C
				\leq
				\ell-\Gamma_\C
				$$
			(where the last inequality follows
			from the definition of $\Bmin$ on page~\pageref{Bmin Bmax def}),
			it follows that $\widehat{\sigma^{(2i,j)}}$
			is in fact an $\ell$-embedding of $\sigma^{(2i,j)}$ all
			of whose configurations have a counter value whose
			absolute difference with $\ell$ is is strictly larger than $\Gamma_\C$
			(except, similar as above, for the exotic case when $i=n/2$, 
			$j=k_{i}$ and
			$\gamma$ in fact ends with $\sigma^{(2i,j)}$, and hence
			the last configuration $\sigma^{(2i,j)}$
			happens to be
			the last configuration 
			$q''(z)$ of $\gamma$).
	\end{enumerate}
\end{remark}

It is worth pointing out that applying the remark immediately would only yield the existence
of an $(N-\Gamma_\C)$-semirun $\gamma'$
that is both a min-rising and max-falling $\ell$-embedding of
$\gamma$ with the same source configuration $q'(y)$ and
the same target configuration $q''(z)$ as $\gamma$
such that 
$\values(\gamma')\subseteq[\min(\gamma)-\Gamma_\C,
\max(\gamma)+\Gamma_\C]$.
However we need to show the existence of such an $\ell$-embedding
rather from $q'(y+\Upsilon_C)$ to $q''(z)$.
For this, we make a final case distinction on whether among the $\Bmax$-hills $\sigma^{(i)}$
and the $(\Bmax-N-1)$-valleys $\sigma^{(2i,j)}$ there exists one whose $\phi$-projection
contains strictly more occurrences of 
$\boldsymbol{[}$ as occurrences
of $\boldsymbol{]}$.

\begin{itemize}
	\item {\em Case A.2.i: Among the $\Bmax$-hills $\sigma^{(i)}$
and the $(\Bmax-N-1)$-valleys $\sigma^{(2i,j)}$ there exists one whose $\phi$-projection
		contains strictly more occurrences of $\boldsymbol{[}$ as
		occurrences of $\boldsymbol{]}$.}
		\medskip

		Without loss of generality let us
		assume that 
		there exists some
		$s\in[2,n/2]$ such that
		$\sigma^{(s)}=\beta^{(2s-1)}\alpha^{(2s-1)}\beta^{(2s)}$ is a $\Bmax$-hill 
		for which 
		$\phi(\sigma^{(s)})$ contains strictly more occurrences
		of $\boldsymbol{[}$ as occurrences of $\boldsymbol{]}$
		--- the case when there is a $(\Bmax-N-1)$-valley $\sigma^{(2s,j)}$
		for which $\phi(\sigma^{(2s,j)})$ has the above property can be
		proven analogously.

		Assume $\sigma^{(s)}=\beta^{(2s-1)}\alpha^{(2s-1)}\beta^{(2s)}$ has source configuration 
	 $r_1(x_1)$ and target configuration 
		$r_2(x_2)$, say.
		Since $\beta^{(2s-1)}$ was surely 
		neither the first nor the last 
		crossing transition
		of $\rho$ (recall that $s\geq 2$), 
		it follows that $\beta^{(2s-1)}$ is doubly-crossing by Remark~\ref{remark Type III}, 
		and therefore $x_1 \leq \Bmin$.
		Recalling  the notion of hybrid semirun 
		(Definition~\ref{definition hybrid semirun})
		we now apply Point 2 of
		Remark~\ref{hill bracket} to our 
		$\Bmax$-hill $\sigma^{(s)}$ and obtain a
		hybrid semirun $\eta$ 
		\begin{itemize}
			\item whose source configuration is $r_1(x_1)$ and
				whose target configuration is $r_2(x_2)$,
			\item that is both a min-rising and max-falling 
				$(\Bmax-\Upsilon_\C-\Gamma_\C-1)$-embedding of 
				$\sigma^{(s)}$, and
			\item that has breadth $1$ and contains precisely one 
				unlowered $+p$-transition.
		\end{itemize}
		From the above and the fact that $\sigma^{(s)}$ is a $\Bmax$-hill
		the following remark follows.
		\begin{remark}\label{remark intermediate}
			All configurations of $\eta$ 
			(except possibly the target configuration
			$r_2(x_2)$ in the exotic case when $\gamma$ ends with $\sigma^{(s)}$) 
		have a counter value whose absolute difference
			with $\ell$ is strictly larger than $\Gamma_\C$.
		Moreover on can write $\eta$ as
		$\eta=\alpha\beta\alpha'$, where
		for some intermediate configurations $r_1'(x_1')$ and $r_2'(x_2')$ we have
		that
		\begin{itemize}
			\item $\alpha$ is an $(N-\Gamma_\C)$-semirun from $r_1(x_1)$ to 
				$r_1'(x_1')$, 
			\item $\beta$ is an $N$-semirun 
				 $r_1'(x_1')\semi{+p,N}r_2'(x_2')$ 
				that is a $+p$-transition, i.e. $x_2'=x_1'+N$, and
			\item $\alpha'$ is an $(N-\Gamma_\C)$-semirun
				from $r_2'(x_2')$ to $r_2(x_2)$.
		\end{itemize}
		\end{remark}
		Let $\widehat{\beta}$ denote the lowering of $\beta$, i.e.
		$\widehat{\beta}$ is the $(N-\Gamma_\C)$-semirun
		$r_1'(x_1')\semi{+p,N-\Gamma_\C}r_2'(x_2'-\Gamma_\C)$.
		By Remark~\ref{remark intermediate}
		it follows that the $(N-\Gamma_\C)$-semirun
		$$\theta\quad=\quad\left(\left(\alpha\widehat{\beta}\right)+\Gamma_\C\right) \alpha'$$
		from $r_1'(x_1+\Gamma_\C)$ to $r_2(x_2)$ is an
		$\ell$-embedding of $\eta$.
		Bearing in mind our factorization of $\gamma$ from Remark~\ref{gamma fact}
		and taking into account Remark~\ref{remark widehat} 
		we obtain that
		$$
		\widetilde{\gamma^{(1)}}\quad=\quad
		\left(\left(\widehat{\alpha^{(2)}}
		\left(\prod_{i=2}^{s-1}\widehat{\sigma^{(i)}}\ 
		\alpha^{(2i,0)}\left(\prod_{j=1}^{k_{i}}
		\widehat{\sigma^{(2i,j)}}\
		\alpha^{(2i,j)}\right)\right) \right) +\Gamma_\C \right) \theta
		$$
		is an $(N-\Gamma_\C)$-semirun from
		$q'(y+\Gamma_\C)$ to $r_2'(x_2')$
		that is an $\ell$-embedding of
		$\gamma$'s prefix $N$-semirun
		$$
		\gamma^{(1)}\quad=\quad\alpha^{(2)}\left(\prod_{i=2}^{s-1}\sigma^{(i)}\ 
		\alpha^{(2i,0)}\left(\prod_{j=1}^{k_{i}}
		\sigma^{(2i,j)}\ 
		\alpha^{(2i,j)}\right)
		\right)\sigma^{(s)}
$$
from $q'(y)$ to $r_2'(x_2')$
 satisfying $\values(\widetilde{\gamma^{(1)}})
\subseteq[\min(\gamma^{(1)}),
\max(\gamma^{(1)})+\Gamma_\C]$.
Moreover we have by Remark~\ref{remark widehat}
that 
$$
\widehat{\gamma^{(2)}}\quad=\quad\alpha^{(2s,0)}\ \left(\prod_{j=1}^{k_s}\widehat{\sigma^{(2i,j)}}\alpha^{(2i,j)}\right)
\left(\prod_{i=s+1}^{n/2}
\widehat{\sigma^{(i)}}\
\alpha^{(2i,0)}\ \left(\prod_{j=1}^{k_i}\widehat{\sigma^{(2i,j)}}\ \alpha^{(2i,j)}\right)\right)
$$
is an $(N-\Gamma_\C)$-semirun from $r_2'(x_2')$ to $q''(z)$ that is both a min-rising and max-falling
$\ell$-embedding of $\gamma$'s remaining suffix $N$-semirun
$$
\gamma^{(2)}\quad=\quad\alpha^{(2s,0)}\ \left(\prod_{j=1}^{k_s}\sigma^{(2i,j)}\alpha^{(2i,j)}\right)
\left(\prod_{i=s+1}^{n/2}
\sigma^{(i)}\
\alpha^{(2i,0)}\ \left(\prod_{j=1}^{k_i}\sigma^{(2i,j)}\ \alpha^{(2i,j)}\right)\right)
$$
from $r_2'(x_2')$ to $q''(z)$.
Altogether 
$\gamma'=\widetilde{\gamma^{(1)}}\widehat{\gamma^{(2)}}$
is the desired $(N-\Gamma_\C)$-semirun
from $q'(y+\Gamma_\C)$ to $q''(z)$
that is an $\ell$-embedding of 
$\gamma=\gamma^{(1)}\gamma^{(2)}$
with $\values(\gamma')\subseteq[\min(\gamma)-\Gamma_\C,\max(\gamma)+\Gamma_\C]$.

\bigskip
	\item{\em Case A.2.ii: Among the $\Bmax$-hills $\sigma^{(i)}$
and the $(\Bmax-N-1)$-valleys $\sigma^{(2i,j)}$ all have a $\phi$-projection
		that contains at least as many occurrences of $\boldsymbol{]}$ as
		occurrences of $\boldsymbol{[}$.}

	Observe that by Remark~\ref{remark widehat} the $(N-\Gamma_\C)$-semirun
	 
$$\widehat{\gamma}\quad=\quad
\alpha^{(2)}\ \left(\prod_{i=2}^{n/2}\widehat{\sigma^{(i)}}\ \alpha^{(2i,0)}\ 
\left(\prod_{j=1}^{k_i}\widehat{\sigma^{(2i,j)}}\ \alpha^{(2i,j)}\right)\right)
$$
from $q'(y)$ to $q''(z)$
is both a min-rising and max-falling $\ell$-embedding of $\gamma$
all of whose configurations (except for possibly the target configuration $q''(z)$)
have a counter value whose absolute difference with $\ell$ is strictly
larger than $\Gamma_\C$.
Yet we need to show the existence of some $(N-\Gamma_\C)$-semirun
$\gamma'$ that is an $\ell$-embedding of $\gamma$
from $q'(y+\Gamma_\C)$ to $q''(z)$
that satisfies $\values(\gamma')
\subseteq[\min(\gamma)-\Gamma_\C,\max(\gamma)+\Gamma_\C]$.

	By Remark~\ref{remark hill/valley} all of the
	lowered $(N-\Gamma_\C)$-semiruns
	$\widehat{\sigma^{(i)}}$ and $\widehat{\sigma^{(2i,j)}}$
	mentioned in Remark~\ref{remark widehat}
	contain at least as many occurrences of $\boldsymbol{[}$ as of
	$\boldsymbol{]}$ or, vice versa, at least as many occurrences
	of $\boldsymbol{]}$ as of $\boldsymbol{[}$, if 
	$\sigma^{(i)}$ does, respectively if $\sigma^{(2i,j)}$ does.

	Thus, by our case we obtain that every 
	$\phi(\widehat{\sigma^{(i)}})$ and $\phi(\widehat{\sigma^{(2i,j)}})$ contains
at least as many occurrences of $\boldsymbol{]}$
as occurrences of $\boldsymbol{[}$.
	
	Recalling that neither $\alpha^{(2)}$ nor any of the $\alpha^{(2i,j)}$ contain
	any $+p$-transitions nor $-p$-transitions (and thus have all a $\phi$-projection
	$\varepsilon$) it follows that
	$\phi(\widehat{\gamma})$ contains at least as many 
	occurrences of $\boldsymbol{]}$ as occurrences
	of $\boldsymbol{[}$.

	Since $\Delta(\gamma)>\Upsilon_\C$ by Point 4 of
Remark~\ref{gamma fact}, and thus
	 $\Delta(\widehat{\gamma})>\Upsilon_\C$,
	there exists a subsemirun $\widehat{\gamma}[c,d]$
	satisfying $\Delta(\widehat{\gamma}[c,d])>\Upsilon_\C$
	and $\phi(\widehat{\gamma}[c,d])\in\Lambda_8$
	by Lemma~\ref{bracket lemma}.
	By now applying Lemma~\ref{lemma zero} there exists an $(N-\Gamma_\C)$-semirun 
	$\chi$ satisfying
	\begin{itemize}
		\item $\Delta(\chi)=\Delta(\widehat{\gamma}[c,d])-\Gamma_\C$ and
		\item $\chi=\widehat{\gamma}[c,d]-I_1-I_2\cdots-I_h$
	for pairwise disjoint intervals $I_1,\ldots,I_h\subseteq[c,d]$
			such that $\phi(\widehat{\gamma}[I_i])\in\Lambda_{16}$
			and $\Delta(\widehat{\gamma}[I_i])>0$ for all $i\in[1,h]$.
	\end{itemize}
	Note that from the definition of $\chi$ and the fact that all intermediate
	configurations of $\widehat{\gamma}$ have counter values
	whose absolute difference with $\ell$ is strictly larger than $\Gamma_\C$
	it follows that 
	$\chi+\Gamma_\C$ 
	is an $\ell$-embedding of $\widehat{\gamma}[c,d]$
	that has the same target configuration 
	as $\widehat{\gamma}[c,d]$ 
	and that satisfies 
	$\values(\chi+\Gamma_\C)\subseteq
	[\min(\widehat{\gamma}[c,d])-\Gamma_\C,\max(\widehat{\gamma}[c,d])+\Gamma_\C]$.
	Analogously, it follows that 
	$\delta=(\widehat{\gamma}[0,c]+\Gamma_\C)\ (\chi+\Gamma_\C)$
	is an $(N-\Gamma_\C)$-semirun from $q'(y+\Gamma_\C)$ to the
	same target configuration as $\widehat{\gamma}[0,d]$
	that is an $\ell$-embedding of $\widehat{\gamma}[0,d]$
	and that satisfies 
	$\values(\delta)\subseteq[\min(\gamma[0,d])-\Gamma_\C,\max(\gamma[0,d])+\Gamma_\C]$.
	
	Finally it follows that
	$$\gamma'\quad=\quad(\widehat{\gamma}[0,c]+\Gamma_\C)\ (\chi+\Gamma_\C)\ 
	\widehat{\gamma}[d,|\widehat{\gamma}|]
	$$
	is the desired 
	$(N-\Gamma_\C)$-semirun from $q'(y+\Gamma_\C)$ to $q''(z)$ that is an
	$\ell$-embedding of $\widehat{\gamma}$ 
	and hence of $\gamma$ 
	that satisfies $\values(\gamma')
	\subseteq[\min(\gamma)-\Gamma_\C,\max(\gamma)+\Gamma_\C]$.
	
\end{itemize}
\medskip

\noindent
{\em Case A.3: All but at least two of the $\Bmax$-hill candidates 
$\sigma^{(i)}$ 
in (\ref{hill factorization}) are in fact $\Bmax$-hills.}

Recall the factorization
$$\rho\quad=\quad \alpha^{(0)}\sigma^{(1)}\alpha^{(2)}\sigma^{(2)}\alpha^{(4)}\sigma^{(3)}\alpha^{(6)}\quad\cdots\quad\sigma^{(n/2)}\alpha^{(n)}$$
from (\ref{hill factorization}) 
where each
$\sigma^{(i)}=\beta^{(2i-1)}\alpha^{(2i-1)}\beta^{(2i)}$ is a $\Bmax$-hill candidate
for all $i\in[1,n/2]$.
By our case we must have $n\geq 4$.

By a similar reasoning as in Case A.2 one can show that the $\Bmax$-hill candidate that 
are not $\Bmax$-hills must be precisely the two subsemiruns $\sigma^{(1)}$ and $\sigma^{(n/2)}$.
In particular there cannot be strictly more than two
$\Bmax$-hill candidates in (\ref{hill factorization}) that are not in fact $\Bmax$-hills.
Moreover, as already reasoned in Case A.2, 
neither $\beta^{(1)}$ nor $\beta^{(n)}$ is doubly-crossing.
Thus $\alpha^{(0)}$ and $\alpha^{(n)}$ are empty.
Hence, one can now factorize our $N$-semirun $\rho$ as
$$
\rho\quad=\quad \left( \beta^{(1)}\alpha^{(1)}\right) \beta^{(2)} \alpha^{(2)}\beta^{(3)}\alpha^{(3)}\quad\cdots\quad
\beta^{(n-1)} \left( \alpha^{(n-1)} \beta^{(n)}\right)\quad.
$$

For finishing this case we will show the existence
\begin{enumerate}
	\item of an $(N-\Gamma_\C)$-semirun that is 
both a min-rising and max-falling $\ell$-embedding of the semirun $\beta^{(1)}\alpha^{(1)}$ with the
same source and target configuration as $\beta^{(1)}\alpha^{(1)}$,
	\item of an $(N-\Gamma_\C)$-semirun that is 
both a min-rising and max-falling $\ell$-embedding of the semirun $\alpha^{(n-1)}\beta^{(n)}$ with the
same source and target configuration as $\alpha^{(n-1)}\beta^{(n)}$, and
\item of an $(N-\Gamma_\C)$-semirun that is both a min-rising and max-falling
	$\ell$-embedding of the semirun $\beta^{(2)} \alpha^{(2)}\beta^{(3)}\alpha^{(3)}\cdots\beta^{(n-1)}$
	with same source and target configuration.
\end{enumerate}

Points 1 and 2 are proven analogously as Point 1 from Case A.2. For proving Point 3, we consider a different factorization
$$
\beta^{(2)} \alpha^{(2)}\beta^{(3)}\alpha^{(3)}\quad\cdots\quad \beta^{(n-1)}
\quad =\quad \tau^{(1)} \alpha^{(3)} \tau^{(2)} \alpha^{(5)} \quad \cdots \quad \tau^{((n-2)/2)},
$$
where $\tau^{(i)}=\beta^{(2i)}\alpha^{(2i)}\beta^{(2i+1)}$ is a $\Bmin$-valley candidate for all $i \in [1,(n-2)/2]$. Since $\beta^{(2i)}$ and $\beta^{(2i+1)}$ have to be doubly crossing for all $i \in [1,(n-2)/2]$, $\tau^{(i)}$ is in fact a $\Bmin$-valley for all $i \in [1,(n-2)/2]$
by a similar reasoning as used in Case A.2 to show that the $\Bmax$-hill candidate that 
is not a $\Bmax$-hill must be $\sigma^{(1)}$ or $\sigma^{(n/2)}$.

We can apply Claim~\ref{claim embedding}
to each $\tau^{(i)}$ and obtain 
an $(N-\Gamma_\C)$-semirun $\widehat{\tau^{(i)}}$ that is both a min-rising and max-falling $\ell$-embedding
of $\tau^{(i)}$ with the same source and target configuration.
Thus, it only remains to show the same for $\alpha^{(2i+1)}$ for each $i\in[1,(n-2)/2]$. 
This is done analogously as in Case A.1 when proving the same for each $\alpha^{(2i)}$ 
for each $i \in[0,n/2]$.

\medskip

\noindent
{\em Case B: The number of crossing transitions $n$ is odd.}

Recall that we had assumed without loss of generality that $\beta^{(1)}$ is a $+p$-transition. 
Since $n$ is odd one can consider the following first factorization
\begin{eqnarray}
\rho\quad=\quad \alpha^{(0)}\sigma^{(1)}\alpha^{(2)}\sigma^{(2)}\alpha^{(4)}\sigma^{(3)}\alpha^{(6)}\quad\cdots\quad\sigma^{(\lfloor n/2 \rfloor)}\alpha^{(n-1)} \beta^{(n)} \alpha^{(n)}, 
	\label{hill factorization 1}
\end{eqnarray}
where $\sigma^{(i)}=\beta^{(2i-1)}\alpha^{(2i-1)}\beta^{(2i)}$ is a $\Bmax$-hill candidate,
$\beta^{(2i-1)}$ is a $+p$-transition, 
$\beta^{(2i)}$ is a
$-p$-transition for all 
$i\in[1,\lfloor n/2 \rfloor]$
, and $\beta^{(n)}$ is a $+p$-transition;
as well as the following second factorization 
\begin{eqnarray}
\rho\quad=\quad \alpha^{(0)}\beta^{(1)} \alpha^{(1)} \tau^{(1)}\alpha^{(3)} \tau^{(2)} \alpha^{(5)} \tau^{(3)} 
	\quad\cdots\quad\tau^{(\lfloor n/2 \rfloor)}\alpha^{(n)}, 
	\label{hill factorization 2}
\end{eqnarray}
where $\beta_1$ is a $+p$-transition,
$\tau^{(i)}=\beta^{(2i)}\alpha^{(2i)}\beta^{(2i+1)}$ is a 
$\Bmin$-valley candidate,
 $\beta^{(2i)}$ is a $-p$-transition and $\beta^{(2i+1)}$ is a $+p$-transition for all 
 $i\in[1,\lfloor n/2 \rfloor]$.
Indeed, this--- as for the Case A factorization (\ref{hill factorization})--- 
immediately follows from the definition of crossing transitions
and the fact that neither $\alpha^{(2i-1)}$ (resp. $\alpha^{(2i)}$) contains
any configuration with counter value strictly less
than $\Bmax$ (resp. strictly larger than $\Bmin$)
for all $i\in[1,\lfloor n/2\rfloor]$.

Our proof next will make a case distinction on the number of 
$\Bmax$-hill candidates $\sigma^{(i)}$ in the factorization~(\ref{hill factorization 1})
that are in fact $\Bmax$-hills and on the number of 
$\Bmin$-valley candidates $\tau^{(i)}$ in the factorization~(\ref{hill factorization 2})
that are in fact $\Bmin$-valleys.

\medskip

\noindent
{\em Case B.1: All of the $\Bmax$-hill candidates $\sigma^{(i)}$ in
(\ref{hill factorization 1}) are in fact $\Bmax$-hills or all of the $\Bmin$-valley candidates 
$\tau^{(i)}$ in
(\ref{hill factorization 2}) are in fact $\Bmin$-valleys.}

Let us assume without loss of generality that all of the $\Bmax$-hill candidates in
(\ref{hill factorization 1}) are in fact $\Bmax$-hills. The case when all
$\Bmin$-valley candidates in
(\ref{hill factorization 2}) are in fact $\Bmin$-valleys can be proven analogously. 
Each $N$-semirun $\sigma^{(i)}$ 
can hence be turned into an $(N-\Gamma_\C)$-semirun $\widehat{\sigma^{(i)}}$ 
that is both a min-rising and max-falling 
$\ell$-embedding of $\sigma^{(i)}$ with same source and target configuration
as $\sigma^{(i)}$ 
according to Claim~\ref{claim embedding}. 
Moreover, the same holds for $\alpha^{(2i)}$ for all
$i \in [0,\lfloor n/2 \rfloor-1]$, as seen in Case A.1.
Thus it remains to deal with the subsemiruns
$\alpha^{(n-1)}$, $ \beta^{(n)}$ and $ \alpha^{(n)}$.

We make a final case distinction on the target configuration of the dangling $+p$-transition $\beta^{(n)}$.
\begin{itemize}
\item {\em Case B.1.i: $\beta^{(n)}$ has a target configuration with a counter
	value strictly larger than $\Bmax+\Upsilon_\C$.}

Then clearly $\beta^{(n)} \alpha^{(n)}$ is a $\Bmax$-hill 
		as well since
		$\alpha^{(n)}$ contains
		no configurations with counter value 
		strictly less than $\Bmax$ besides 
		its last one. 
The $N$-semirun $\beta^{(n)} \alpha^{(n)}$ can hence be turned into a $(N-\Gamma_\C)$-semirun that is both a min-rising and max-falling $\ell$-embedding with the
		same source and target configuration according to Claim~\ref{claim embedding}.
		Moreover, the same holds for $\alpha^{(n-1)}$,
		as analogously proven for 
		$\alpha^{(2i)}$ for all
$i \in [0,\lfloor n/2 \rfloor]$ in Case A.1.
The concatenation of these two $\ell$-embeddings yields 
		an $(N-\Gamma_\C)$-semirun that is a
min-rising and max-falling $\ell$-embedding of $\alpha^{(n-1)} \beta^{(n)} \alpha^{(n)}$ with 
the same source and target configuration.

\item {\em Case B.1.ii: $\beta^{(n)}$ has a target configuration with a counter value
	strictly less than $\Bmax$.}

		It immediately follows that $\beta^{(n)}$ is crossing but not doubly-crossing,
		thus $\alpha^{(n)}$ is empty. 
		The remaining $\alpha^{(n-1)} \beta^{(n)}$ can thus be factorized as
$$
 \alpha^{(n-1)} \beta^{(n)}\quad=\quad \xi^{(1)} \chi^{(1)}\cdots \xi^{(k)} \chi^{(k)},
$$
where each $\chi^{(i)}$ is a $(\Bmax-N-1)$-valley 
and each $\xi^{(i)}$ satisfies 
		$\values(\xi^{(i)})\subseteq[\Bmax-N,\Bmin+1]$, using a similar factorization as for proving Point 1 in Case A.2.
The $N$-semirun $\alpha^{(n-1)} \beta^{(n)}$ can hence be turned into an
		$(N-\Gamma_\C)$-semirun that is both a min-rising and max-falling $\ell$-embedding with same source and target
configuration.
Recalling that $\alpha^{(n)}$ is empty, the above embedding is 
		an $(N-\Gamma_\C)$-semirun that is both a
min-rising and max-falling $\ell$-embedding of $\alpha^{(n-1)} \beta^{(n)} \alpha^{(n)}$ 
		with same source and target configuration.
\item {\em Case B.1.iii: $\beta^{(n)}$ has a target configuration with counter value 
	in $[\Bmax, \Bmax + \Upsilon_\C]$.}

		Thus, the source configuration of
		$\beta^{(n)}$ has a counter value in 
		$[\Bmax-N, \Bmax + \Upsilon_\C-N]$. 
		One finishes this case 
		analogously as 
		Points 1 and 2 in Case A.2:
\begin{enumerate}
	\item Firstly, one shows the existence of an $(N-\Gamma_\C)$-semirun that is
both a min-rising and max-falling $\ell$-embedding of $\alpha^{(n)}$ with the
same source and target configuration as $\alpha^{(n)}$ 
		as follows:
		one factorizes $\alpha^{(n)}$ into $(N-\Gamma_\C)$-semiruns
		that have all counter values in 
		$[\Bmax-1,\Bmin+N]$
		and into $(\Bmin+N+1)$-hills.
\item Secondly, let us assume that $\beta^{(n)}$ is an $N$-semirun from $q'(y)$ to $q''(z)$ and that
moreover
	$\alpha^{(0)}\sigma^{(1)}\alpha^{(2)}\cdots\sigma^{(\lfloor n/2 \rfloor)}\alpha^{(n-1)}$ is an $N$-semirun
		from $q(x)$ to $q'(y)$.
Stipulating that $\beta^{(n)}=q'(y)\semi{+p,N}q''(z)$, we explicitly lower $\beta^{(n)}$
	into the $(N-\Gamma_\C)$-semirun $q'(y+\Gamma_\C)\semi{+p,N-\Gamma_\C}q''(z)$,
		which is --- since $\beta^{(n)}$ is doubly-crossing --- 
		obviously both a min-rising and max-falling
		$\ell$-embedding of $\beta^{(n)}$ from $q'(y+\Gamma_\C)$ to $q''(z)$.
	Then one shows the existence of an $(N-\Gamma_\C)$-semirun 
		that is an $\ell$-embedding of $\alpha^{(0)}\sigma^{(1)}\alpha^{(2)}\cdots\sigma^{(\lfloor n/2 \rfloor)}\alpha^{(n)}$
		from $q(x)$ to $q'(y+\Gamma_\C)$
		with configurations all of whose counter values lie in 
		the interval $[\min(\rho)-\Gamma_\C,\max(\rho)+\Gamma_\C]$ as follows:
		one subfactorizes each of the $\alpha^{(2i)}$
		into $(N-\Gamma_\C)$-semiruns that have a counter values
		in $[\Bmax-N,\Bmin]$ and into $(\Bmax-N-1)$-valleys
		and by recalling that $\sigma^{(i)}$ is a $\Bmax$-hill
		for all $i \in [1,\lfloor n/2 \rfloor]$.
\end{enumerate}
\end{itemize}

\medskip

\noindent
{\em Case B.2: Not all of the $\Bmax$-hill candidates $\sigma^{(i)}$ in
(\ref{hill factorization 1}) are in fact $\Bmax$-hills and not all 
of the $\Bmin$-valley candidates $\tau^{(i)}$ in (\ref{hill factorization 2}) are in fact $\Bmin$-valleys.}

Since they all start and end with a doubly-crossing transition
we remark that $\sigma^{(i)}$ is in fact a $\Bmax$-hill for all $i\in[2,\lfloor n/2 \rfloor]$
and $\tau^{(i)}$ is in fact a $\Bmin$-valley for all $i\in[1,\lfloor n/2 \rfloor-1]$.
Hence our case implies that $\sigma^{(1)}$ is in fact not a $\Bmax$-hill and
that $\tau^{(\lfloor n/2 \rfloor)}$ is in fact not a $\Bmin$-valley.
As in Case A.2, $\beta^{(1)}$ and $\beta^{(n)}$ are hence not doubly-crossing,
and hence $\alpha^{(0)}$ and $\alpha^{(n)}$ are empty.

By definition of a Type III semirun, $\rho$ contains at least one doubly-crossing
transition and thus $n \geq 3$.
Since the $+p$-transition $\beta^{(1)}$ is not doubly-crossing (and therefore ends
at a counter value strictly larger than $\Bmin+N$) but $\beta^{(2)}$ is,
it follows that
the only reason for $\sigma^{(1)}=\beta^{(1)}\alpha^{(1)}\beta^{(2)}$ not
to be a $\Bmax$-hill is that the $-p$-transition $\beta^{(2)}$ has a source
configuration with a counter value in $[\Bmax, \Bmax+\Upsilon_\C]$ 
and hence a target configuration with counter value in $[\Bmax-N, \Bmax+\Upsilon_\C-N]$,
similarly as seen in Case A.2.
Analogously, 
the only reason for $\tau^{(\lfloor n/2\rfloor)}=\beta^{(n-1)}\alpha^{(n-1)}\beta^{(n)}$ not to be a $\Bmin$-valley
is that the doubly-crossing $-p$-transition $\beta^{(n-1)}$ 
has a target configuration with counter value in $[\Bmin-\Upsilon_\C,\Bmin]$.

Recalling that $\sigma^{(\lfloor n/2\rfloor)}=\beta^{(n-2)}\alpha^{(n-2)}\beta^{(n-1)}$,
for finishing this case we will apply an analogous reasoning as in Case A.2:
\begin{enumerate}
	\item Firstly, one shows the existence of an $(N-\Gamma_\C)$-semirun that is
both a min-rising and max-falling $\ell$-embedding of $\beta^{(1)}\alpha^{(1)}$ with the
same source and target configuration as $\beta^{(1)}\alpha^{(1)}$.
\item Secondly, let us assume that $\beta^{(2)}$ is an $N$-semirun from $q(x)$ to $q'(y)$ and that moreover
	$\alpha^{(2)}\sigma^{(2)}\alpha^{(4)}\cdots\sigma^{(\lfloor n/2 \rfloor)}$ is an $N$-semirun
		from $q'(y)$ to $q''(z)$, with $y \in [\Bmax-N, \Bmax+\Upsilon_\C- N]$
		and $z \in [\Bmin-\Upsilon_\C,\Bmin]$.
		Noting that $\beta^{(2)}=q(x)\semi{-p,N}q'(y)$, we explicitly lower $\beta^{(2)}$
	into the $(N-\Gamma_\C)$-semirun $q(x)\semi{-p,N-\Gamma_\C}q'(y+\Gamma_\C)$,
		which is --- since $\beta^{(2)}$ is doubly-crossing --- obviously 
		both a min-rising and max-falling
		$\ell$-embedding of $\beta^{(2)}$ from $q(x)$ to $q'(y+\Gamma_\C)$.
		Then one shows, as done in Case 2.A, 
		the existence of an $(N-\Gamma_\C)$-semirun that is an
	$\ell$-embedding of $\alpha^{(2)}\sigma^{(2)}\alpha^{(4)}\cdots\sigma^{(\lfloor n/2 \rfloor)}$
		all of whose counter values lie in $[\min(\rho)-\Gamma_\C,\max(\rho)+\Gamma_\C]$
		from $q'(y+\Gamma_\C)$ to $q''(z)$ 		
		by subfactorizing each of the $\alpha^{(2i)}$
		into $(N-\Gamma_\C)$-semiruns that have counter values
		in $[\Bmax-N,\Bmin]$ and into $(\Bmax-N-1)$-valleys,
		by recalling that $\sigma^{(i)}$ is a $\Bmax$-hill 
		for all $i \in [2,\lfloor n/2 \rfloor]$,
		and that 
		
			\begin{eqnarray*}
				z-y &\geq&
		 N-(\Bmax-\Bmin)-2 \Upsilon_\C\\
				&\stackrel{(\ref{5/6 upper})}{>}&
				N-(5/6\cdot N+2\Upsilon_\C+4\Gamma_\C+2)-2\Upsilon_\C\\
				&=& N/6-(5\Upsilon_\C-4\Gamma_\C-2) + \Upsilon_\C\\
				&>&M_\C/6-(5\Upsilon_\C-4\Gamma_\C-2) + \Upsilon_\C\\
				&>& \Upsilon_\C. 
		\end{eqnarray*}

	\item Finally (analogously as Point 1) one shows the existence of an $(N-\Gamma_\C)$-semirun that is
both a min-rising and max-falling $\ell$-embedding of $\alpha^{(n-1)} \beta^{(n)}$ with the
same source and target configuration as $\alpha^{(n-1)} \beta^{(n)}$.
\end{enumerate}

\section{Proof of the Small Parameter Theorem}\label{application section}

This section is devoted to proving the Small Parameter Theorem (Theorem~\ref{theorem upper}).

\medskip

\noindent

\noindent
For proving this let us fix some $N> M_\C$ and 
some  accepting $N$-run $\pi$ in $\C$ with $\values(\pi)\subseteq[0,4N]$ of the form
$$
\pi \quad= \quad r_0(x_0) \quad\xrightarrow{\pi_0,N} \quad r_1(x_1) \quad\cdots \quad\xrightarrow{\pi_{n-1},N} \quad r_n(x_n)
$$
with $r_n \in F$.
We will assume that accepting runs in $\C$ end with counter value $0$ and hence, that $x_n=x_0=0$.
We do not lose generality by making this assumption. Indeed, from every POCA $\C$, one can build a POCA $\C'$  with all its accepting runs ending in configuration with counter value $0$ 
such that, for all $N \in \N$,
there exists an accepting
 $N$-run  in $\C$ with values in $[0, 4N ]$
if, and only if, there exists an accepting
$N$-run in $\C'$ with values in $[0, 4N ]$.
This is clear when one considers the construction $\C'$ obtained from $\C$ by adding two 
control states $r_-$ and $r_f$ such that every final control state of $\C$ has a $\geq 0$ rule leading to $r_-$, $r_-$ 
has a $-1$ rule that is a loop, and finally a $=0$ rule to $r_f$, the only final state of $\C'$.

\medskip

\noindent
Starting from the accepting $N$-run $\pi$, we need to prove the existence of an accepting $(N-\Gamma_\C)$-run in $\C$. For every $a,b\in\mathbb{Q}$ with $a < b$ we define 
$[a,b[=\{c\in\mathbb{Q}\mid a\leq c<b\}$ and
$]a,b]=\{c\in\mathbb{Q}\mid a<c\leq b\}$.

\medskip

\noindent
Since {$\frac{N}{3}<N-\Gamma_\C$}, as $\Gamma_\C < \frac{2 M_\C}{3} < \frac{2 N}{3}$ by definition of the constants on page~\pageref{constant definitions}, the following claim is clear.
\begin{claim}\label{claim N/3}
	Every subrun $\rho$ of $\pi$ with $\values(\rho)\subseteq[0, \frac{N}{3}[$ is already
	an $(N-\Gamma_\C)$-run.
\end{claim}

\noindent
We can therefore uniquely factorize $\pi$ as 
\begin{eqnarray}
	\pi \quad= \quad\rho^{(0)}\sigma^{(1)}\rho^{(1)} \quad\cdots \quad \sigma^{(m)}\rho^{(m)},\label{fact N/3}
\end{eqnarray}
where each $\rho^{(j)}$ satisfies $\values(\rho^{(j)})\subseteq[0, \frac{N}{3}[$
and each $\sigma^{(j)}$ is some subrun $\pi[c,d]$
with $x_c<\frac{N}{3}$, $x_d<\frac{N}{3}$ and $x_k \geq \frac{N}{3}$ for all $k \in[c+1,d-1]$,
where $[c+1,d-1]\not=\emptyset$.

\medskip

\noindent
To finish the proof of the Small Parameter Theorem (Theorem~\ref{theorem upper}), 
by Claim~\ref{claim N/3} it thus suffices to prove the following statement for the rest of this section.

\fbox{\parbox[t][2.6cm][c]{13cm}{
	For every $N>M_\C$ and every $N$-run 
	$$\sigma \quad= \quad q_0(z_0) \quad\xrightarrow{\sigma_0,N} \quad q_1(z_1)  \quad \cdots  \quad
	\xrightarrow{\sigma_{m-1},N}  \quad q_m(z_m)$$
	satisfying $\values(\sigma)\subseteq[0,4 N]$, 
	$z_0,z_m< \frac{N}{3}$ and $z_i\geq \frac{N}{3}$ for all $i\in[1,m-1]$,
	there exists an $(N-\Gamma_\C)$-run from $q_0(z_0)$ to $q_m(z_m)$.
}}

\begin{center}
	\begin{figure}
\includegraphics[width=0.6\textwidth]{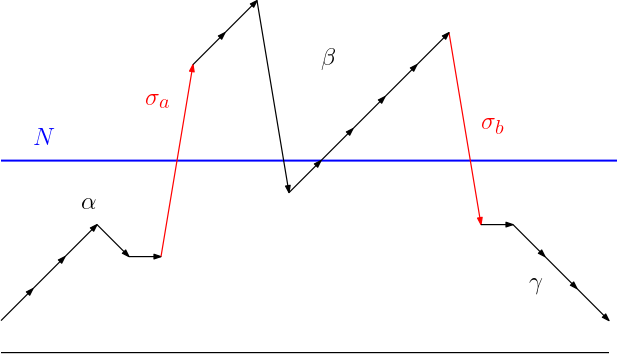}
	\caption{Illustration of the factorization (\ref{sigma fact}).}
\label{figure of the factorization}
	\end{figure}
\end{center}

\noindent
Let $\sigma$ be such an $N$-run. Let us first assume that $z_i\geq N$ for some $i\in[1,m-1]$ ; the case
$z_i<N$ for all $i\in[1,m-1]$ will be treated later.
By this asumption, one can uniquely factorize $\sigma$ --- as seen in Figure~\ref{figure of the factorization} --- as
\begin{eqnarray}
	\sigma \quad = \quad \alpha \ \sigma[a,a+1] \ \beta \ \sigma[b,b+1] \ \gamma,\label{sigma fact}
\end{eqnarray}
where, for some $a,b\in[0,m-1]$,

\begin{samepage}
\begin{itemize}
	\item $\alpha=\sigma[0,a]$ is the maximal 
		prefix of $\sigma$ satisfying $\values(\alpha)\subseteq[0,N[$, in particular the transition
$q_a(z_a)\xrightarrow{ \sigma_a ,N}q_{a+1}(z_{a+1})$ satisfies
		$z_a\in[0,N[$ and $z_{a+1}\in[N,4 N]$,
	\item $\gamma=\sigma[b+1,m]$ is the maximal 
		suffix of $\sigma$ satisfying $\values(\gamma)\subseteq[0,N[$, i.e.
		the transition
$q_b(z_b)\xrightarrow{ \sigma_b ,N}q_{b+1}(z_{b+1})$ satisfies
	$z_b\in[N,4 N]$ and $z_{b+1}\in[0,N[$,
		and
	\item $\beta=\sigma[a+1,b]$ is the remaining infix of $\sigma$ (note that $a+1=b$ is possible).
\end{itemize}
\end{samepage}

\noindent
We will apply the 5/6-Lemma (Lemma~\ref{lemma 5/6}) 
to one of the subruns 
\begin{eqnarray*}
\beta, \quad  \sigma[a,a+1] \ \beta, \quad
\beta \ \sigma[b,b+1], \text{ or } \quad \sigma[a,a+1] \ \beta \ \sigma[b,b+1],
\end{eqnarray*}
hereby showing the existence
of a suitable $(N-\Gamma_\C)$-semirun with same source and target configuration,
respectively. 
We then shift this $(N-\Gamma_\C)$-semirun by $-\Gamma_\C$ to obtain a suitable
$(N-\Gamma_\C)$-run.
To which of the subruns we will choose to apply
the 5/6-Lemma   will depend on the counter values 
$z_a,z_{a+1},z_b$ and $z_{b+1}$.
For deciding this, we make a case distinction
on which of the five intervals 
$\left\{\left[\frac{iN}{3},\frac{(i+1)N}{3}\right[\ : i\in[1,5]\right\}$
they lie in, respectively.
\medskip

Before the above-mentioned distinction 
on $z_a$, $z_{a+1}$, $z_b$ and $z_{b+1}$ we first 
claim that one can
turn the possible resulting prefixes
$\alpha$ and $\alpha \sigma[a,a+1] $ and
possible suffixes $ \sigma[b,b+1] \gamma$ and 
$\gamma$ into $(N-\Gamma_\C)$-runs separately. 
The following claim tells us 
when these latter prefixes (resp. suffixes) 
can be turned into $(N-\Gamma_\C)$-runs
whose target (resp. source) configuration has been shifted down by $\Gamma_\C$.

\begin{claim}[Possible lowering of the prefixes
	and suffixes]\label{claim fix}
	\ 
\begin{enumerate}
	\item If $z_{a+1}\in[N,\frac{5N}{3}[$, then 
	there exists an $(N-\Gamma_\C)$-run from
	$q_0(z_0)$ to $q_{a+1}(z_{a+1}-\Gamma_\C)$.
	\item If $z_a\in [\frac{N}{3}+\Upsilon_\C,N[$, then 
	there exists an $(N-\Gamma_\C)$-run from
	$q_0(z_0)$ to $q_{a}(z_{a}-\Gamma_\C)$.
	\item If $z_b\in[N,\frac{5N}{3}[$, then there exists an
	$(N-\Gamma_\C)$-run from 
		$q_b(z_b-\Gamma_\C)$ to $q_m(z_m)$.
	\item If $z_{b+1}\in[\frac{N}{3}+\Upsilon_\C,N[$,
		then there exists an $(N-\Gamma_\C)$-run
		from $q_{b+1}(z_{b+1}-\Gamma_\C)$
		to $q_m(z_m)$.
\end{enumerate}
\end{claim}
We postpone
the proof of Claim~\ref{claim fix} to the end of this section but refer to
Figure~\ref{prefix examples} for an illustration of Points 1 and 2.

We can use Point (1) or Point (2) of the claim to turn the possible resulting prefixes
$\alpha \sigma[a,a+1] $ or $\alpha$ respectively into $(N-\Gamma_\C)$-runs
with target configuration
 shifted down by $\Gamma_\C$.
Symmetrically, we can use Point (3) or Point (4) of the claim to turn the possible resulting suffixes
$ \sigma[b,b+1] \gamma$  or
 $\gamma$ respectively  into $(N-\Gamma_\C)$-runs
with source configuration
 shifted down by $\Gamma_\C$.
The claim will rely on the Depumping Lemma (Lemma~\ref{lemma zero})
and on the fact that a transition with operation $+p$ or $-p$
has an absolute counter effect 
of $N$ in an $N$-run but $N-\Gamma_\C$ in an
$(N-\Gamma_\C)$-run.

Let us  for the moment assume $z_i\geq N$ for some $i\in[1,m-1]$
along with the factorization~(\ref{sigma fact}) of $\sigma$ and Claim~\ref{claim fix}.

Assuming Claim~\ref{claim fix} we conclude the proof  
by treating the following exhaustive cases
on the positions
of $z_{a+1}$ and $z_b$ separately.

\begin{center}
	\begin{figure}
\includegraphics[width=0.5\textwidth]{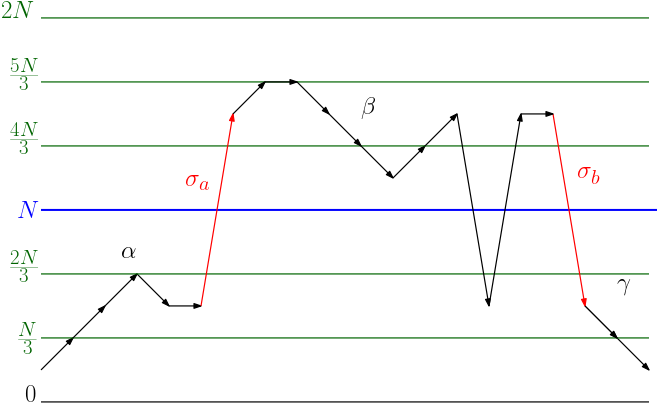}
\includegraphics[width=0.5\textwidth]{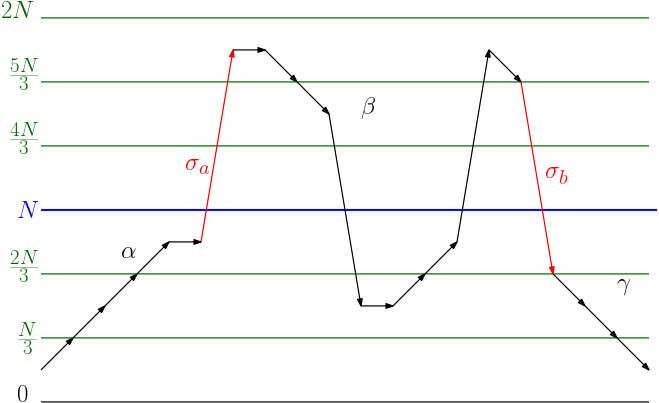}
	\caption{Illustration of Case 1, i.e. $z_{a+1},z_{b}\in[N,\frac{5N}{3} [$, on the left,
				and Case 2, i.e. $z_{a+1},z_{b}\in[\frac{5N}{3},2N[$, on the right.}
\label{case 1 and 2}
	\end{figure}
\end{center}

\noindent
{\em Case 1. $z_{a+1},z_{b}\in[N,\frac{5N}{3} [$, 
cf. Figure~\ref{case 1 and 2}.}\\ 

Recall that $\beta=\sigma[a+1,b]$ as defined in (\ref{sigma fact}) 
is an $N$-run from $q_{a+1}(z_{a+1})$ to $q_{b}(z_b)$
satisfying $\values(\beta)\subseteq[\frac{N}{3},4N]$.
We view $\beta$ as an $N$-semirun. We consider $\ell=N$ and 
observe that 
$$\max(z_{a+1},z_b,\ell)-\min(z_{a+1},z_b,\ell)<
\frac{5N}{3}-N=\frac{2N}{3}\leq\frac{5 N}{6}.
$$
Hence we can apply the 5/6-Lemma (Lemma~\ref{lemma 5/6}) to $\beta$:
there exists an $(N-\Gamma_\C)$-semirun $\widehat{\beta}$ from $q_{a+1}(z_{a+1})$
to $q_b(z_b)$ that is an $N$-embedding of $\beta$
with $\values(\widehat{\beta}) \subseteq [\min(\beta)-\Gamma_\C,\max(\beta)+\Gamma_\C]$.
Since moreover $\frac{N}{3} -  2 \Gamma_\C >  \max(\Const(\C))$,
from $M_\C$'s definition on page~\pageref{constant definitions},
and because $\min(\beta) \geq N/3 $,
it follows that $\widehat{\beta}-\Gamma_\C$, the shifting of 
$\widehat{\beta}$ by $-\Gamma_\C$, is
in fact an $(N-\Gamma_\C)$-run from $p(z_{a+1}-\Gamma_\C)$
to $q_{b}(z_b-\Gamma_\C)$. It thus remains to show the 
existence
of an $(N-\Gamma_\C)$-run from $q_0(z_0)$ to $q_{a+1}(z_{a+1}-\Gamma_\C)$
and one from $q_b(z_b-\Gamma_\C)$ to $q_m(z_m)$:
the former follows from Point (1) of Claim~\ref{claim fix},
and the latter follows from Point (3) of 
Claim~\ref{claim fix}.\\

\medskip 

\noindent
{\em Case 2. $z_{a+1},z_{b}\in[\frac{5N}{3},2N[$,
cf. Figure~\ref{case 1 and 2}.}\\

It follows that $z_a,z_{b+1}\in[\frac{2N}{3},N[$, and that $ \sigma_a $ and $ \sigma_b$ must
be a $+p$ and $-p$ respectively.
We apply the 5/6-Lemma (Lemma~\ref{lemma 5/6}) to $$ \sigma[a,a+1] \ \beta \ \sigma[b,b+1] $$ with $\ell=N$, then shift the output by $-\Gamma_\C$.
Then we apply Point (2) of Claim~\ref{claim fix} and Point (4) of Claim~\ref{claim fix}.

\begin{center}
	\begin{figure}
\includegraphics[width=0.5\textwidth]{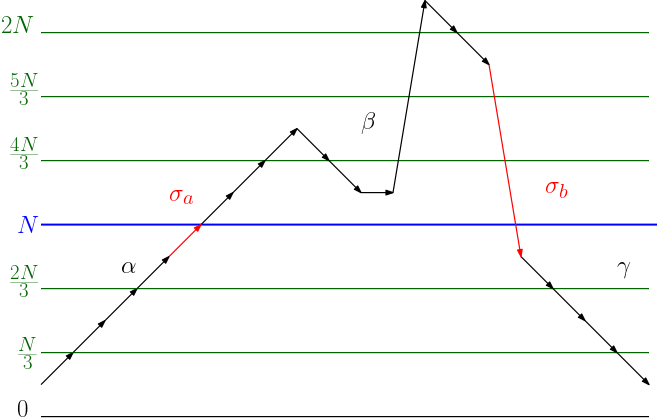}
\includegraphics[width=0.5\textwidth]{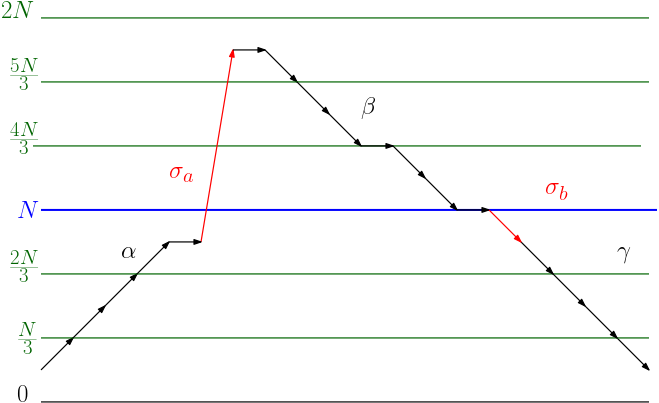}
	\caption{Illustration of Case 3, i.e. $z_{a+1}\in[N,\frac{4N}{3}[$ and $z_{b}\in[\frac{5N}{3},2N[$, on the left,
				 and Case 4, i.e. $z_{a+1}\in[\frac{5N}{3},2N[$ and $z_{b}\in[N,\frac{4N}{3}[$, on the right.}
\label{case 3 and 4}
	\end{figure}
\end{center} 

\noindent
{\em Case 3. $z_{a+1}\in[N,\frac{4N}{3}[$ and $z_{b}\in[\frac{5N}{3},2N[$,
 cf. Figure~\ref{case 3 and 4}.}\\

It follows that $z_{b+1}\in[\frac{2N}{3},N[$.
We apply the 5/6-Lemma (Lemma~\ref{lemma 5/6}) 
to $$ \beta \ \sigma[b,b+1] $$ with $\ell=N$, then shift the output by $-\Gamma_\C$.
Then we apply Point (1) of Claim~\ref{claim fix} and Point (4) of Claim~\ref{claim fix}.\\

\medskip \medskip

\noindent
{\em Case 4. $z_{a+1}\in[\frac{5N}{3},2N[$ and $z_{b}\in[N,\frac{4N}{3}[$,
 cf. Figure~\ref{case 3 and 4}.}\\

It follows that $z_a\in[\frac{2N}{3},N[$.
We apply the 5/6-Lemma (Lemma~\ref{lemma 5/6}) 
to $$ \sigma[a,a+1] \ \beta $$ with $\ell=N$, then shift the output by $-\Gamma_\C$.
Then we apply Point (2) of Claim~\ref{claim fix} and Point (3) of Claim~\ref{claim fix}.

\begin{center}
	\begin{figure}\label{case 5 and 6}
\includegraphics[width=0.5\textwidth]{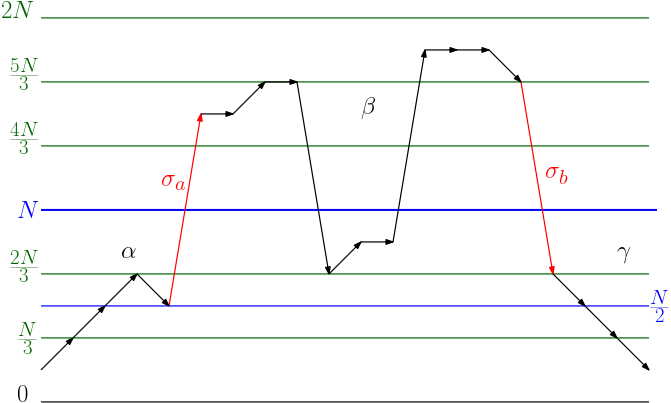}
\includegraphics[width=0.5\textwidth]{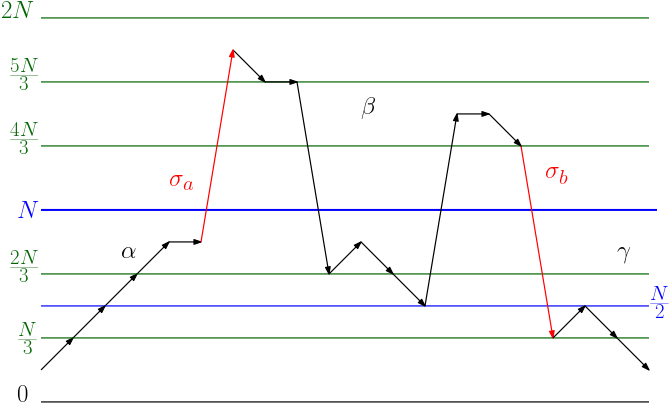}
		
	\caption{Illustration of Case 5, i.e. $z_{a+1}\in[\frac{4N}{3},\frac{5N}{3}[$
and $z_{b}\in[\frac{5N}{3},2N[$, on the left (with moreover $z_a \notin[\frac{N}{3},\frac{N}{2}[$), and Case 6, i.e. $z_{a+1}\in[\frac{5N}{3},2N[$ and $z_{b}\in[\frac{4N}{3},\frac{5N}{3}[$, on the right (with moreover $ z_{b+1} \in [\frac{N}{3},\frac{N}{2}[$ ).
		}
\label{case 5 and 6}
	\end{figure} 
\end{center}

\noindent
{\em Case 5. $z_{a+1}\in[\frac{4N}{3},\frac{5N}{3}[$
and $z_{b}\in[\frac{5N}{3},2N[$,
  cf. Figure~\ref{case 5 and 6}.}\\

It follows $z_{b+1}\in[\frac{2N}{3},N[$, and that $ \sigma_a $ and $ \sigma_b $ must
be a $+p$ and $-p$ respectively.
We distinguish whether $z_a\in[\frac{N}{3},\frac{N}{2}[$ or not.
\medskip

\noindent
{\em Case 5.A. $z_a\not\in[\frac{N}{3},\frac{N}{2}[$.}\\

It follows $z_a\in[\frac{N}{2},N[$.
We apply the 5/6-Lemma (Lemma~\ref{lemma 5/6}) to $$ \sigma[a,a+1] \ \beta \ \sigma[b,b+1] $$ with $\ell=N$, then shift the output by $-\Gamma_\C$.
Then we apply Point (2) of Claim~\ref{claim fix} and Point (4) of Claim~\ref{claim fix}.

\medskip

\noindent
{\em Case 5.B. $z_a\in[\frac{N}{3},\frac{N}{2}[$.}\\

It follows $z_{a+1}\in[\frac{4N}{3},\frac{3N}{2}[$.
We apply the 5/6-Lemma (Lemma~\ref{lemma 5/6}) to $$ \beta \ \sigma[b,b+1] $$ with $\ell=N$, then shift the output by $-\Gamma_\C$.
Then we apply Point (1) of Claim~\ref{claim fix} and Point (4) of Claim~\ref{claim fix}.

\medskip \medskip

\noindent
{\em Case 6. $z_{a+1}\in[\frac{5N}{3},2N[$ and $z_{b}\in[\frac{4N}{3},\frac{5N}{3}[$,
  cf. Figure~\ref{case 5 and 6}.}\\

It follows $z_a\in[\frac{2N}{3},N[$, and that $ \sigma_a $ and $ \sigma_b $ must
be a $+p$ and $-p$ respectively.
We distinguish whether $z_{b+1}\in[\frac{N}{3},\frac{N}{2}[$ or not.
\medskip

\noindent
{\em Case 6.A. $z_{b+1}\not\in[\frac{N}{3},\frac{N}{2}[$.}\\

It follows $z_{b+1}\in[\frac{N}{2},\frac{2N}{3}[$.
We apply the 5/6-Lemma (Lemma~\ref{lemma 5/6})
to $$ \sigma[a,a+1] \ \beta \ \sigma[b,b+1] $$ with $\ell=N$, then shift the output by $-\Gamma_\C$.
Then we apply Point (2) of Claim~\ref{claim fix} and Point (4) of Claim~\ref{claim fix}.\\

\medskip

\noindent
{\em Case 6.B. $z_{b+1}\in[\frac{N}{3},\frac{N}{2}[$.}\\

It follows $z_{b}\in[\frac{4N}{3},\frac{3N}{2}[$.
We apply the 5/6-Lemma (Lemma~\ref{lemma 5/6}) to $$ \sigma[a,a+1] \ \beta$$ with $\ell=N$, then shift the output by $-\Gamma_\C$.
Then we apply Point (2) of Claim~\ref{claim fix} and Point (3) of Claim~\ref{claim fix}.\\

\medskip

\noindent
It remains to provide the proof of the Claim~\ref{claim fix} before 
discussing the remaining case when $z_i<N$ for all $i\in[1,m-1]$

\begin{center}
	\begin{figure}
\includegraphics[width=0.5\textwidth]{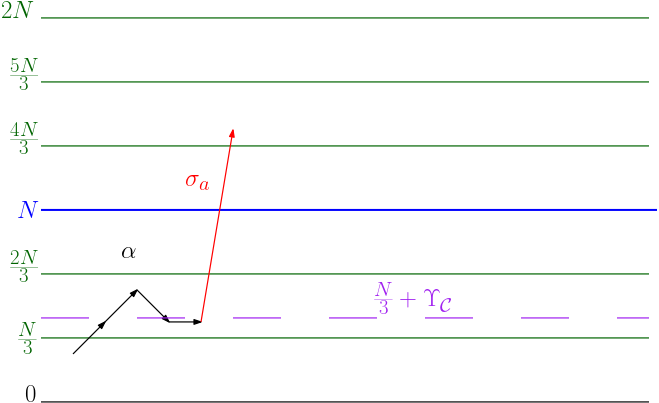}\newline
\includegraphics[width=0.5\textwidth]{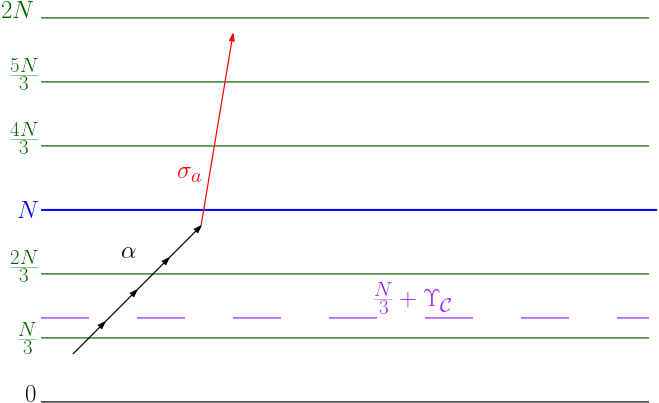}
		\caption{Claim~\ref{claim fix}: Examples for Point 1 (above) and Point 2 (below).}
	\label{prefix examples}
	\end{figure} 
\end{center}

\medskip

\noindent
\begin{proof}[Proof of Claim~\ref{claim fix}]
	Let us only prove Points (1) and (2).
	Points (3) and (4) can be proven in a symmetrical manner as Points (1) and (2). 
	Let us first prove Point (1),
	so let us assume that $z_{a+1}\in[N,\frac{5N}{3}[$.
	We refer to Figure~\ref{prefix examples} for an example of such a situation.
	Recall that $\alpha=\sigma[0,a]$,
	$z_0<\frac{N}{3}$ and 
	$z_i\in[\frac{N}{3},N[$ for all $i\in[1,a]$.

\noindent
	We first factorize $\alpha$, as seen in Figure~\ref{chi factor}, as 
	$$
	\alpha=\left(\prod_{i=1}^t \chi_i\right)\zeta,
	$$
	where 
	\begin{itemize}
	\item each $\chi_i$ is a subrun of $\alpha$ that either
	\begin{enumerate}
		\item[a)] starts in a configuration 
			with counter value 
		strictly less than $N-\Gamma_\C-\Upsilon_\C$
			and ends in the first next configuration with 
					counter value
		at least $N-\Gamma_\C$, or conversely
		\item[b)] starts in a configuration 
			with counter value 
		at least $N-\Gamma_\C$ and ends in the first next 
					configuration with 
			counter value 
		strictly less than $N-\Gamma_\C-\Upsilon_\C$, and
			\end{enumerate}
		\item the (possibly empty) suffix $\zeta$'s prefixes are neither of form a) nor b), i.e. $\values(\zeta)\subseteq[0,N-\Gamma_\C[$
			or
	$\values(\zeta)\subseteq[N-\Gamma_\C-\Upsilon_\C,N[$.
\end{itemize}

\begin{center}
	\begin{figure}
\includegraphics[width=0.7\textwidth]{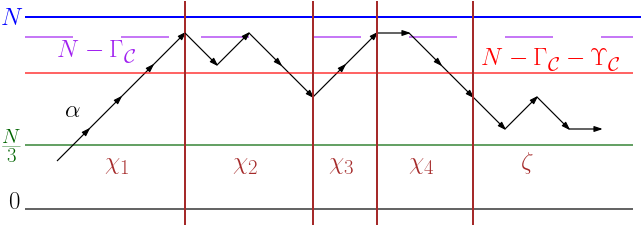}
	\caption{Illustration of the factors $\chi_1$, $\chi_2$, $\chi_3$, $\chi_4$ and $\zeta$ of $\alpha$.}
	\label{chi factor}
	\end{figure} 
\end{center}
	First observe that $\alpha$ and hence in particular $\chi_1,\ldots,\chi_t,\zeta$ all do not contain 
	any $+p$-transition nor any $-p$-transition, and that
	$|\Delta(\chi_i)|> \Upsilon_h$ for all $i \in [1,t]$ by definition. 

	Next observe that $t=0$ is possible; in this case we have $\alpha=\zeta$ and
	$\values(\alpha)\subseteq[0,N-\Gamma_\C[$.

	We will however first treat the case $t>0$, the case $t=0$ will be treated later.
	It follows 
	from $z_0<\frac{N}{3}$ 
	that $\chi_1$ must be of type a); more generally,
	$\chi_i$ is of type a) for all odd $i\in[1,t]$ and of
	type b) for all even $i\in[1,t]$.
	Since $z_0<\frac{N}{3}$, observe that if for $\alpha$'s last counter value
	$z_a$
	we have	$z_a\in[N-\Gamma_\C,N[$, then $t$ must be odd,
		and, similarly (but not entirely dually), if for $\alpha$'s last counter value
	$z_a$
	we have	$z_a\in[\frac{N}{3}, N-\Gamma_\C-\Upsilon_\C[$, then $t$ must be even.

\noindent
	In the following we prefer to write $\alpha$ as 
	$\alpha=\alpha[0,a]$ rather than $\sigma[0,a]$.
	It is important to recall that $z_0<\frac{N}{3}$ and 
	$z_s\in[\frac{N}{3},N[$ for all $s\in[1,a]$.

\medskip

\noindent
	Let 
	$$
q_0(z_0)=q_{j_1}(z_{j_1})\xrightarrow{\chi_1}q_{j_2}(z_{j_2})
\xrightarrow{\chi_2}q_{j_3}(z_{j_3})
\quad\cdots\quad
\xrightarrow{\chi_{t-1}}q_{j_{t}}(z_{j_{t}})\xrightarrow{\chi_t}
q_{j_{t+1}}(z_{j_{t+1}})
\xrightarrow{\zeta}q_{j_{t+2}}(z_{j_{t+2}}).
	$$

\noindent
Note that $j_{t+1}=j_{t+2}$ is possible if $\zeta$ is empty.
In the following, we will first show how to turn any $\chi_i$ of type a) (resp. b)) into a $(N-\Gamma_\C)$-run
with target (resp. source) configuration shifted down by $\Gamma_\C$,
and then we make a case distinction on how to end the proof based
on the parity of $t$.

\begin{subclaim}
Let $i\in[1,t]$ be odd.
Then \textcolor{black}{there} exists an $(N-\Gamma_\C)$-run $\widehat{\chi_i}$ from $q_{j_i}(z_{j_i})$ to
	$q_{j_{i+1}}(z_{j_{i+1}}-\Gamma_\C)$.
\end{subclaim}

\begin{proof}[Proof of Subclaim 1]

	Indeed, $\phi(\chi_i)=\varepsilon\in \Lambda_{8}$, as $\alpha$ contains neither $+p$-transitions nor 
	$-p$-transitions. 
	Since moreover $\Delta(\chi_i)> \Upsilon_\C$ we can now apply Lemma~\ref{lemma zero} to $\chi_i$ (viewed as an $N$-semirun)
	and obtain an $N$-semirun $\widehat{\chi_i}$ with 
	$\Delta(\widehat{\chi_i})=\Delta(\chi_i)-\Gamma_\C$ 
	that is 
	such that 
$\widehat{\chi_i}=\alpha[j_i,j_{i+1}]-I_1-I_2\cdots -I_k$ for pairwise disjoint intervals
			$I_1,\ldots,I_k\subseteq[j_i,j_{i+1}]$ such that 
						
\begin{itemize}	
\item $\phi(\alpha[I_h])\in\Lambda_{16}$,

\item $\Delta(\alpha[I_h]) \in \macro_\C \Z$ 
							and $\Delta(\alpha[I_h])>0$ 
							for all $h\in[1,k]$.
\end{itemize}

\noindent
	Recall 	
	$\values(\alpha[1,a])\subseteq[\frac{N}{3},N[$ and
	$\frac{N}{3}-\Gamma_\C>\frac{M_\C}{3}-\Gamma_\C> \Gamma_\C > \max(\Const(\C))$,
	where the inequalities follows from $M_\C$'s and 
	$\Gamma_\C$'s definition on page~\pageref{constant definitions}. 
	It follows that $\widehat{\chi_i}$
	has all its counter values (except for the first one) in  
	$ [\frac{N}{3}-\Gamma_\C , N - \Gamma_\C [$.
	Moreover, the first transition's operation must be a $+1$ update and
	therefore cannot be a test,	
	and hence
	$\widehat{\chi_i}$
	is an $(N-\Gamma_\C)$-run
	from $q_{j_i}(z_{j_i})$ to $q_{j_{i+1}}(z_{j_{i+1}}-\Gamma_\C)$.

\end{proof}

\begin{subclaim}
Let $i\in[1,t]$ be even.
Then there exists an $(N-\Gamma_\C)$-run $\widehat{\chi_i}$
	from $q_{j_i}(z_{j_i}-\Gamma_\C)$ to
	$q_{j_{i+1}}(z_{j_{i+1}})$.
\end{subclaim}

\begin{proof}[Proof of Subclaim 2.]
Analogously, by use of Lemma~\ref{lemma zero},
for $i\in[1,t]$ even,
	there exists an $(N-\Gamma_\C)$-semirun $\chi'$
	from $q_{j_i}(z_{j_i})$ to
	$q_{j_{i+1}}(z_{j_{i+1}}+\Gamma_\C)$,
	from which we obtain an 
	$(N-\Gamma_\C)$-semirun $\widehat{\chi_i}=\chi'-\Gamma_\C$
	from $q_{j_i}(z_{j_i}-\Gamma_\C)$ to
	$q_{j_{i+1}}(z_{j_{i+1}})$ by shifting $\chi'$ by $-\Gamma_\C$.
	Moreover, as $\values(\alpha[1,a])\subseteq[\frac{N}{3},N[$ and
	$\frac{N}{3}-\Gamma_\C> \Gamma_\C > \max(\Const(\C))$ as seen in the proof
	of Subclaim 1, $\widehat{\chi_i}$ is an $(N-\Gamma_\C)$-run	
	as required. \newline
\end{proof}

	To finish the proof of the existence of an $(N-\Gamma_\C)$-run from $q_0(z_0)$ to $q_{a+1}(z_{a+1}-\Gamma_\C)$
	we make a case distinction on the parity of $t$.

	Assume first that the parity of $t$ is odd.
	By applying Subclaims 1 and 2
	to the runs $\chi_1,\ldots,\chi_t$ appropriately
	 we obtain the  $(N-\Gamma_\C)$-run
	$$\widehat{\chi_1}\cdots\widehat{\chi_t}$$
	from
	$q_{j_1}(z_{j_1})$ to $q_{j_{t+1}}(z_{j_{t+1}}-\Gamma_\C)$.
	Since $t$ is odd we have that $\chi_t$ is of type a), $z_{j_{t+1}} \in [N-\Gamma_\C,N[$ 
	and 
	$\values(\zeta)\subseteq[N-\Gamma_\C-\Upsilon_\C,N[$. 
	As 
	 $(N-\Gamma_\C-\Upsilon_\C)-\Gamma_\C > (M_\C-\Gamma_\C-\Upsilon_\C)-\Gamma_\C> \max(\Const(\C))$,
	following from $M_\C$'s, $\Gamma_\C$'s and $\Upsilon_\C$'s definition on page~\pageref{constant definitions},
	and $\phi(\alpha)= \varepsilon$,
	it follows that
	$$
	\widehat{\chi_1}\cdots\widehat{\chi_t}\left(\zeta-\Gamma_\C\right)
	$$
	is an $(N-\Gamma_\C)$-run from 
	$q_{j_1}(z_{j_1})$ to $q_{j_{t+2}}(z_{j_{t+2}}-\Gamma_\C)$.
	Recall that $z_{a+1} < \frac{5N}{3}$ by case assumption and also recall that $N > M_\C$.
	Since $\values(\zeta)\subseteq[N-\Gamma_\C-\Upsilon_\C,N[$ we have
	$z_a=z_{j_{t+2}} \geq N-\Gamma_\C-\Upsilon_\C$ and hence 
	$0<\Delta(\sigma,{a}) < \frac{5N}{3}-(N-\Gamma_\C-\Upsilon_\C)
	\leq \frac{2N}{3} + \Gamma_\C + \Upsilon_\C < \frac{2N}{3} + \frac{M_\C}{3} < N$,
	where the penultimate inequality follows from the definition of $M_\C$
	on page~\pageref{constant definitions}.
	Hence, as $ \sigma_a $ is not a test nor a 
	$+p$-transition we have that 
	$$
	\widehat{\chi_1}\cdots\widehat{\chi_t}
	\left(\zeta-\Gamma_\C\right)( \sigma[a,a+1] -\Gamma_\C)
	$$
	is an $(N-\Gamma_\C)$-run 
	from $q_{j_1}(z_{j_1})=q_0(z_0)$ to $q_{j_{t+2}}(z_{j_{t+2}}-\Gamma_\C)=
	q_{a+1}(z_{a+1}-\Gamma_\C)$ as required.

\medskip

\noindent
	Let us now treat the case when $t$ is even. It follows
	$\values(\zeta)\subseteq[0,N-\Gamma_\C[$, in particular $z_a\in[0,N-\Gamma_\C[$. 
	Again,
	$$\widehat{\chi_1}\cdots
	\widehat{\chi_t}$$
	is an $(N-\Gamma_\C)$-run from $q_{j_1}(z_{j_1})=q_0(z_0)$ to
	$q_{j_{t+1}}(z_{j_{t+1}})$.
	Since $z_{a+1}\geq N$ and $z_a<N-\Gamma_\C$ it follows that
	$ \sigma_a $ is 
	a $+p$-transition, in particular	
	$\Delta(\sigma, a)>\Gamma_\C$. Thus, 
	$$
\widehat{\chi_1}\cdots\widehat{\chi_t}\zeta \tau,
	$$
	where $\tau = q_a(z_a) \xrightarrow{ \sigma_a , N-\Gamma_\C} q_{a+1}(z_{a+1}-\Gamma_\C)$ with
	$\Delta(\tau)=\Delta(\sigma, a)-\Gamma_\C= N-\Gamma_\C$\textcolor{black}{,} 
	is an $(N-\Gamma_\C)$-run
	from $q_{j_1}(z_{j_1})=q_0(z_0)$ to $q_{j_{t+2}}(z_{j_{t+2}}-\Gamma_\C)=
	q_{a+1}(z_{a+1}-\Gamma_\C)$, as required.

It remains to discuss the case when $t=0$. This case can be proven analogously.
Indeed, from $t=0$ it follows immediately that $\alpha=\zeta$ and $\values(\zeta)\subseteq[0,N-\Gamma_\C[$
and the proof 
is analogous as the case when $t>0$ and when $t$ is even.

\medskip

\noindent
Let us now sketch the proof of Point (2) of Claim~\ref{claim fix}. 
Let us assume $z_a\in[\frac{N}{3}+\Upsilon_\C,N[$.
Similarly as in \textcolor{black}{Point} 
 (1) we can factorize $\alpha$ as 
$\alpha=\left(\prod_{i=1}^t\chi_i\right)\zeta$
 and
 Subclaims 1 and 2 hold again.

If $t$ is odd, then 
by Subclaims 1 and 2
 we have that the run
$\left(\prod_{i=1}^t\widehat{\chi_i}\right)(\zeta-\Gamma_\C)$,
stipulating that $\widehat{\chi_i}$ is the of Subclaims 1 and 2 respectively 
(depending on the parity of $i$), 
is the desired $(N-\Gamma_\C)$-run from $q_0(z_0)$ to $q_a(z_a-\Gamma_\C)$.

If $t$ is even, then again by Subclaims 1 and 2 we have that
$\xi=\left(\prod_{i=1}^t\widehat{\chi_i}\right)\zeta$ 
is an $(N-\Gamma_\C)$-run
from $q_0(z_0)$ to $q_a(z_a)$, where again $\widehat{\chi_i}$ is defined as above.
By definition the run $\xi$ does not contain any $+p$-transitions nor $-p$-transitions, 
thus $\phi(\xi)=\varepsilon\in\Lambda_8$.
By construction also the run $\xi$ has all counter values, besides the first, above
$\frac{N}{3} -\Gamma_\C  >  \Gamma_\C + \max(\Const(\C))$.
Moreover, as $z_a \geq \frac{N}{3}+\Upsilon_\C$ and $z_0 < \frac{N}{3}$,
we have $\Delta(\xi)>\Upsilon_\C$.
We can thus apply Lemma~\ref{lemma zero} to $\xi$,
obtaining 
an $(N-\Gamma_\C)$-run $\xi'$ from $q_0(z_0)$ to $q_a(z_a-\Gamma_\C)$, as required.

\end{proof}

We now conclude the proof   of 
 our statement 
by treating the only remaining
case, the case when $\sigma$ is such that $z_i<N$ for all $i\in[1,m-1]$. 
In this case we can factorize $\sigma$ as
$
\sigma=\prod_{i=1}^t\chi_i \zeta
$
similarly as done in the proof of Point (1) of Claim~\ref{claim fix}, where $t$ is even,
and analogously prove that
$\prod_{i=1}^t\widehat{\chi_i}\zeta$ is an $(N-\Gamma_\C)$-run from $q_0(z_0)$ to $q_m(z_m)$,
where $\widehat{\chi_i}$ is the output of Subclaim 1 and 2 respectively (depending on the parity of $i$).

\section{Conclusion}
In this paper we have shown that the reachability problem
for parameteric timed automata with two parametric clocks and one
parameter is complete for exponential space.

For the lower bound proof, inspired by \cite{GHOW10,GollerL13}, we made use of two results from complexity theory.
First, we made use of a serializability characterization of $\EXPSPACE$ from \cite{GHOW10}
which is a padded version of the serializability characterization
of $\PSPACE$ from~\cite{HLSVW93}, which in turn has
its roots in Barrington's Theorem~\cite{Bar89}.
Second, we made use of a result of Chiu, Davida, Litow that 
states that numbers in chinese remainder representation can
be translated into binary representation in $\mathsf{NC}^1$ 
(and thus in logarithmic space).
We are convinced that it is worthwhile to develop a suitable
programming language that serves as a unifying framework
in that it provides an interface for proving lower bounds
for various problems involving automata.
In a sense, we have developed the corresponding interface ``by hand''
when defining how parametric timed automata can compute functions
(Definition~\ref{def computes}).

For the $\EXPSPACE$ upper bound we first followed the approach
of Bundala and Ouaknine~\cite{BundalaO17} by providing an exponential
time translation from reachability in parametric timed automata
with two parametric clocks and one parameter (i.e. $(2,1)$-PTA) to reachability in parametric
one-counter automata (POCA) over one parameter, yet on a slightly less expressive POCA model
as introduced in~\cite{BundalaO17}.
We then studied the reachability in POCA with one parameter $p$.
Our main result, the Small Parameter Theorem (Theorem~\ref{theorem upper}), states that 
such a parametric one-counter automaton (POCA)
has an accepting run all of whose counter values lie in $[0,4\cdot p]$ if, and only if, there
exists such an accepting run for some $p$ that is at most exponential in
the size of the POCA.
Since the translation from $(2,1)$-PTA to POCA is computable in exponential time,
this gives a doubly exponential upper bound on the parameter value of the original 
$(2,1)$-PTA and hence an $\EXPSPACE$ upper bound for reachability in 
$(2,1)$-PTA (Corollary~\ref{corollary}).

In proving the Small Parameter Theorem we introduced the notion of semiruns
and gave several techniques for manipulating them. 
The Depumping Lemma~\ref{lemma zero} allowed us to construct from semiruns with large absolute
counter effect new semiruns with a smaller absolute counter effect. 
The Bracket Lemma~\ref{bracket lemma} allowed us to find in semiruns having a sufficiently 
large absolute counter effect and satisfying some majority condition
on the number of occurrences of $+p$-transitions and $-p$-transitions some subsemirun
that has again a large absolute counter effect and moreover some equal bracketing 
properties.
Our Hill and Valley Lemma (Lemma~\ref{lemma hill/valley}) allowed to turn for sufficiently large $N$ any
$N$-semirun that is either a hill or a valley
into an $N'$-semirun for some $N'<N$.
Our $5/6$-Lemma (Lemma~\ref{lemma 5/6}) allowed to turn for sufficiently large $N$ any $N$-semirun
 with an absolute counter effect of at most $5/6N$ into an $N'$-semirun for some
$N'<N$.

We hope that extensions of our techniques provide a line of attack for finally showing decidability
(and the precise complexity) of reachability in parametric timed automata with two parametric
clocks over an arbitrary
number of parameters (i.e. $(2,*)$-PTA). For these however, it seems that the reduction to
POCA indeed requires the presence of the above-mentioned $+[0,p]$-transitions.
When analyzing runs in the corresponding more general POCA model that in turn also 
involves an arbitrary number of parameters, it will
become necessary to ``de-scale'' semiruns in the following sense.
Already in the presence of two parameters one can see that it becomes necessary
to decrease the value of both parameters simultaneously proportionally:
for instance one can build a $(2,2)$-PTA for which reachability holds only
if the first parameter is a multiple of the second parameter.

\newpage
\bibliography{bibliography/bib}

\newpage

\appendix
\section{Proof of the reduction (Theorem~\ref{theorem ptatopoca})}\label{appendix ptatopoca}

In this appendix, we concern ourselves with proving the reduction from parametric timed automata with two parametric clocks and one parameter to 
parametric one-counter automata with one parameter, which was stated in
the main theorem of Section~\ref{ptatopoca section},
Theorem~\ref{theorem ptatopoca}.

\subsection{Overview of the proof of the reduction}\label{poca reduction proof overview}

In this section, we recall the main theorem of Section~\ref{ptatopoca section}, and provide its proof overview.

Let us first recall Theorem~\ref{theorem ptatopoca}.

\begin{samepage}
\begin{theorem}(Theorem~\ref{theorem ptatopoca}) 
        The
        following is computable in exponential time:

	INPUT: A $(2,1)$-PTA $\A$.
	
	OUTPUT: A POCA $\C$ over one parameter\\
	such that
	\begin{enumerate}
\item for all $N \in \N$
all accepting $N$-runs $\pi$ in $\C$ satisfy 
$\values(\pi)\subseteq[0,4 \cdot \max(N,| \C|)]$, and
\item reachability holds for $\A$ if, and only if, 
	reachability holds for $\C$.
	\end{enumerate}
\end{theorem}
\end{samepage}
	A more general (but strictly speaking incomparable) 
	result involving two parametric clocks
	but an arbitrary number of parameters instead of only one
has already been proven in \cite{BundalaO17},
however with a different POCA formalism:
Bundala and Ouaknine's model for POCA differs in that
	it contains operations 
	that allow to nondeterministically add to the counter 
	a value that lies in $[0, p]$. 
By restricting ourselves to the case of only one parameter $p$, we will prove
in a thorough analysis that we no longer need such operations in the construction.

As in \cite{BundalaO17} we follow the following proof strategy:
\begin{itemize}
\item In Subsection~\ref{A to B section}, we reduce the reachability problem of a parametric timed automaton $\A=(Q_\A,\Omega_{\!\A},P,R_\A,q_{\A},F_\A)$
--- in our setting later with two parametric clocks ---
		to the reachability problem of a 
		so-called {\em parametric $0/1$ timed automaton}
		$\B=(Q_\B,\Omega_\B,P,R_{\B,0},R_{\B,1}, q_{\B},F_\B)$, 
where $\Omega_\B \subseteq \Omega_{\!\A}$ contains only the 
		non-parametric clocks of $\Omega_{\!\A}$,  and $\Const(\B)=\{ 0 \}$.

\item In Subsection~\ref{preliminary construction section}
we present the region abstraction technique introduced by Alur and Dill in \cite{alur1994theory} 
		to mimic region-restricted runs (runs inside a region)
		of parametric $0/1$ timed automata
		with one parameter by arithmetic progressions.

\item Finally, we present the final step of the reduction 
	in Subsection~\ref{A to C}, where it is shown
	how to use the above-mentioned technique to mimic 
		reset-free region-restricted runs in $\B$, and 
furthermore how to provide a construction in order to 
		mimic resets in $\B$.
The precise construction itself mainly deviates from \cite{BundalaO17}
in the gadget construction for resets.
\end{itemize}

\subsection{How to remove non-parametric clocks and non-parametric guards}\label{A to B section}

In this subsection we show how non-parametric guards and non-parametric clocks can 
be eliminated from parametric timed automata.
Initially introduced in \cite{AHV93-stoc} we define
the notion of parametric $0/1$ timed automata:
these are essentially parametric timed automata in which each rule 
dictates 
whether a unit of time passes or not.
Alur, Henzinger and Vardi have already shown 
in \cite{AHV93-stoc} how the reachability problem for 
parametric timed automata can be reduced to the 
reachability problem for parametric $0/1$ timed automata
that do not contain any non-parametric clocks. 
We will provide in Lemma~\ref{eliminating non-parametric clocks} below an analogous reduction by not only eliminating all non-parametric clocks, but also all non-parametric guards (except
for empty guards).

A {\em parametric $0/1$ timed automaton } ({\em $0/1$-PTA} for short) is a tuple
$$\B=(Q,\Omega,P, R_0, R_1, q_{init}, F),$$ where
$\B_i=(Q,\Omega,P, R_i, q_{init}, F)$ is a PTA for all $i \in \{0,1\}$.
For simplicity we define its {\em size}
as $|\B|=|\B_0|+|\B_1|$.
Analogously, a clock $\omega\in \Omega$ is {\em parametric}
if it is parametric in $\B_0$ or in $\B_1$.
We analogously denote the constants of $\B$ 
by $\Const(\B)$ and its configurations by  $\Conf(\B)$.

\begin{definition}
	For each $i\in\{0,1\}$,
	each parameter valuation $\mu:P\rightarrow\N$ 
and each $(\delta,t)\in R_i\times\N$ with $\delta = (q,g,U,q')\in R_i$,
 we 
define the binary relation $\xrightarrow{\delta,i,\mu}$ over 
$\Conf(\B)$ as 
$q(v)\xrightarrow{\delta,i, \mu} q'(v')$ if
	$v+i \models_{\mu} g$, 
	$v'(u)=0$ for all $u \in U$ and $v'(\omega)=v(\omega)+ i$ for all $\omega \in \Omega \setminus U$.
\end{definition}

As expected, we write $q(v)\xrightarrow{\mu}q'(v')$ if 
$q(v)\xrightarrow{\delta,i,\mu}q'(v')$ for some 
$i\in\{0,1\}$, and some $\delta \in R_i$.
The notions of a (reset-free) $\mu$-run (resp. $N$-run) and when reachability holds
for $\B$ are also defined as expected.

The convention used in this and the following subsections is that parametric $0/1$ timed automata are denoted by $\B$. 
The main result of this subsection is the following lemma,
stated slightly less general in \cite{AHV93-stoc} in that 
there is no requirement $\Const(\B)=\{0\}$.

\begin{samepage}
\begin{lemma}[\cite{AHV93-stoc}]
	\label{eliminating non-parametric clocks}
The following is computable in exponential time:\newline

        INPUT: A PTA $\A=(Q_\A,\Omega_{\!\A},P,R_\A,q_{\A},F_\A)$.

        OUTPUT: A  $0/1$-PTA $\B=(Q_\B,\Omega_\B,P,R_{\B,0},R_{\B,1}, q_{\B},F_\B)$, where $\Omega_\B \subseteq \Omega_{\!\A}$ contains precisely the parametric clocks of 
	$\Omega_{\!\A}$, $\Const(\B)=\{ 0 \}$, and
         such that
                          reachability holds for $\A$, if, and only if,
                          reachability holds for $\B$.
\end{lemma}
\end{samepage}

We adjust the proof from \cite{AHV93-stoc}. While the idea of the construction remains the same, ours slightly deviates  in that we explicitly have
$\Const(\B)=\{ 0 \}$, i.e. we remove all non-parametric guards of the form $\omega \bowtie c$ with $c \neq 0$ as well as 
all non-parametric clocks.

\begin{proof}
Let us assume without loss of generality that $\A$ contains at least one parametric clock and let us fix one such clock $x$. We define the empty guard $g_\epsilon$ as $g_\epsilon = x \geq 0$ and observe that this guard is always satisfied.
Let $c_{max} = \max(\Const(\A))$ denote the largest constant appearing in $\A$.
Note that once the value assigned to a clock $\omega$ by a valuation $v$ is strictly above  $c_{max}$, the 
	precise value $v(\omega)$ is no longer of importance,
	merely the fact that $v(\omega)$ exceeds $c_{max}$ is relevant.
	Since we work with discrete time configurations, the value assigned to $\omega$ is always a non-negative integer. 
	We will eliminate all non-parametric clocks of $\Omega_{\!\A}$ by storing in the state space of $\B$ 
	the values of clocks up to $c_{max}+1$, where $c_{max}+1$ 
	will stand for any value greater $c_{\max}$.
	Moreover we eliminate all non-empty non-parametric 
	guards by also storing in the state space of $\B$ the values of 
	parametric clocks in the same fashion. 
	Formally, we define $ \Omega_\B = \{ \omega \in \Omega_{\!\A}\ | \text{ $\omega$ is parametric} \} $, $ Q_{\B} =  Q_{\A} \times [0, c_{max} + 1]^{  \Omega_{\!\A}  }$, $P$ is the same in both automata,  
	$F_\B = F_{\A} \times [0, c_{max} + 1]^{  \Omega_{\!\A}  }$,
	and $q_{\B} = (q_{\A}, v_0)$, where 
	$v_0(\omega)=0$ for all $\omega \in \Omega_{\!\A}$.

We ensure that the stored clocks progress simultaneously with the remaining parametric clocks by exploiting the fact 
	that the rules dictate whether or not time elapses, and build the rules of $\B$ such 
	that the $+1$ rules correspond to the progress of time in $\A$ whereas the $+0$ rules correspond to using a rule in $\A$. Formally,

\begin{itemize}
	\item for every $q \in Q_{\!\A}$, $ v \in [0, c_{max} + 1]^{  \Omega_{\!\A}  } $, 
 we introduce a rule of the form 
		$( (q,v) ,  g_\epsilon , \emptyset , (q,v'))$ in $R_{\B,1}$, where $v'(\omega)=\min\{v(\omega)+1,c_{max}+1\}$ 
		for all $\omega \in \Omega_{\!\A}$,
\item for every $(q,g,U,q') \in R_\A$ with $g \in \G(\Omega_\B,P)$ a parametric guard, 
every  $ v \in [0, c_{max} + 1]^{  \Omega_{\!\A}  } $
we introduce a rule $( (q,v) , +0 , g , U', (q',v'))\in R_{\B,0}$, 
where $v'$ 
is obtained from $v$ 
except for assigning $0$ to 
every clock in $U$ and $U' = U \cap \Omega_\B$ is the subset of parametric clocks of $U$, and

\item for every $(q,g,U,q') \in R_\A$ with $g \in \G(\Omega_{\!\A},P)$ a non-parametric guard, 
every  $ v \in [0, c_{max} + 1]^{  \Omega_{\!\A}  } $
 such that $v \models g$, 
we introduce a rule $( (q,v) , +0 , g_\epsilon , U', (q',v'))\in R_{\B,0}$, 
where $v'$ is obtained from $v$ 
except for assigning $0$ to 
every clock in $U$ and $U' = U \cap \Omega_\B$ is the subset of parametric clocks of $U$. 
\end{itemize}

\end{proof}

For the remaining subsections, let us fix a PTA $\A=(Q_\A,\Omega_{\!\A},P,R_\A,q_{\A},F_\A)$ with two parametric clocks $x$ and $y$, and with $P=\{p\}$. Let us
also fix the $0/1$-PTA $\B=(Q_\B,\Omega_\B,P,R_{\B,0}, R_{\B,1},q_{\B},F_{\B})$ 
produced by Theorem~\ref{eliminating non-parametric clocks} applied to PTA $\A$,
and recall that $\B$ satisfies
\begin{itemize}
\item $P=\{p\}$,
\item $\Omega_\B = \{x,y\}$, where $x$ and $y$ are parametric, 
\item $\Const(\B)=\{ 0 \}$, and
\item reachability holds for $\A$ if, and only if, reachability
	holds for $\B$.
\end{itemize}

\subsection{Capturing reset-free runs via the region abstraction technique}\label{preliminary construction section}

In this section we perform another preliminary construction before providing the proof of Theorem~\ref{theorem ptatopoca}. 
We build parametric one-counter automata without tests and 
with updates only in
$\{+0, +1\}$ that can mimic the behavior of parametric 
$0/1$ timed automata with two parametric clocks 
and one parameter 
inside a reset-free run
having only clocks valuations in a certain set. 
We first simply remove rules resetting at least one clock.
We then show how to remove non-empty guards from parametric $0/1$ timed automata taking inspiration from the region abstraction technique for timed automata first introduced in \cite{alur1994theory}. The technique appears already in the proofs of reduction
from parametric timed automata with two clocks to parametric one-counter automata
given in \cite{haase2012complexity,HaaseOW16} (for empty sets of parameters) and in \cite{BundalaO17}.
We refer to \cite{baier2008principles} for further discussions on the region abstraction technique.

Recall that our fixed $0/1$-PTA $\B$ satisfies
$P=\{p\}$,
$\Omega_\B = \{x,y\}$, where $x$ and $y$ are parametric, and
$\Const(\B)=\{0\}$.

Let us now explain the set of regions.
 For any valuation $\mu$ that assigns to our only paramter $p$
 the value $N$ we prefer to write $\models_N$ instead of 
 $\models_\mu$.
Moreover, we prefer
to view clock valuations $v : \{ x, y \} \rightarrow \N$ as pairs $(v(x),v(y))$. Sets of
clock valuations will correspondingly be denoted as subsets 
of $\N \times \N$.
The regions are essentially, 
when assigning $N$ to the one parameter $p$, maximal subsets of $\N \times \N$ 
 equivalent with regards to the sets of guards of $\B$ their valuations satisfy.
In other words, the regions we define are equivalence classes for the relation $\sim_N$, where 
$v \sim_N v'$ if for all possible guards $g$ of $\B$ we have $ v \models_N g$ if, and only if, $v' \models_N g$. 
Since the latter guards can only compare (using comparisons $<, \leq, = , \geq, >$)  the 
clock valuations against values from the
set $\{0,N\}$, it follows that $\sim_N$ has at most the following $16$ equivalence classes, each of which we call {\em region}
in the following:

	\begin{eqnarray*}
			&( 0 , 0 ), (0, N),(N,0),(N,N)  \\
			&\text{ to respectively denote the singleton sets } \{ (0,0) \},  \{ (0,N) \}, \{ (N,0) \}, \{ (N,N) \}, \\ \\
			&( 0 , 0 ) \leftrightarrow ( 0 , N ), (N , 0 ) \leftrightarrow ( N , N ),
( 0 , N ) \leftrightarrow ( 0 , +\infty ),  ( N , N ) \leftrightarrow ( N , +\infty ), \\
	&\text{ to respectively denote the sets }  \\
	& \{ (0,i) \ | \ 0 < i < N \}, \{ (N,i) \ | \ 0 < i < N \}, \{ (0,i) \ | \  i > N  \}, \{ (N,i) \ | \  i > N  \}, \\ \\
&( 0 , 0 ) \leftrightarrow ( N , 0 ), (0 , N ) \leftrightarrow ( N , N ),
(  N,0 ) \leftrightarrow (  +\infty,0 ),  ( N , N ) \leftrightarrow ( +\infty,N ), \\
	&\text{ to respectively denote the sets }  \\
	& \{ (i,0) \ | \ 0 < i < N \}, \{ (i,N) \ | \ 0 < i < N \}, \{ (i,0) \ | \  i > N  \}, \{ (i,N) \ | \  i > N  \}, \\ \\
	& \textsc{Lower-Left}, \textsc{Upper-Left}, \textsc{Lower-Right}, \textsc{Upper-Right} \\
	&\text{ to respectively denote the sets }  \\
	& \{ (i,j) \ | \ 0 < i,j < N \},\{ (i,j) \ | \ 0 < i < N, j > N \}, \{ (i,j) \ | \ i > N, 0 < j < N  \}, \{ (i,j) \ | \  i,j > N  \}. 
	\end{eqnarray*}

We refer to Figure~\ref{region} for an illustration
of the different regions.

\begin{center}
	\begin{figure}
		\hspace{2cm}
\includegraphics[width=0.7\textwidth]{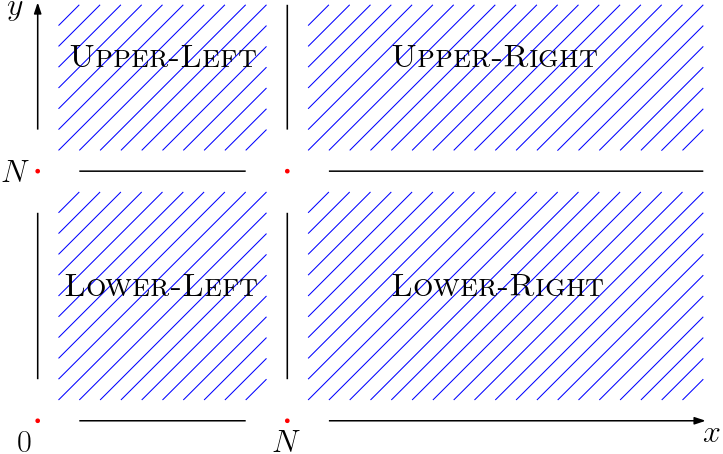}
	\caption{An illustration of the different regions. }
	\label{region}
	\end{figure}
\end{center}

\renewcommand{\R}{\mathcal{R}}

\noindent
As expected, for every guard $g$ of $\B$
and every region $\R$ we write $\R \models_N g$ if $v\models_N g$ for all $v \in \R$. 
For each region $\R$  we say a run 
$q_0(v_0)\xrightarrow{\delta_1,i_1,\mu}q_1(v_1)\cdots\xrightarrow{\delta_n,i_n,\mu} q_n(v_n)$
of $\B$ is {\em $\R$-restricted} if $v_j\in\R$ for all $j\in[0,n]$.

We remark that, in any $\R$-restricted run in $\B$, the set of guards being satisfied or
not are the same for all configurations appearing in it. 
Thus, the set of guards that are satisfied only depend on the region 
and not the particular configurations of the $\R$-restricted run.
We simply write $\R\models g$ when a region $\R$ satisfies guard $g$.

We use this property to remove guards from the parametric $0/1$ timed automaton
$\B$ while still mimicking reset-free $\R$-restricted runs. 

For each region $\R$ we introduce the {\em region automaton} 
$\B_{\R}$ obtained from $\B$ instantiating all comparisons appropriately and by 
removing all rules that reset some clock. 
We fix $g_\epsilon$ to be the empty guard  $ x \geq 0$.
Formally, the automaton $\B_{\R}$ is the $0/1$-PTA obtained from $\B$ by 
\begin{samepage}
\begin{itemize}
\item removing all rules $(q,g,U,q')$ with $U \neq \emptyset$,
\item removing all rules $(q,g,\emptyset,q')$ for which $\R \nmodels g$, and
\item replacing all rules  $(q,g,\emptyset,q')$ for which $\R \models g$ by 
	$(q,g_\epsilon,\emptyset,q')$.
\end{itemize}
\end{samepage}

The following lemma is immediate.

\begin{lemma}\label{region simulation}
	From the  $0/1$-PTA $\B=(Q_\B,\{x,y\},\{p\}, R_{\B,0}, R_{\B,1}, q_{\B}, F_\B)$  
	with $\Const(\B)=\{0\}$ one can compute in polynomial time (in $|\B|$) 
	the sixteen
	$0/1$-PTA $\{ \B_\R | \ \R\text{ is a region}\}$ such that for all $N \in \N$, all 
	regions $\R$, and all configurations $q(v)$ and $q'(v')$
	for which $v,v'\in\R$ the following are equivalent:
	\begin{itemize}
	\item There exists an $\R$-restricted reset-free $N$-run from $q(v)$ to $q'(v')$ in $\B$.
	\item There exists an $N$-run from $q(v)$ to $q'(v')$ in $\B_\R$.
	\end{itemize}
\end{lemma}

\subsubsection{Capturing reset-free runs via arithmetic progressions}

A {\em one-counter automaton} is a POCA $\C = (Q,P,R,q_{init},F)$ with $P=\emptyset$, and with
only $>0$, $\geq 0$, and $= 0$ tests. 
As $P=\emptyset$, we write $q(z)\xrightarrow{} q'(z')$
instead of $q(z)\xrightarrow{\mu}q'(z')$ for one-counter automata.

Given a one-counter automaton $\C$ and two of its 
control states $q$ and $q'$ we define the set $\Pi(\C,q,q')$ of counter values 
that configurations in control state $q'$ can have from runs
starting in $q(0)$:
$$
\Pi(\C,q,q') = \{ v \in \N \ | \ q(0) \rightarrow^* q'(v) \}. 
$$

For all $a\geq 0$ and $b\geq 1$ we define the arithmetic progression $a+b\N$ as 
$a+b\N=\{a + b\cdot n \mid n\in\N\}$. 
The following theorem in an immediate consequence of a result by To
analyzing the succinctness between unary finite automata and arithmetic progressions \cite{To09B}.

\begin{theorem}[Theorem 2 in \cite{To09B}]\label{semilinearity}
Let $\C=(Q,\emptyset,R,q_{init},F)$ be a one-counter automaton with $+0, +1$ updates only. Then for every two control states $q,q' \in Q$ one can compute in polynomial time a set
$\{(a_j,b_j) \in \N^2\mid j \in[1,r]\}$  such that
$\Pi(\C, q, q') = \bigcup_{1 \leq j \leq r} a_j + b_j \N $, where moreover $r \in O(|Q|^2)$,
$a_j \in O(|Q|^2)$, and $b_j \in O(|Q|)$ for all $j \in [1,r]$.

\end{theorem}
We remark that Theorem~\ref{semilinearity} also holds in the presence
of transitions that decrement the counter, cf. Lemma 6 in \cite{GMT09}.

\begin{remark}\label{Region POCA}
	Let $\R$ be a region and let $\B_{\R}=(Q_{\B_\R},\{x,y\},\{p\},R_{\B_\R,0},R_{\B_\R,1},
	q_{\B_\R,init},F_{\B_\R})$ be the $0/1$-PTA for $\R$.
	Then all rules in $R_{\B_\R,0}\cup R_{\B_\R,1}$ have as guard the empty guard $g_\varepsilon$.
	Let 
	$\widehat{\B_\R}=(Q_{\B_\R},\emptyset,R,q_{\B_\R,init},F_{\B_\R})$
be the one-counter automaton, 
	where $$R=\{(q,+i,q')\mid (q,g_\varepsilon,\emptyset,q')\in R_{\B_\R,i}, i\in\{0,1\}\}$$
	only contains $+0$ and $+1$ updates and does not contain any $=0$-tests,
	Then for all $(k,\ell)\in\N\times\N$, all $q,q'\in Q$, and
	all $n\in\N$ the following are equivalent:
	\begin{itemize}
		\item There is a run from $q(k,\ell)$ to $q'(k+n,\ell+n)$ in $\B_\R$.
		\item There is a run from $q(0)$ to $q'(n)$ in $\widehat{\B_\R}$.
	\end{itemize}
	Notably, every run from $q(k,\ell)$ to $q'(k',\ell')$ in $\B_{\R}$ satisfies
	$k'=k+n$ and $\ell'=\ell+n$ for some $n\in\N$.
\end{remark}

We apply Theorem~\ref{semilinearity} to all one-counter automata $\widehat{\B_R}$
from Remark~\ref{Region POCA}.
This yields the following characterization.

\begin{lemma}\label{lemma arithmetic progression}
	From the $0/1$-PTA $\B=(Q_\B,\{x,y\},\{p\}, R_{\B,0}, R_{\B,1}, q_{\B}, F_\B)$  
for every two control states $q,q' \in Q_\B$, for all regions $\R$
	one can compute in polynomial time (in $|\B|$) a set
$\{(a_j,b_j) \in \N^2 | \ j \in[1,r]\}$  
	such that for all $N,t\in \N$ and all $v,v+t\in\R$ the following are equivalent:
\begin{itemize}
\item There exists a reset-free $\R$-restricted $N$-run from $q(v)$ to $q'(v + t)$ in $\B$.
\item $t \in \bigcup_{1 \leq j \leq r} a_j + b_j \N $.
\end{itemize}
Moreover, $r \in O(|Q_\B|^2)$,
$a_j \in O(|Q_\B|^2)$, and $b_j \in O(|Q_\B|)$ for all $j \in [1,r]$.
\end{lemma}

\subsection{Proof of Theorem~\ref{theorem ptatopoca}: Construction of $\C$}\label{A to C}

Let us recall the fixed $0/1$-PTA  $\B=(Q_\B,\Omega_\B,\{p\},R_\B,q_{\B},F_\B)$ the $0/1$-PTA obtained from $\A$ by Theorem~\ref{eliminating non-parametric clocks}, and recall that $\B$ satisfies
\begin{itemize}
\item $P=\{p\}$,
\item $\Omega_\B = \{x,y\}$ and $x$ and $y$ are parametric, and
\item $\Const(\B)=\{ 0 \}$.
\end{itemize}
Recall also the set of regions of $\B$ defined in Section~\ref{preliminary construction section}.
We want to construct some POCA $\C=(Q_\C,\{p\},R_\C,q_c,F_\C)$ such that 
reachability holds for $\B$ if, and only if, reachability holds for $\C$ and moreover
for all $N \in \N$, every accepting $N$-run $\pi$ in $\C$ 
	satisfies
$\values(\pi)\subseteq[0, 4 \cdot \max(N,|\C|)]$.

\medskip

\noindent
The to be constructed POCA $\C$ (again over one parameter
that will be evaluated to the same value as the only parameter of $\B$)
will test whether an accepting $N$-run exists in $\B$ by using the definitions of regions and Lemma~\ref{lemma arithmetic progression} from the last subsection, but also using additional
gadgets to mimic the reset of a clock inside a particular region. 

\medskip

\noindent
In what follows we denote the current value of the counter of $\C$ by $z$. 
For the time being in our construction $z$ can be negative: we will later show how to 
obtain non-negativity and the required restriction that all $N$-runs $\pi$
of $\C$ satisfy $\values(\pi)\subseteq[0,4N]$.

\medskip

\noindent
The idea of the reduction is to factorize any possible accepting
$N$-run into maximal reset-free subruns. 
We will use the current counter value $z$ of $\C$ to store the clock valuation difference 
$v(x)-v(y)$, thus initially $0$. 
We remark that between two consecutive resets, the difference $v(x)-v(y)$ stays the same throughout, 
but after some clock of $\Omega_\B$ (either $x$ or $y$) is reset, this particular reset 
clock will be equal to zero but not necessarily the other one.
The counter of $\C$ therefore needs to be modified accordingly. 
As expected, we construct $\C$ in such a way that after a reset of $y$, the counter value 
$z$ equals $v(x)$, and after a reset of $x$ the counter value $z$ equals
$-v(y)$. 
See Figure~\ref{xyz} for an idea of the relationship between $v(x),v(y)$ and 
$z$ along the curve of the clock values.

Notice that once the value of a clock becomes strictly larger than $N$, its exact value is irrelevant to any future parametric comparison in $\B$, hence one only needs to remember that its value 
is strictly larger than $N$. 
Thus, our counter $z$ will only track the values $v(x)$ and $v(y)$ up to $N$ 
and possibly remember which of the two clock values exceeds $N$.
Therefore, when a reset occurs and we store the value of the other clock in the counter, if it exceeds this 
$N$
 we can and will  replace it by $N+1$, and if it is strictly below $-N$, we can and will replace it by $-N-1$. Let us therefore assume for now that the value of the counter $z$ following the last reset is in this interval $[-N-1, N+1]$. Initially this is surely true as initially the value of 
 the counter is $0$. 
 We will show how to provide this invariant on the next reset assuming it holds on the last reset.  

\begin{center}
	\begin{figure}
			\hspace{2cm}
\includegraphics[width=0.6\textwidth]{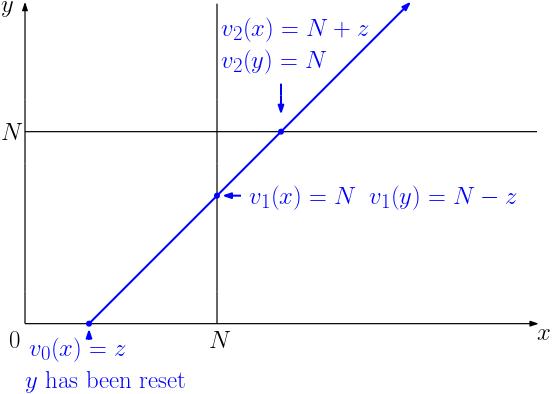}
	\caption{Curve of the clock values after a reset of clock $y$. Initially the difference $z$ between the values of $x$ and $y$ is equal to the value of $x$.}
	\label{xyz}
	\end{figure}
\end{center}

Recall the definition of regions from Subsection~\ref{preliminary construction section}. 
Let us assume a subrun $q(v)\xrightarrow{N}q'(v') \xrightarrow{N}^* q''(v'')\xrightarrow{N} q'''(v''')$ starting and ending by a reset of at least one of the two clocks $\{x,y\}$
and where $q'(v')\xrightarrow{N}^*q''(v'')$ is reset-free.
We want $\C$ to be able to check whether such a run can exist. 

\fbox{\parbox[t][1.4cm][c]{14.4cm}{
	For rest of the proof let us assume without loss of generlity that
$y$ is reset along $q(v)\xrightarrow{N}q'(v')$, where the latter configuration
can hence be written as $q'(z,0)$, as we want the counter $z$ to store the 
value of $x$.}}
\medskip

\begin{samepage}
The POCA $\C$ guesses
\begin{itemize}
\item the regions $\R_0$, $\R_1$, $\ldots$, $\R_l$ visited  and the order in which they are visited, where here by convention $\R_k$ denotes the region assumed to be the $k$-th visited region,
\item the control states $s_0, \ldots$, $ s_l$ when each region is visited for the first time,
\item the control states $q_0, \ldots$, $ q_l$ when each region is visited for the last time,
\item the control state $q''$ in which the next reset of $\B$ occurs, and
\item which clock is going to be reset next (either $x$ or $y$).
\end{itemize}
\end{samepage}
Note that there are only a finite number of regions.
Our POCA $\C$ then checks that the sequence $\R_0$, $\R_1$, $\ldots$ $\R_l$ is valid, 
retaining the counter value $z$.

First $\C$ checks that $(z,0)$ lies in $\R_0$ i.e. that $z$ is 
equal to $0$ if $\R_0 =  (0,0) $, 
strictly between $0$ and $N$ if $\R_0=(0,0) \leftrightarrow (N,0)$, 
equal to $N$ if $\R_0 = (N,0)$ and
strictly above $N$ if $\R_0=(N,0) \leftrightarrow (+\infty,0)$.
and moreover checks that the guessed regions are adjacent, and
that the regions can be visited in the guessed order.

\medskip

\noindent
Then $\C$ checks reachability within each individual region using 
Lemma~\ref{lemma arithmetic progression} as follows.
To each region $\R_k$ one can associate a set 
$\{(a_{k,j},b_{k,j}) \in \N^2 \mid j \in[1,r_k]\}$ obtained by
Lemma~\ref{lemma arithmetic progression}. 
This allows $\C$ to check, for every $k < l$, for every $v \in \R_k$, $v + t \in \R_k$, reachability of $q_k(v+t)$ from $s_k(v)$ in the region $\R_k$
by checking whether or not $t \in \bigcup_{1 \leq j \leq r} a_j + b_j \N $.
In order to check reachability inside a region $\R_k$ of the form 
$(\alpha,\beta)$ or
$(\alpha,\beta) \leftrightarrow (\gamma,\eta)$ 
for $\alpha,\beta \in \{0,N\}$,
and $\gamma,\eta \in \{0,N, +\infty\}$, it suffices to check that
$\bigcup_{1 \leq j \leq r} a_{k,j} + b_{k,j} \N $ contains $0$, as the clock values cannot both 
increment and remain inside these regions, i.e. for any such $\R_k$, for all $v \in \R_k$,
$v+t \in \R_k$ implies that $t = 0$.
Indeed, one can check easily check whether
$0 \in \bigcup_{1 \leq j \leq r} a_{k,j} + b_{k,j} \N $ by computing 
$\{(a_{k,j},b_{k,j}) \in \N^2 \mid j \in[1,r_k]\}$, which can be done
in polynomial time in $|\B|$.

Now, to check that an $N$-run exists in $\B$ in a given region $\R_k$ of the form 
$\textsc{Lower-Left} $, $\textsc{Lower-Right}$, $ \textsc{Upper-Left} $ or $\textsc{Upper-Right}$,
the automaton $\C$ furthermore distinguishes
whether the computation in the region $\R_k$ starts on the left side
	or on the bottom side,   
and whether 
	the computation in the region $\R_k$    
	 ends on the right side
	or on the top side,
 and uses the semilinearity property to check that the value added to the clocks is indeed in  
$\bigcup_{1 \leq j \leq r} a_{k,j} + b_{k,j} \N $.
Note that the first configuration of $\textsc{Lower-Left} $ is necessarily of the form 
$s_k(z+1,1)$ as $y$ has been assumed to be the last clock to be reset,
the first configuration of $\textsc{Lower-Right} $ is of the form
$s_k(z+1,1)$ or $s_k(N+1,N+1 - z)$, depending on whether it has been reached from 
the bottom or from the left corner (or possibly both),
and finally
note that $ \textsc{Upper-Left}$ cannot be reached if $y$ was the last clock to be reset.

Thus, to check reachability inside $\R_k$, our POCA $\C$ guesses an offset 
$a = a_{k,j}$ and a period $b=b_{k,j}$ among the generators of  
$\{(a_{k,j},b_{k,j}) \in \N^2 \mid j \in[1,r_k]\}$
 that it will use to reach $q_k$.
Secondly we define four gadgets in order to handle the three regions
possibly traversed, namely $\textsc{Lower-Left}$,$\textsc{Lower-Right}$, and $\textsc{Upper-Right}$.
\newline

\noindent
{\em Case 1. Checking reachability in the $\textsc{Lower-Left} $ region.}\\

Here the region is necessarily reached from bottom side as $y$ was the last clock to be reset.
Moreover, as clocks progress at the same rate, the region is necessarily exited in the 
right corner (or both in the right and upper corner).
Here
$\C$ checks that $q_{k}(N-1,N-1-z)$ is reachable from $s_k(z+1,1)$, i.e.
$\C$ checks that $(N -1) - (z +1) \in a+b\N$ 
which in turn is equivalent to checking if 
  $z +2 -N +a  =  -n \cdot b$ for some $n\in\N$.
Figure~\ref{xyz} for an illustration of the trajectories of the counter values.
In order to restore the value $z$ the POCA $\C$ does this by a carefully chosen
gadget shown in Figure~\ref{gadget reachability 1}.
Since $(z+1,1) \in \textsc{Lower-Left}$ it follows $z\in[0,N-2]$,
thus the counter value along the gadget stays inside the interval 
$[-(N-2),\max(N,a)]$.

\begin{center}
	\begin{figure}
				\hspace{1cm}
\includegraphics[width=0.8\textwidth]{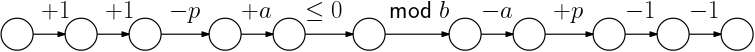}
	\caption{Gadget testing reachability for Case 1.}
	\label{gadget reachability 1}
	\end{figure}
\end{center}

\noindent
{\em Case 2. Checking reachability in the $\textsc{Lower-Right} $ region when reached from the bottom side.}\\

Here the region is necessarily exited in the top side, and
we will show how 
$\C$ can check that $q_{k}(N+z-1,N-1)$ is reachable from $s_k(z+1,1)$
and then restore $z$.
Indeed, since $y$ was the last clock that was reset,
due our convention $z\in[-(N+1),N+1]$ and by our case we must have $z+1\in\{N+1,N+2\}$,
and therefore $z\in\{N,N+1\}$.
Our POCA distinguishes the two cases $z=N$ and $z=N+1$ explicitly as follows.
To check that that $q_{k}(N+z-1,N-1)$ is reachable from $s_k(z+1,1)$
we need to test if $N-2\in a+b\N$.
Our POCA $\C$ first tests if $z$ equals $N$ or if $z$ equals $N+1$,
then does the test by a carefully chosen sequence of operations
that allow to restore the counter value $z\in\{N,N+1\}$
as can be seen in the gadget in Figure~\ref{gadget reachability 2}.
Since $(z+1,1) \in \textsc{Lower-Right}$ the counter value along the gadget 
stays inside the interval $[-(a+2),N+1]$.

\begin{center}
	\begin{figure}
				\hspace{1cm}
\includegraphics[width=0.8\textwidth]{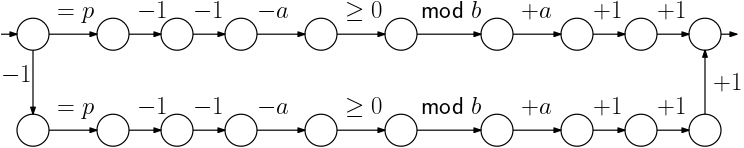}
	\caption{Gadget testing reachability for Case 2.}
	\label{gadget reachability 2}
	\end{figure}
\end{center}

\noindent
{\em Case 3. Checking reachability in the $\textsc{Lower-Right} $ region when reached from the left side.}\\

Here the region is necessarily exited in the top side, and
$\C$ checks that $q_{k}(N+z-1,N-1)$ is reachable from $s_k(N+1,N+1-z)$
, i.e.
$\C$ checks that $  (z + N-1) - (N+1) \in a + b\N$
or equivalently if $z-2\in a+b\N$.
Since $(N+1,N+1-z)\in\textsc{Lower-Right}$ it follows $z\in[0,N]$.
Again by a carefully chosen sequence of operations
that allow to restore the counter value $z\in[0,N]$
we can realize this test as seen in the gadget in Figure~\ref{gadget reachability 3}.
Since $(N+1,N+1-z) \in \textsc{Lower-Right}$ the counter value along the gadget 
stays inside the interval $[-(a+2),N]$.

\begin{center}
	\begin{figure}
			\hspace{2.8cm}
\includegraphics[width=0.6\textwidth]{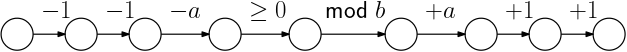}
	\caption{Gadget testing reachability for Case 3.}
	\label{gadget reachability 3}
	\end{figure}
\end{center}

No other region is reachable from $\textsc{Upper-Right}$. Moreover, if $y$ was among the last clocks to be reset, as the clocks valuations increment at the same rate, region  $ \textsc{Upper-Left} $ is not reachable. Thus the three treated above cases conclude the question of reachability inside a region.
Next, in order to test whether or not it is possible to reach $\R_{k+1}$ in state $s_{k+1}$ from $\R_k$ and state $q_k$, we check whether or not 
 in $\B$
there exists some $+1$ rule of the form $(q_k,g,\emptyset,s_{k+1})$ such that $\R_k\models g$
(and there is hence a corresponding rule in $\B_{\R_k}$). \newline

To finish the construction our POCA $\C$ needs to be able to simulate clock resets in an $N$-run in $\B$. 
The process will depend on the guessed region $\R_l$ in which the reset is assumed to occur. 
For $\R_l$ of the form 
$(\alpha,\beta)$ 
with $\alpha,\beta \in \{0,N\}$, the precise value of each clock is known: if 
$x$ is the next clock to be reset, then the new counter value should be $-v(y)$, i.e. $-\beta$, 
and if $y$ is the next
clock to be reset, then the new counter value should be $v(x)$, i.e. $\alpha$.
For $\R_l$ of the form 
$(\alpha,\beta) \leftrightarrow (\gamma,\beta)$,
with $\alpha,\beta \in \{0,N\}$, and
with $\gamma \in \{0,N, +\infty\}$,
the precise value of each clock again is known:  if $x$ is the next clock to be reset, then 
the new counter value should be $-v(y)$, i.e. $- \beta$. If $y$ is the next
clock to be reset, then the new counter value should be $v(x)$, which, when $z$ is the 
value of $x$ when $y$ was last reset, is equal to $z$ plus the value of $y$, i.e. $z+\beta$. 
If $z$ has has absolute value at most $N$, then
 $z+N$ has absolute value at most $2 \cdot N$. 
 We thus test whether or not the absolute value of $z+\beta$'s exceeds $N+1$ or not, and, 
 if it is the case, we set it to $N+1$ before performing any other operation. 

 The case when $\R_l$ is of the form 
$(\alpha,\beta) \leftrightarrow (\alpha,\delta)$
with $\alpha,\beta \in \{0,N\}$, and
with $\delta \in \{0,N, +\infty\}$
is only possible if $\alpha=\delta=N$ (we refer to Figure~\ref{xyz}) and is done as follows.
 The case when $y$ is the next clock to be reset is again easy, we
 set the new counter value to $N$.
If $x$ is the next
clock to be reset, then the new counter value should be $-v(y)$.
To do so, observe that $v(y)$, when $z$ was the 
value of $x$ when $y$ was last reset, is equal to $N-z$, thus 
the new counter value should be $-(N-z)=z-N$. 
If already $z = -(N+1)$, we do not add anything.
Since $z\in[0,N+1]$ by our case the new counter value has absolute value at
most $N$.

Observe that since we have assumed without loss of generality that
$y$ was the last clock to be reset, we cannot have a reset inside
the region $\textsc{Upper-Left}$.
Thus, it remains to simulate resets in the regions 
$\textsc{Lower-Left} $, $\textsc{Lower-Right}$, and $\textsc{Upper-Right}$.
For this observe that the precise value of each clock is not known, but it is feasible to nondeterministically guess the value of the clocks when the reset occurs, based on the region and whether it was reached from the bottom side or the left side.
This case distinction allows us to know the exact starting clock valuation $v_l$ of the $\R_l$-restricted run preceding the reset. From this, we guess an element $t$ of $\bigcup_{1 \leq j \leq r_l} a_{l,j} + b_{l,j} \N $ to increment 
the clock valuation by $t$ in such a way that $v_l + t \in \R_l$. 
We will distinguish which of the two clocks $x$ and $y$ will be reset next.
\newline

\noindent
{\em Case 1. Simulating resets in the $\textsc{Lower-Left} $ region.}\\

Let us first discuss the case when $y$ (and only $y$) is the next clock to be reset.
In this case $\C$ nondeterministically guesses a configuration
$q (z+1 + \delta,1+ \delta)$
with 
$z+1 + \delta \leq N - 1$ 
reachable from 
$s_l (z+1,1)$, i.e. $\delta \in \bigcup_{1 \leq j \leq r_l} a_{l,j} + b_{l,j} \N $. 
To do that $\C$ adds a number of the form $1+a+ b\cdot n$ for some $n\in\N$ to the counter
and checks that it is at most $N-1$, as seen in Figure~\ref{gadget reset 1}.
We remark that counter values along this gadget
stay inside $[0,N-1]$.

Let us now discuss the case when $x$ (and only $x$) is the next clock to be reset.
In this case $\C$ nondeterministically
establishes a counter value of the form $-\delta-1$ such that $-(\delta+1)\geq z-N+1$, 
where $\delta=a+b\cdot n$ for some $n\in\N$, as seen in Figure~\ref{gadget reset 1 x}
We remark that the counter values along this gadget
stay inside $[-(N-1),N-1]$.

The case when $x$ and $y$ are next to be reset simultaneously can be done analogously by setting the new counter
to $0$ and is not discussed in detail here.

\begin{center}
	\begin{figure}
				\hspace{4.5cm}
\includegraphics[width=0.3\textwidth]{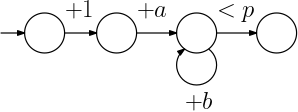}
	\caption{A gadget implementing a reset of clock $y$ in the Case 1.}
	\label{gadget reset 1}
	\end{figure}
\end{center}

\begin{center}
	\begin{figure}
				\hspace{2.8cm}
\includegraphics[width=0.6\textwidth]{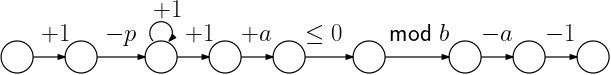}
	\caption{A gadget implementing a reset of clock $x$ in the Case 1.}
	\label{gadget reset 1 x}
	\end{figure}
\end{center}

\noindent
{\em Case 2. Simulating resets in the $\textsc{Lower-Right} $ region when reached from the left side.}\\

Let us first discuss the case when $y$ (and only $y$) is the next clock to be reset.
In this case our POCA $\C$ nondeterministically guesses a configuration $q(N +1 + \delta,N+1+\delta-z)$
with 
$N+1+ \delta -z\leq N - 1 $ 
reachable from $s_l(N+1,N+1-z)$, i.e. where $\delta  \in \bigcup_{1 \leq j \leq r_l} a_{l,j} + b_{l,j} \N $
is of the form $a + b \cdot n $ with $a,b,n \in \N$. Then $\C$ will have counter value
$N +1 + \delta > N$, and thus  
$\C$ sets the counter value to $ N+1$.
To do that, $\C$ works as seen in Figure~\ref{gadget reset 2}.
We remark that the counter values along this gadget
stay inside $[-1,2N]$. 

Let us now discuss the case when $x$ (and only $x$) is the next clock to be reset.
In this case our POCA $\C$ establishes the new counter
value $z-\delta-N-1$, realized by the gadget seen in Figure~\ref{gadget reset 2 x}.
We remark that the counter values along this gadget
stay inside $[-(N-1),N+1]$. 

The case when $x$ and $y$ are next to be reset simultaneously can be done analogously by setting the new counter
to $0$ and is not discussed in detail here.

\begin{center}
	\begin{figure}
			\hspace{0.4cm}
\includegraphics[width=0.9\textwidth]{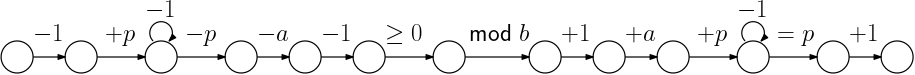}
	\caption{A gadget implementing a reset of clock $y$ in the Case 2.}
	\label{gadget reset 2}
	\end{figure}
\end{center}

\begin{center}
	\begin{figure}
			\hspace{2.1cm}
\includegraphics[width=0.7\textwidth]{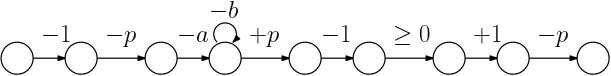}
	\caption{A gadget implementing a reset of clock $x$ in the Case 2.}
	\label{gadget reset 2 x}
	\end{figure}
\end{center}

\noindent
{\em Case 3. Simulating resets in the $\textsc{Lower-Right} $ region when reached from the bottom side.}\\

Let us first discuss the case when $y$ (and only $y$) is the next clock to be reset.
In this case our POCA $\C$ nondeterministically guesses a configuration
$s_l (z+1 + \delta,1+ \delta)$
with 
$1 + \delta \leq N-1$ 
reachable from 
$s_l (z+1,1)$.
We
 need to check that there exists  $\delta  \in \bigcup_{1 \leq j \leq r_l} a_{l,j} + b_{l,j} \N $ which moreover satisfies the inequality
$1+\delta \leq N-1$, or equivalently 
$z+1+\delta \leq N-1+z$
. Moreover, as by assumption $z \leq N+1$, and moreover $(z+1,1) \in \textsc{Lower-Right}$, 
we must have $z \in \{ N, N+1\}$.
Our POCA distinguishes the two cases $z=N$ and $z=N+1$ explicitly similarly
as checking reachability in the \textsc{Lower-Right} region when reached from
the bottom side. 
The gadget can be found in Figure~\ref{gadget reset 3}.
We remark that the counter values along this gadget
stay inside $[1,2N]$. 

Let us now discuss the case when $x$ (and only $x$) is the next clock to be reset.
The gadget can be found in Figure~\ref{gadget reset 3 x}.
We remark that the counter values along this gadget
stay inside $[-(N-1),N+1]$.\newline

The case when $x$ and $y$ are next to be reset simultaneously can be done analogously by setting the new counter
to $0$ and is not discussed in detail here.\newline

\begin{center}
	\begin{figure}
			\hspace{1cm}
\includegraphics[width=0.7\textwidth]{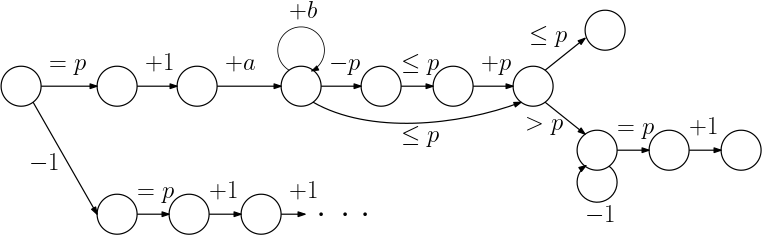}
		\caption{A gadget implementing a reset of clock $y$ in the Case 3 with details for the case 
		$z=N$. The $\cdots$ corresponds to the case $z=N+1$ and works the same way.}
	\label{gadget reset 3}
	\end{figure}
\end{center}

\begin{center}
	\begin{figure}
				\hspace{4.4cm}
\includegraphics[width=0.3\textwidth]{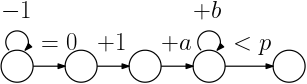}
	\caption{A gadget implementing a reset of clock $x$ in the Case 3.}
	\label{gadget reset 3 x}
	\end{figure}
\end{center}

\noindent
{\em Case 4. Simulating resets in the $\textsc{Upper-Right} $ region.}\\

Here by definition of the region the values of the clocks are above $ N +1 $ and hence again
their precise
value is not relevant, only the existence of a way to reach the configuration when the
reset occurs.
Here we precompute in our reduction whether $\bigcup_{1 \leq j \leq r_l} a_{l,j} + b_{l,j} \N $
is not empty, and then set the counter to $N+1$ (if $y$ is to next to be reset)
and to $-(N+1)$ (if $x$ is next to be reset)
and to $0$ if both are to be reset.

We notice that for each gadget implementation for testing reachability inside a region and for implementing the 
resets of clock $x$, clock $y$ or both simultaneously, the value of the counter stays
inside the interval $[ - 2 \cdot \max(a+2,N), 2 \cdot \max(a+2,N)]$, 
where $a$ is the value of the offset used in the gadget.

Checking reachability and simulating resets when $x$ was the last clock to be reset, instead of $y$, works 
again in a symmetrical way and can be dually shown to be such that the value of the counter stays inside 
the same interval.  
Testing reachability of a guessed final state inside a region works the same way as the implementation of a reset in the region, with $\C$ guessing a final control state in which the computation ends instead of a control
state in which the next reset occurs.

\begin{center}
	\begin{figure}
				\hspace{2.5cm}
\includegraphics[width=0.5\textwidth]{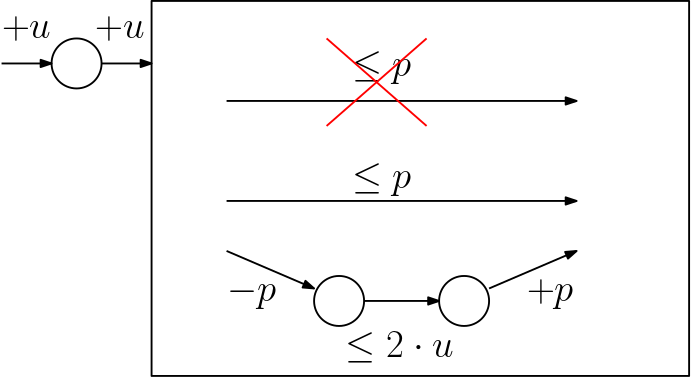}
	\caption{A gadget for adjusting a $\leq p$-test when intially offsetting the counter by $2\cdot u$.}
	\label{non-negativity}
	\end{figure}
\end{center}

Finally we show how to achieve non-negativity. 
First, our final automaton checks  whether or not the value
$N$  is greater than $(2+ c_{\max})$, where $c_{\max}$ 
is the maximal of all offsets $a_{k,j}$ and all periods $b_{k,j}$ used in any gadget.
Then, fixing $ u = \max(c_{\max}+2,N)$, we transition into a new POCA obtained
from the POCA described above (the construction where we allowed the counter to take negative
values)
 by first adding two $+u$ gadgets before entering the initial state, as seen in Figure~\ref{non-negativity}.
Furthermore, any comparison operation $\leq p$ (resp. $\leq c$) is replaced by a
gadget as seen in Figure~\ref{non-negativity},
using an appropriate adjusted gadget for $\leq (2 \cdot u)$ comparison.
Comparisons of the form $> p$, $=p$, $<p$, and
$ \leq p$ (resp. $> c$, $=c$, $<c$, and
$ \leq c$) are performed in an analogous manner.

Finally, for any modulo test, to simulate a $ \mod \ b$ rule, we have two parallel branches, 
\begin{itemize}
\item firstly a $\geq (2 \cdot u)$ comparison followed by determining the residual modulo $b$ of the current counter value, say $r_1$, using the state space (by repeatedly substractiong at most $b$ from the counter, performing $\mod \ b$, then adding the same amount as substracted), then substracting $u$, then determining the new residual modulo $b$, say $r_2$, keeping track of it using the state space too (by repeatedly substracting at most $b$ to the counter, then
	performing $\mod \ b$, and then adding the same amount as substracted), 
\item secondly a $\leq (2 \cdot u)$ comparison,
followed by a similar gadget but where instead of using a $-u$ operation, we use
a $+u$ operation and instead of substracting at most $b$, adding at most $b$.
\end{itemize}
We then compare the two residual $r_1$ and $r_2$ stored in the state space, and check whether or not
$r_1 - 2 \cdot (r_1 - r_2)$, the residual the counter value would have had without the $2 \cdot u$ offset, 
is equal to $0$ (in the state space), before restoring the counter value to the value it had before entering the gadget. 
Notice that this enforces that the value of the counters stays between 
$0$ and $4 \cdot (\max (N, (2+c_{\max}) ) $, and
by observing that $|\C|\geq 2+c_{\max}$, this enforces that the counter value stays 
between $0$ and $4 \cdot (\max (N, |\C|) ) $.

\end{document}